\documentclass[a4paper,11pt]{article}

\usepackage{geometry}										
\usepackage[english]{babel}									
\usepackage[utf8]{inputenc}									
\usepackage[T1]{fontenc}									

\usepackage{lmodern}										

\usepackage{amsmath}										
\usepackage{amstext}										
\usepackage{amsfonts}										
\usepackage{amssymb}										

\usepackage{bm}											
\usepackage{dsfont}										
\usepackage{units}										

\usepackage{textcomp}										
\usepackage{hyphenat}                                       

\usepackage{graphicx}										
\usepackage[svgnames]{xcolor}									
\usepackage{booktabs}										
\usepackage{multirow}										
\usepackage[inline]{enumitem}
\usepackage{algorithm}
\usepackage{algorithmic}
\usepackage{appendix}

\graphicspath{{./newfig/}}

\usepackage{amsthm}
\newtheorem{theorem}{Theorem}
\newtheorem{remark}{Remark}
\newtheorem*{remark*}{Remark}
\newtheorem{lemma}{Lemma}

\usepackage{mathtools}

\newcommand{\mathbbm}[1]{\text{\usefont{U}{bbm}{m}{n}#1}} 
\newcommand{\Esp}[1]{{\mathbb E}\left[ #1 \right]}

\newcommand{\Var}[1]{{\rm Var}\left[ #1 \right]}

\newcommand{\ve}[1]{\boldsymbol{#1}}
\newcommand{\acc}[1]{\left\{#1\right\}}
\newcommand{\eqdef}{\stackrel{\text{def}}{=}}

\DeclarePairedDelimiter\abs{\lvert}{\rvert}
\DeclarePairedDelimiter\norm{\lVert}{\rVert}
\makeatletter
\let\oldabs\abs
\def\abs{\@ifstar{\oldabs}{\oldabs*}}
\let\oldnorm\norm
\def\norm{\@ifstar{\oldnorm}{\oldnorm*}}
\makeatother

\newcommand{\ca}{{\mathcal A}}
\newcommand{\cc}{{\mathcal C}}
\newcommand{\cd}{{\mathcal D}}
\newcommand{\cg}{{\mathcal G}}
\newcommand{\ch}{{\mathcal H}}
\newcommand{\cl}{{\mathcal L}}
\newcommand{\cm}{{\mathcal M}}
\newcommand{\cn}{{\mathcal N}}
\newcommand{\cu}{{\mathcal U}}
\newcommand{\cx}{{\mathcal X}}
\newcommand{\cy}{{\mathcal Y}}
\newcommand{\Rr}{{\mathbb R}}
\newcommand{\Nn}{{\mathbb N}}

\newcommand{\GLD}{{\rm GLD}}

\DeclareMathOperator*{\PC}{PC}
\DeclareMathOperator{\suppt}{supp}
\DeclareMathOperator{\Betafun}{B}

\DeclareMathOperator{\AOLS}{AOLS}
\DeclareMathOperator{\OLS}{OLS}
\DeclareMathOperator{\WLS}{WLS}
\newcommand{\support}[1]{\suppt\left(#1\right)}

\newcommand{\D}{\mathrm{d}}

\providecommand{\ie}{i.e.,\;}
\providecommand{\eg}{e.g.,\;}

\newlength{\HYDROsubWidth}	
\newlength{\HYDROfigHeight}	
\newlength{\HYDROmapHeight}	
\newlength{\HYDROfigHeightNew}
\usepackage{authblk}										

\title{Emulation of stochastic simulators using generalized lambda models}
\author[1]{Xujia Zhu \thanks{zhu@ibk.baug.ethz.ch}}
\author[1]{Bruno Sudret\thanks{sudret@ethz.ch}}
\affil[1]{Chair of Risk, Safety and Uncertainty Quantification, ETH Z\"{u}rich, Stefano-Franscini-Platz 5, 8093 Z\"{u}rich, Switzerland}

\date{\today}

\usepackage{microtype}										

\usepackage{indentfirst}									

\usepackage{caption}
\usepackage{subcaption}

\usepackage[numbers,sort&compress]{natbib}							
\usepackage{hyperref}
\hypersetup{
pdftitle = {Emulation of stochastic simulators using generalized lambda models},
pdfauthor = {Xujia Zhu, Bruno Sudret},
pdfsubject = {Manuscript submitted to Reliab. Eng. Sys. Safety},
pdfkeywords = {Stochastic simulators, Surrogate modeling, Generalized lambda 
distributions, Polynomial chaos expansions},
colorlinks = false,										
}

\usepackage[capitalize]{cleveref}								
\creflabelformat{equation}{#2(#1)#3}								
\crefrangelabelformat{equation}{#3(#1)#4 to #5(#2)#6}						
\labelcrefformat{subequation}{#2(#1)#3}								
\labelcrefrangeformat{subequation}{#3(#1)#4 to #5(#2)#6}					

\begin{document}

\maketitle

\begin{abstract}
Stochastic simulators are ubiquitous in many fields of applied sciences and engineering. In the context of uncertainty quantification and optimization, a large number of simulations is usually necessary, which becomes intractable for high-fidelity models. Thus surrogate models of stochastic simulators have been intensively investigated in the last decade. In this paper, we present a novel approach to surrogating the response distribution of a stochastic simulator which uses generalized lambda distributions, whose parameters are represented by polynomial chaos expansions of the model inputs. As opposed to most existing approaches, this new method does not require replicated runs of the simulator at each point of the experimental design. We propose a new fitting procedure which combines maximum conditional likelihood estimation with (modified) feasible 
generalized least-squares. We compare our method with  state-of-the-art nonparametric kernel estimation on four different applications stemming from mathematical finance and epidemiology. Its performance is illustrated in terms of the accuracy of both the mean/variance of the stochastic simulator and the response distribution. As the proposed approach can also be used with experimental designs containing replications, we carry out a comparison on two of the examples, showing that replications do not necessarily help to get a better overall accuracy and may even worsen the results (at a fixed total number of runs of the simulator).
\end{abstract}

\section{Introduction}
\label{sec:intro}
With increasing demands on the functionality and performance of modern 
engineering systems, design and maintenance of complex products and 
structures require advanced computational models, a.k.a. simulators. They help 
assess the reliability and optimize the behavior of the 
system already at the design phase. Classical simulators are usually 
deterministic because they implement solvers for the governing equation of the 
system. Thus, repeated model evaluations with the same input parameters 
consistently result in the same value of the output quantities of interest 
(QoIs). In contrast, \emph{stochastic simulators} contain intrinsic randomness, 
which leads to the QoI being a random variable conditioned on the given set of 
input parameters. In other words, each model evaluation with the same input 
values generates a realization of the response random variable that follows an 
unknown distribution. Formally, a stochastic simulator $\cm_s$ can be expressed 
as
\begin{equation}\label{eq:defsto}
	\begin{split}
		\cm_s: \cd_{\ve{X}} \times \Omega &\rightarrow \Rr \\
		(\ve{x},\omega) &\mapsto \cm_s(\ve{x},\omega),
	\end{split}
\end{equation}
where $\ve{x}$ is the input vector that belongs to the input space 
$\cd_{\ve{X}}$, and $\Omega$ denotes the sample space of the probability space  
$\acc{\Omega,\mathcal{F},\mathbb{P}}$ that represents the internal source of 
randomness. 
\par
Stochastic simulators are widely used in modern engineering, finance, and 
medical sciences. Typical examples include evaluating the performance of a wind 
turbine under stochastic loads \cite{AbdallahPEM2019}, predicting the price 
of an option in financial markets \cite{Shreve2004}, and the spread of a 
disease in epidemiology \cite{Britton2010}.
\par
Due to the random nature of stochastic simulators, repeated model evaluations 
with the same input parameters, called hereinafter \emph{replications}, are 
necessary to 
fully characterize the probability distribution of the corresponding QoI. In 
addition, uncertainty quantification and optimization problems 
typically require model evaluations for various sets of input parameters. 
Altogether, it is necessary to have a large number of model runs, which becomes 
intractable for costly models. To alleviate the computational burden, surrogate 
models, a.k.a. emulators, can be used to replace the original model. Such a  
model emulates the input-output relation of the simulator and is 
easy and cheap to evaluate.
\par
Among several options for constructing surrogate models, this paper focuses on 
the so-called \emph{nonintrusive} approaches. More precisely, the 
computational model is considered as a ``black box'' and is only required to be 
evaluated on a limited number of input values, called the \emph{experimental 
	design} (ED). 
\par
Three classes of methods can be found in the literature for emulating 
the entire response distribution of a stochastic code in a nonintrusive 
manner. The first one is the 
\emph{random field approach}, which approximates the stochastic simulator by a 
random field. The definition in \cref{eq:defsto} implies that a stochastic 
simulator can be regarded as a random field indexed by its input variables. 
Controlling the intrinsic randomness allows one to get access to different 
trajectories of the simulator, which are deterministic functions of the 
input variables. In practice, this is achieved by fixing the \emph{random 
	seed} inside the simulator. Evaluations of the trajectories over the 
experimental design can then be extended to continuous 
trajectories, either by classical surrogate methods \cite{Jimenez2017} or 
through Karhunen--Lo\`eve expansions \cite{Azzi2019}. Since this approach 
requires the effective access to the random seed, it is only applicable to data 
generated in a specific way.
\par
Another class of methods is the \emph{replication-based approach}, which relies 
on using replications at all points of the experimental design to represent the response distribution through a suitable parametrization. The estimated distribution parameters are then treated as (noisy) outputs of 
a deterministic simulator. Then, conventional surrogate 
modeling methods, such as Gaussian processes \cite{Rasmussen2006} and 
polynomial chaos expansions (PCEs) \cite{SudretJCP2011}, can emulate these 
parameters as a function of the model input \cite{Moutoussamy2015,Browne2016}. 
Because this approach employs two separate steps, the surrogate quality depends 
on the accuracy of the distribution estimation from replicates in the 
first step \cite{zhuIJUQ2020}. Therefore, many replications are necessary, 
especially when nonparametric estimators are used for the local inference
\cite{Moutoussamy2015,Browne2016}.
\par
A third class of methods, known as the \emph{statistical approach}, 
does not require replications or controlling the random seed. If the response 
distribution belongs to the exponential family, generalized 
linear models \cite{McCullagh1989} and generalized additive models 
\cite{Hastie1990} can be efficiently applied. When the QoI for a given set of 
input parameters follows an arbitrary distribution, nonparametric estimators 
can be considered, notably kernel density estimators \cite{Fan1996,Hall2004} 
and projection estimators \cite{Efromovich2010}. However, it is well known that 
nonparametric estimators suffer from the \emph{curse of dimensionality} 
\cite{Tsybakov2009}, meaning that the necessary amount of data increases 
drastically with increasing input dimensionality.
\par
In a recent paper \cite{zhuIJUQ2020}, we proposed a novel stochastic 
emulator called the \emph{generalized lambda model} (GLaM). Such a surrogate 
model uses generalized lambda distributions (GLDs) to represent the response 
probability density function (PDF). The dependence of the distribution 
parameters on the input is modeled by PCEs. However, the 
methods developed in \cite{zhuIJUQ2020} rely on replications. 
In the present contribution, we propose a new statistical approach combining 
feasible 
generalized least-squares with maximum conditional likelihood estimations to 
get rid of the need for replications. Therefore, the proposed method is much 
more versatile in the sense that replications and seed controls are no longer necessary.
\par
The paper is organized as follows. In \Cref{sec:GLD,sec:PCE}, we briefly review 
GLDs and PCEs, which are the 
two main elements constituting the GLaM. In 
\Cref{sec:SGLaM}, we recap the GLaM framework and introduce the maximum 
conditional likelihood estimator. Then, we present the 
algorithm developed to find an appropriate starting point to optimize the 
likelihood, and to design ad hoc truncation schemes for the PCEs of distribution parameters. In \Cref{sec:examples}, we 
validate the proposed method on two analytical examples and two case studies in 
mathematical finance and epidemiology, respectively, to 
showcase its capability to tackle real problems. Finally, we summarize the main 
findings of the paper and provide an outlook for future research in 
\Cref{sec:conclusions}.

\section{Generalized lambda distributions}
\label{sec:GLD}
\subsection{Formulation}
The generalized lambda distribution (GLD) is a flexible probability distribution 
family. It is able to approximate most of the well-known parametric 
distributions \cite{Freimer1988,Karian2000}, \eg uniform, normal, Weibull, and 
Student's t distributions. The definition of a GLD relies on a parametrization of 
the \emph{quantile function} $Q(u)$, which is a nondecreasing function defined 
on $[0,1]$. In this paper, we consider the GLD of the 
Freimer--Kollia--Mudholkar--Lin family \cite{Freimer1988}, which is defined by
\begin{equation}\label{eq:FKML}
	Q(u;\ve{\lambda}) = \lambda_1 + \frac{1}{\lambda_2}\left( 
	\frac{u^{\lambda_3}-1}{\lambda_3} - 
	\frac{(1-u)^{\lambda_4}-1}{\lambda_4}\right),
\end{equation}
where $\ve{\lambda}=\acc{\lambda_l: l=1,\ldots,4}$ are the four distribution 
parameters. More precisely, $\lambda_1$ is the location parameter, $\lambda_2$ 
is the scaling parameter, and $\lambda_3$ and $\lambda_4$ are the shape 
parameters. To ensure valid quantile functions (\ie $Q$ being nondecreasing on 
$u \in [0,1]$), it is required that $\lambda_2$ be positive. Based on 
the 
quantile function, the PDF  $f_{W}(w;\ve{\lambda})$ of a random variable $W$
following a GLD can be derived as
\begin{equation}\label{eq:FKMLpdf}
	f_{W}(w;\ve{\lambda}) = 
	\frac{1}{Q^\prime(u;\ve{\lambda})} 
	=\frac{\lambda_2}{ 
		u^{\lambda_3-1} + (1-u)^{\lambda_4-1}} \mathbbm{1}_{[0,1]}(u), \text{ with } u 
	= Q^{-1}(w;\ve{\lambda}) ,
\end{equation}
where $Q^\prime(u;\ve{\lambda})$ is the derivative of $Q$ with respect 
to $u$, and $\mathbbm{1}_{[0,1]}$ is the indicator function. A closed-form 
expression of $Q^{-1}$, and therefore of $f_W$, is in general not available, 
and thus the PDF is evaluated by solving the nonlinear equation 
\cref{eq:FKMLpdf} numerically. 

\subsection{Properties}
\label{sec:GLDprop}
GLDs cover a wide range of unimodal shapes, including bell-shaped, U-shaped, S-shaped and bounded-mode distributions, which 
is determined by $\lambda_3$ and $\lambda_4$, as illustrated in 
\Cref{fig:GLD_shape} \cite{zhuIJUQ2020}. For 
instance, $\lambda_3 = \lambda_4$ produces symmetric PDFs, and $\lambda_3, 
\lambda_4<1$ leads to bell-shaped distributions. Moreover, $\lambda_3$ and 
$\lambda_4$ are closely linked to the support and the tail properties of the 
corresponding PDF. $\lambda_3>0$ implies that the PDF support is left-bounded 
and $\lambda_4>0$ corresponds to right-bounded PDFs. Conversely, the 
distribution has lower infinite support for $\lambda_3 \leq 0$ and upper 
infinite support for $\lambda_4 \leq 0$. More precisely, the support of the PDF 
denoted by $\support{f_W(w;\ve{\lambda})}= [B_l,B_u]$ is given by
\begin{equation}\label{eq:Bounds}
	B_l\left(\ve{\lambda}\right) = \begin{cases}
		-\infty, &\lambda_3 \leq 0, \\
		\lambda_1 - \frac{1}{\lambda_2 \lambda_3}, &\lambda_3 > 0,
	\end{cases} \quad
	B_u\left(\ve{\lambda}\right) = \begin{cases}
		+\infty, &\lambda_4 \leq 0, \\
		\lambda_1 + \frac{1}{\lambda_2 \lambda_4}, &\lambda_4 > 0.
	\end{cases}
\end{equation}
Importantly, for $\lambda_3<0$ ($\lambda_4<0$), the left (resp., right) 
tail decays asymptotically as a power law, and thus the GLD family can also 
provide fat-tailed distributions. Due to this power law decay, for 
$\lambda_3\leq-\frac{1}{k}$ or $\lambda_4\leq-\frac{1}{k}$, moments of order 
greater than $k$ do not exist. For $\lambda_3,\lambda_4>-0.5$, the mean and 
variance exist and are given by
\begin{equation}\label{eq:GLD_mv}
	\mu = \Esp{W} = \lambda_1 - \frac{1}{\lambda_2}\left(\frac{1}{\lambda_3+1} - 
	\frac{1}{\lambda_4+1}\right), \quad v = \Var{W}=\frac{(d_2-d_1^2)}{\lambda_2^2},
\end{equation}
where the two auxiliary variables $d_1$ and $d_2$ are defined by
\begin{equation}\label{eq:GLD_MM_aux}
	\begin{split}
		d_1 &= \frac{1}{\lambda_3}\Betafun(\lambda_3+1,1) - 
		\frac{1}{\lambda_4}\Betafun(1,\lambda_4+1), \\
		d_2 &= \frac{1}{\lambda^2_3}\Betafun(2\lambda_3+1,1) - 
		\frac{2}{\lambda_3\lambda_4}\Betafun(\lambda_3+1,\lambda_4+1) +
		\frac{1}{\lambda^2_4}\Betafun(1,2\lambda_4+1),
	\end{split}
\end{equation}
with $\Betafun$ denoting the beta function. 

\begin{figure}[!htbp]
	\centering
	\includegraphics[width=.85\linewidth, keepaspectratio]{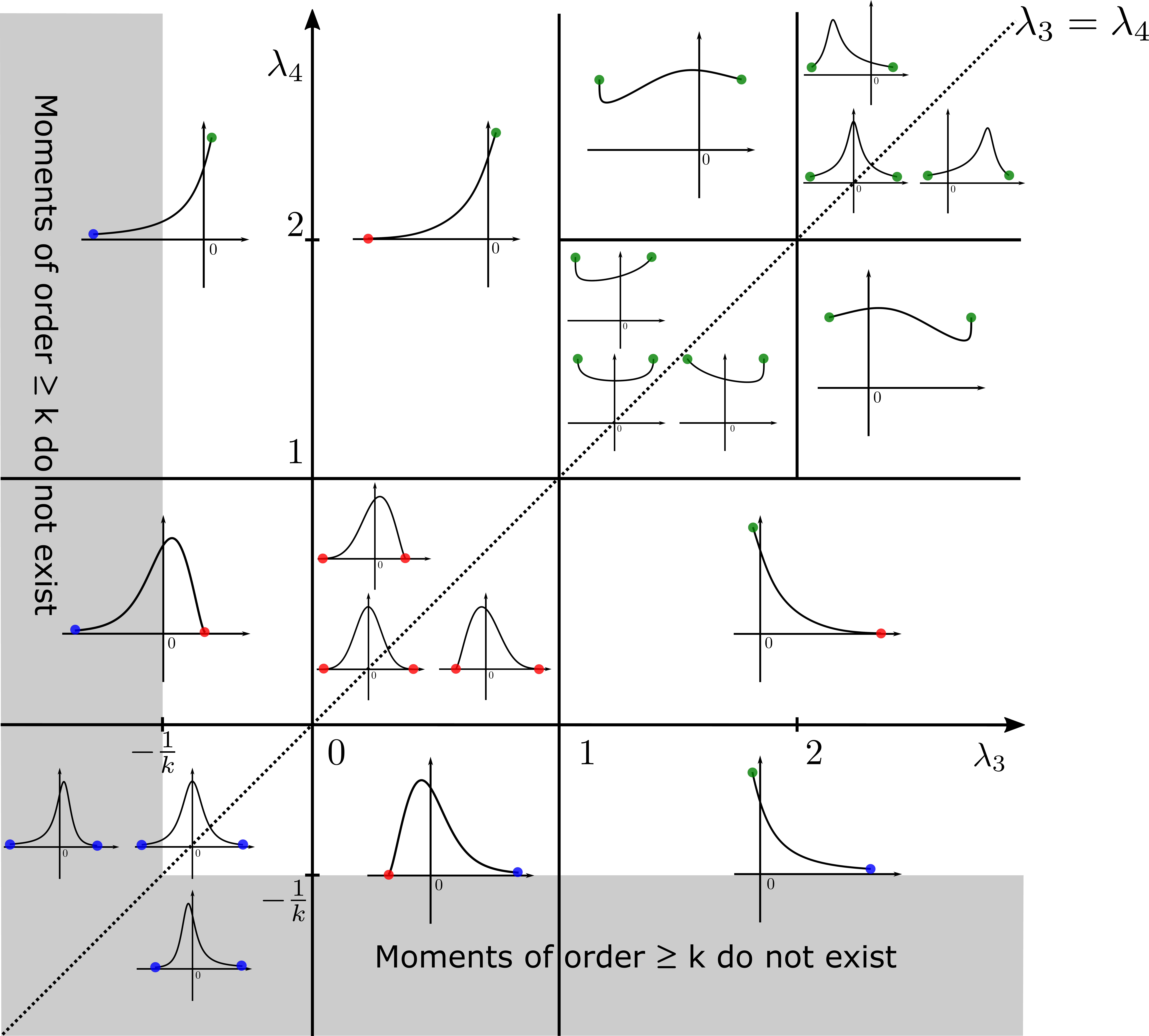}
	\caption{A graphical illustration of the PDF of	the FKML family of 
		GLD as a 
		function of $\lambda_3$ and $\lambda_4$. The values of $\lambda_1$ and 
		$\lambda_2$ are set to 0 and 1, respectively. The 
		blue points indicate that the PDF has infinite support in the marked 
		direction. In contrast, both the red and green points denote the boundary 
		points of the PDF support. More precisely, the PDF $f_W(w) = 0$ on the red 
		dots, whereas $f_W(w) = 1$ on the green ones.}
	\label{fig:GLD_shape}
\end{figure}

\section{Polynomial chaos expansions}
\label{sec:PCE}
Consider a deterministic computational model $\cm_d(\ve{x})$ that maps a set of 
input parameters $\ve{x} = \left(x_1,x_2,\ldots,x_M\right)^T \in \cd_{\ve{X}} 
\subset\Rr^M$ to the system response $z \in \Rr$. In the context of uncertainty 
quantification, the input variables are affected by uncertainty due to lack of 
knowledge or intrinsic variability (also called aleatory uncertainty). 
Therefore, they are modeled by random 
variables and grouped into a random vector $\ve{X}$ characterized by a joint 
PDF $f_{\ve{X}}$. The uncertainty in the input variables propagates through the 
the model $\cm_d$ to the output, which becomes a random variable denoted by 
$Z = \cm_d(\ve{X})$. 

\begin{remark*}
	$f_{\ve{X}}$ is the joint PDF for the input variables, which is needed to define orthogonal polynomials as described below. It should not be confused with the stochasticity of the simulator addressed in the next sections.
\end{remark*}

Provided that the output random variable $Z$ has finite variance, $\cm_d$ 
belongs to the Hilbert space $\ch$ of square-integrable functions associated 
with the inner product
\begin{equation}
	\langle u,v\rangle_{\ch} \eqdef \Esp{u(\ve{X})v(\ve{X})} = \int_{\cd_{\ve{X}}} 
	u(\ve{x})v(\ve{x})f_{\ve{X}}(\ve{x}) \D\ve{x}.
\end{equation}
If the joint PDF $f_{\ve{X}}$ fulfills certain conditions \cite{Ernst2012}, the 
space spanned by multivariate polynomials is dense in $\ch$. In other words, 
$\ch$ is a separable Hilbert space admitting a polynomial basis.
\par
In this study, we assume that $\ve{X}$ has mutually independent components, and thus the joint distribution $f_{\ve{X}}$ is expressed as
\begin{equation}
	f_{\ve{X}}(\ve{x}) = \prod_{j=1}^{M}f_{X_j}(x_j) .
\end{equation}
\par
Let $\{\phi^{(j)}_k:k\in\Nn\}$ be the orthogonal polynomial basis with 
respect to the marginal distribution of $f_{X_j}$, \ie
\begin{equation}\label{eq:unPCE_ortho}
	\Esp{\phi^{(j)}_{k}(X_j)\,\phi^{(j)}_{l}(X_j)}=\delta_{kl},
\end{equation}
with $\delta$ being the Kronecker symbol defined by $\delta_{kl}=1$ if $k=l$ 
and $\delta_{kl}=0$ otherwise. Then, the multivariate orthogonal polynomial basis 
can be obtained as the tensor product of univariate polynomials 
\cite{Soize2004}: 
\begin{equation}\label{eq:PCE_basis}
	\psi_{\ve{\alpha}}(\ve{x}) = \prod_{j=1}^{M} \phi^{(j)}_{\alpha_j}(x_j),
\end{equation}
where $\ve{\alpha} = \left(\alpha_1,\ldots,\alpha_M\right) \in \Rr^M$ denotes 
the multi-index of degrees. Each component $\alpha_j$ indicates the polynomial 
degree of $\phi_{\alpha_j}$ and thus of $\psi_{\ve{\alpha}}$ in 
the $j$th variable $x_j$. For some classical distributions, \eg normal, 
uniform, exponential, the associated univariate orthogonal polynomials are 
well known as Hermite, Legendre, and Laguerre polynomials \cite{Xiu2002}. For 
arbitrary marginal distributions, such a basis can 
be computed numerically through the \emph{Stieltjes procedure} 
\cite{Gautschi2004}.
\par
Following the construction defined in \cref{eq:PCE_basis}, 
$\acc{\psi_{\ve{\alpha}}(\cdot),\ve{\alpha} \in \Nn^M}$ forms an orthogonal 
basis for $\ch$. Thus, the random output $Z$ can be represented by
\begin{equation}\label{eq:PCE}
	Z=\cm_d(\ve{X}) = \sum_{\ve{\alpha}\in 
		\Nn^M}c_{\ve{\alpha}}\psi_{\ve{\alpha}}(\ve{X}),
\end{equation}
where $c_{\ve{\alpha}}$ is the coefficient associated with the basis function 
$\psi_{\ve{\alpha}}$. The spectral representation in \cref{eq:PCE} is a 
series with infinitely many terms. In practice, it is necessary to adopt 
truncation 
schemes to \emph{approximate} $\cm_d(\ve{x})$ with a finite series defined by a 
finite 
subset $\ca\subset\Nn^M$ of multi-indices. A typical scheme is the hyperbolic 
($q$-norm) truncation scheme \cite{BlatmanPEM2010}:
\begin{equation}\label{eq:qnorm}
	\ca^{p,q,M} = \acc{\ve{\alpha}\in \Nn^M, \|\ve{\alpha}\|_{q} = 
		\left(\sum_{i=1}^{M}\abs{\alpha_i}^q\right)^{\frac{1}{q}}\leq p},
\end{equation}
where $p$ is the maximum total degree of polynomials, and $q \leq 1$ defines 
the quasi-norm $\norm{\cdot}_q$. Note that with $q=1$, we obtain 
the so-called full basis of total degree less than $p$.

For an arbitrary distribution $f_{\ve{X}}$ with dependent components of $\ve{X}$, the usual practice is to transform $\ve{X}$ into an auxiliary vector $\ve{\xi}$ with independent components (\eg a standard normal vector) using the Nataf or Rosenblatt transform \cite{Torre2019}. Alternatively, polynomials orthogonal to the joint distribution may be computed on the fly using a numerical Gram--Schmidt orthogonalization \cite{Jakeman2019}.

\section{Generalized lambda models (GLaM)}
\label{sec:SGLaM}

\subsection{Introduction}
\label{sec:GLaM}
Because of their flexibility, we assume that the response random variable 
of a stochastic simulator for a given input vector $\ve{x}$ follows a 
GLD. Hence, the 
distribution parameters $\ve{\lambda}$ are functions of the input 
variables:
\begin{equation} \label{eq:condGLD}
	Y (\ve{x}) \sim \GLD\left(\lambda_1(\ve{x}),\lambda_2(\ve{x}),\lambda_3(\ve{x}),\lambda_4(\ve{x})\right).
\end{equation}
\par
Under appropriate conditions discussed in \Cref{sec:PCE}, each component of 
$\ve{\lambda}(\ve{x})$ admits a spectral representation in terms of orthogonal 
polynomials. Recall that $\lambda_2(\ve{x})$ is required to be positive (see 
\Cref{sec:GLD}). Thus, we choose to build the associated PCE on the 
natural logarithm transform $\log\left(\lambda_2(\ve{x})\right)$. This results 
in the following approximations:
\begin{align}
	\lambda_l\left(\ve{x}\right) &\approx \lambda^{\PC}_l\left(\ve{x};\ve{c}\right) 
	= \sum_{\ve{\alpha}\in \ca_l} c_{l,\ve{\alpha}}\psi_{\ve{\alpha}}(\ve{x}), 
	\quad l = 1,3,4, \label{eq:lamPCE}\\
	\lambda_2\left(\ve{x}\right) &\approx \lambda^{\PC}_2\left(\ve{x};\ve{c}\right) = \exp \left(\sum_{\ve{\alpha}\in \ca_2} c_{2,\ve{\alpha}}\psi_{\ve{\alpha}}(\ve{x}) \right), \label{eq:lamPCE_lam2}
\end{align}
where $\ve{\ca}=\acc{\ca_l: l=1,\ldots,4}$ are the truncation sets defining the 
basis functions, and $\ve{c}=\acc{c_{l,\ve{\alpha}}: 
	l=1,\ldots,4, \,\ve{\alpha} \in \ca_l}$ are coefficients associated to the 
bases. For the purpose of clarity, we explicitly express $\ve{c}$ in the 
spectral approximations as in $\ve{\lambda}^{\PC}\left(\ve{x};\ve{c}\right)$ to 
emphasize that $\ve{c}$ are the model parameters. 

The generalized lambda model presented above is a statistical model. It 
involves two approximations. First, the response distribution of a stochastic 
simulator is approximated by GLDs. As illustrated 
in \Cref{fig:GLD_shape}, GLDs cover a wide range of unimodal shapes but 
cannot produce multimodal distributions. Thus, the GLD representation is 
appropriate when the response distribution stays unimodal. In this case, the 
flexibility of GLDs allows capturing the possible shape 
variation of the response distribution within a single parametric family. 
Second, the distribution parameters $\ve{\lambda}(\ve{x})$ seen as functions of $\ve{x}$ are represented by truncated polynomial chaos expansions. So they must belong to the Hilbert space of square-integrable functions with respect to $f_{\ve{X}}(\ve{x}) \, \D\ve{x}$.

\subsection{Estimation of the model parameters}
\label{sec:estimation}
Given the truncation sets $\ve{\ca}$, the coefficients $\ve{c}$ need to be 
estimated from data to build the surrogate model. In this paper, as opposed to
\cite{zhuIJUQ2020} and the vast majority of the literature on 
stochastic simulators, 
the simulator is required to be evaluated \emph{only once} on the experimental 
design $\cx =\acc{\ve{x}^{(1)},\ldots,\ve{x}^{(N)}}$, and the associated model 
responses are collected in $\cy = \acc{y^{(1)},\ldots,y^{(N)}}$. To develop 
surrogate models in a nonintrusive manner, we propose using the maximum 
conditional likelihood estimator:
\begin{equation}\label{eq:joint}
	\hat{\ve{c}} = \arg\max_{\ve{c} \in \cc} \,
	\mathsf{L}\left(\ve{c}\right),
\end{equation}
where 
\begin{equation}\label{eq:nloglh}
	\mathsf{L}\left(\ve{c}\right) =  \sum_{i=1}^{N}\log\left( 
	f^{\GLD}\left(y^{(i)} ; 
	\ve{\lambda}^{\PC}\left(\ve{x}^{(i)};\ve{c}\right)\right)\right). 
\end{equation}
Here, $f^{\GLD}$ denotes the PDF of the GLD defined 
in \cref{eq:FKMLpdf}, and $\cc$ is the search space for $\ve{c}$. The estimator 
introduced in \cref{eq:nloglh} can be derived from minimizing the 
Kullback--Leibler divergence between the surrogate PDF and the 
underlying true response PDF over $\cd_{\ve{X}}$; see details in 
\cite{zhuIJUQ2020}. The advantages of this estimation method are 
twofold. On 
the one hand, it removes the need for replications in the experimental design. 
On the other hand, if a GLaM for a certain choice of $\ve{c}$ can exactly 
represent the stochastic simulator, the proposed estimator is \emph{consistent} 
under mild conditions, as shown in \Cref{thm:consistMLE} (see 
\Cref{sec:MLEConsistency} for a detailed proof).
\begin{theorem}\label{thm:consistMLE}
	Let $(\ve{X}^{(1)},Y^{(1)}),\ldots, (\ve{X}^{(N)},Y^{(N)})$ be independent and identically distributed random variables following $\ve{X} \sim P_{\ve{X}}$ and $Y(\ve{x}) \sim \GLD\left(\ve{\lambda}^{\PC}(\ve{x};\ve{c}_0)\right)$. If the following conditions are fulfilled, the estimator defined in \cref{eq:joint} is consistent, that is,
	\begin{equation}
		\hat{\ve{c}}\xrightarrow{\text{a.s.}}\ve{c}_0.
	\end{equation}
	\begin{enumerate}[label=(\roman*)]
		\item $P_{\ve{X}}$ is absolutely continuous with respect to the 
		Lebesgue measure of $\Rr^M$, \ie the joint PDF $f_{X}(\ve{x})$ is Lebesgue-measurable.
		\item $f_{X}$ has a compact support $\cd_{\ve{X}}$.
		\item $\cc$ is compact, and $\ve{c}_0 \in \cc$.
		\item There exists a set $A \subset \cd_{\ve{X}}$ with 
		$P_{\ve{X}}\left(\ve{X} 
		\in A\right)>0$ such that $\forall \ve{x} \in A $, 
		$Y(\ve{x})$ does not follow a uniform distribution.
	\end{enumerate}
\end{theorem}

Most of the assumptions in the \Cref{thm:consistMLE} are realistic, 
except the one that the true model can be exactly represented by a GLaM, which 
is rather technical to guarantee the consistency. In practice, we do not 
require the QoI for any input parameters following a GLD but assume that the 
response distribution can be well approximated by GLDs.

It is worth remarking that since a GLD can have 
very fat tails (see \Cref{sec:GLDprop}), solving the optimization problem may 
produce response PDFs with unexpected infinite moments when the model is 
trained on a small data set. To prevent too-fat tails (if no prior knowledge 
suggests it), we apply the threshold $\lambda^{\PC}_3(\ve{x}) = 
\max\acc{\lambda^{\PC}_3(\ve{x};\hat{\ve{c}}),-0.3}$ 
and $\lambda^{\PC}_4(\ve{x}) = 
\max\acc{\lambda^{\PC}_4(\ve{x};\hat{\ve{c}}),-0.3}$, which indicates that we 
enforce the surrogate PDFs to have finite moments up to order $3$ (higher 
order moments may exist depending on $\hat{\ve{c}}$). Thresholds larger than 
$-0.3$ (\eg from $-0.1$ to $0$) can be used if the response PDF is known to be 
light-tailed. Note that when enough data are available, these 
operations are unnecessary because the resulting model does not exceed the 
threshold. Although the thresholdings could have been imposed in the model 
definition in \cref{eq:lamPCE}, they change the regularity of the optimization 
problem, and do not generally improve the performance according to our 
experience. Therefore, we only use them for postprocessing.

\begin{remark}
	\label{rem:rep}
	While we consider the simulator to be evaluated only
	once for each point of the experimental design in this paper, the estimator 
	defined in \cref{eq:joint} is not limited to this type of data. When 
	replications are available, the objective function can be reformulated to
	\begin{equation}\label{eq:nloglhR}
		\mathsf{L}\left(\ve{c}\right) = \sum_{i=1}^{N} 
		\frac{1}{R^{(i)}}\sum_{r=1}^{R^{(i)}}\log\left( 
		f^{\GLD}\left(y^{(i,r)} ; 
		\ve{\lambda}^{\PC}\left(\ve{x}^{(i)};\ve{c}\right)\right)\right),
	\end{equation}
	where $R^{(i)}$ denotes the number of replications at point $\ve{x}^{(i)}$, and 
	$y^{(i,r)}$ is the model response for $\ve{x}^{(i)}$ at the $r$th replication. 
	In addition, if $R^{(i)}$ is constant for all points $\ve{x}^{(i)} \in \cx$, 
	\cref{eq:nloglhR} provides the same estimator as in our previous work 
	\cite{zhuIJUQ2020}.
\end{remark}

\subsection{Fitting procedure}
\label{sec:local}
In practice, the evaluation of $\mathsf{L}(\ve{c})$ is not straightforward 
because the PDF of GLDs does not have an explicit 
form as shown in \cref{eq:FKMLpdf}. Details about the evaluation procedure are 
given in  \cite{zhuIJUQ2020}. Note that the optimization 
problem \cref{eq:joint} is subject to complex inequality constraints due 
to the dependence of the PDF support on $\ve{\lambda}$ (see
\cref{eq:Bounds}). Given a starting point, we follow the optimization strategy 
developed in \cite{zhuIJUQ2020}: We first apply the 
derivative-based 
\emph{trust-region} optimization algorithm \cite{Steihaug1983} without 
constraints. If none of the inequality constraints is activated at the optimum, 
we keep the results as the final estimates. Otherwise, the 
constrained (1+1)-CMA-ES algorithm \cite{Arnold2012} available in the software 
UQLab \cite{UQdoc_13_201} is used instead. 
\par
Because $\mathsf{L}(\ve{c})$ is highly nonlinear, a good starting point is 
necessary to guarantee the convergence of the optimization algorithm. In this 
section, we introduce a robust method to find a suitable starting point.
\par
According to \cref{eq:GLD_mv}, the mean $\mu(\ve{x})$ and the variance function 
$v(\ve{x})$ of a GLaM satisfy 
\begin{equation}\label{eq:GLaM_mv}
	\begin{split}
		\mu(\ve{x}) &= \lambda^{\PC}_1(\ve{x}) + 
		\frac{1}{\lambda^{\PC}_2(\ve{x})}g\left( 
		\lambda^{\PC}_3(\ve{x}), \lambda^{\PC}_4(\ve{x}) \right), \\
		\log\left(v(\ve{x})\right) &= -2\log\left(\lambda^{\PC}_2(\ve{x})\right) + 
		h\left( \lambda^{\PC}_3(\ve{x}), \lambda^{\PC}_4(\ve{x}) \right), 
	\end{split}
\end{equation}
where we group the dependence of $\mu$ and $\log(v)$ on $\lambda_3$ and 
$\lambda_4$ into $g$ and $h$, respectively, for the purpose of simplicity. If 
$\lambda^{\PC}_3(\ve{x})$ and $\lambda^{\PC}_4(\ve{x})$ do not vary strongly on 
$\cd_{\ve{X}}$, we observe that the variations of the mean and the 
variance function are mostly dominated by the location parameter 
$\lambda^{\PC}_1(\ve{x})$ and the scale parameter $\lambda^{\PC}_2(\ve{x})$. 

Recall that the spectral approximation for $\lambda_2(\ve{x})$ is on its 
logarithmic transform. If a PCE can be constructed for $\mu(\ve{x})$ and 
$-\frac{1}{2}\log\left(v(\ve{x})\right)$, the associated coefficients can be 
used as a preliminary guess for the coefficients of $\lambda^{\PC}_1(\ve{x})$ 
and $\lambda^{\PC}_2(\ve{x})$, respectively. As a result, we first focus on 
estimating the mean and the variance function as follows:
\begin{equation*}
	\mu(\ve{x}) = \sum_{\ve{\alpha}\in \ca_{\mu}} 
	c_{\mu,\ve{\alpha}}\psi_{\ve{\alpha}}(\ve{x}), \quad
	v(\ve{x}) = \exp \left(\sum_{\ve{\alpha}\in \ca_v} 
	c_{v,\ve{\alpha}}\psi_{\ve{\alpha}}(\ve{x}) \right),
\end{equation*}
where the form of the variance function implies a multiplicative 
\emph{heteroskedastic} effect (see \cite{Harvey1976}).
\par
The mean estimation is a classical regression problem. However, since the 
variance function is also unknown and needs to be estimated, the 
heteroskedastic effect should be taken into account. Many methods have 
been developed in statistics and applied science to tackle heteroskedastic 
regression problems. They can be classified into two groups: one class 
of methods relies on repeated measurements at given input values 
\cite{Sadler1985,Ankenman2009,Murcia2018} (replication-based), whereas a second 
class of methods jointly estimates both quantities by optimizing certain 
functions without the need for replications 
\cite{Nelder1987,Davidian1987,Goldberg1997,Marrel2012}. Some studies 
\cite{Davidian1987,Marrel2012} have shown higher efficiency of the second class 
of methods over the former. This finding supports our pursuit for a 
replication-free approach. In particular, we opt for feasible generalized least-squares (FGLS) \cite{Wooldridge2013}, which iteratively fits the mean and 
variance functions in an alternative way. 
\par
The details are described in \cref{alg:FGLS}. In this 
algorithm, $\OLS$ denotes the use of ordinary least-squares, and 
$\WLS$ is weighted least-squares. $\hat{\ve{v}}$ corresponds to the set of 
estimated variances on the design points in $\cx$ which are then used as 
weights in $\WLS$ to re-estimate $\ve{c}_{\mu}$.
\par
\begin{algorithm}[H]
	\caption{Feasible generalized least-squares (FGLS)}
	\label{alg:FGLS}
	\begin{algorithmic}[1]
		\STATE $\hat{\ve{c}}_{\mu} \gets \OLS\left(\cx,\cy\right)$
		\FOR{$i \gets 1,\ldots,N_{\rm FGLS}$}
		\STATE $\hat{\ve{\mu}} \gets  \sum_{\ve{\alpha} \in 
			\ca_{\mu}}c_{\mu,\ve{\alpha}}\psi_{\ve{\alpha}}(\cx)$
		\STATE $\tilde{\ve{r}} \gets 2\log\left(\abs{\cy - 
			\hat{\ve{\mu}}}\right)$
		\STATE $\hat{\ve{c}}_{v} \gets \OLS\left(\cx,\tilde{\ve{r}}\right)$
		\STATE $\hat{\ve{v}}=\exp\left(\sum_{\ve{\alpha} \in 
			\ca_{v}}c_{v,\ve{\alpha}}\psi_{\ve{\alpha}}(\cx)\right)$
		\STATE $\hat{\ve{c}}_{\mu} \gets 
		\WLS\left(\cx,\cy,\hat{\ve{v}}\right)$
		\ENDFOR
		\STATE Output: $\hat{\ve{c}}_{\mu}$, $\hat{\ve{c}}_{v}$
	\end{algorithmic}
\end{algorithm}
\par
After obtaining $\hat{\ve{c}}_{\mu}$ and $\hat{\ve{c}}_{v}$ from FGLS, we 
perform two rounds of the optimization procedure described 
at the beginning of this section to build the GLaM surrogate. First, we set 
the starting points as $\ve{c}_1 = \ve{c}_{\mu}$, $\ve{c}_2 = 
-\frac{1}{2}\ve{c}_{v}$, and 
$\lambda^{\PC}_3(\ve{x})=\lambda^{\PC}_4(\ve{x})=0.13$, which corresponds to a 
normal-like shape. Then, we fit a GLaM with $\lambda^{\PC}_3(\ve{x})$ 
$\lambda^{\PC}_4(\ve{x})$ being only constant; \ie the coefficients of 
nonconstant basis functions are kept as zeros during the fitting. Finally, we 
use the resulting estimates as a starting point and construct a final GLaM with 
all the considered basis functions by solving \cref{eq:nloglh}.

\subsection{Truncation schemes}
\label{sec:adapt}
Provided that the bases of $\ve{\lambda}^{\PC}(\ve{x})$ are 
given, we have presented a procedure to construct GLaMs from data in the 
previous section. However, there is generally no prior knowledge that would 
help select the truncation sets $\ca_l$'s ab initio. In this section, we develop a 
method to determine a suitable hyperbolic truncation scheme $\ca^{p,q,M}$ 
presented in \cref{eq:qnorm} for each component of  
$\ve{\lambda}^{\PC}(\ve{x})$. 
\par
As discussed in \Cref{sec:GLD}, $\lambda^{\PC}_3(\ve{x})$ and 
$\lambda^{\PC}_4(\ve{x})$ control the shape variations of the response PDF. We 
assume that the shape does not vary in a strongly nonlinear way. Hence, the 
associated $p$ can be set to a small value, e.g., $p=1$, in practice. In contrast, 
$\lambda^{\PC}_1(\ve{x})$ and $\lambda^{\PC}_2(\ve{x})$ require possibly larger 
degree $p$ since their behavior is associated with the mean and the variance 
function, which might vary nonlinearly over $\cd_{\ve{X}}$. To this 
end, we modify \Cref{alg:FGLS} to adaptively find appropriate truncation 
schemes for $\mu(\ve{x})$ and $v(\ve{x})$, which are then used for 
$\lambda_1(\ve{x})$ and $\lambda_2(\ve{x})$, respectively.
\par
\begin{algorithm}
	\caption{Modified feasible generalized least-squares}
	\label{alg:MFGLS}
	\begin{algorithmic}[1]
		\STATE Input: $\left(\cx,\cy\right)$, $\ve{p}_1$, $\ve{q}_1$, 
		$\ve{p}_2$, $\ve{q}_2$
		\STATE $\ca_{\mu},\,\, \hat{\ve{c}}_{\mu} \gets 
		\AOLS\left(\cx,\cy,\ve{p}_1,\ve{q}_1\right)$
		\FOR{$i \gets 1,\ldots,N_{\rm FGLS}$}
		\STATE $\hat{\ve{\mu}} \gets \sum_{\ve{\alpha} \in 
			\ca_{\mu}}c_{m,\ve{\alpha}}\psi_{\ve{\alpha}}(\cx)$
		\STATE $\tilde{\ve{r}} \gets 2\log\left(\abs{\cy - 
			\hat{\ve{\mu}}}\right)$
		\STATE $\ca^i_{v},\,\,\hat{\ve{c}}^i_{v},\,\, \varepsilon^i_{\rm 
			LOO}\gets 
		\AOLS\left(\cx,\tilde{\ve{r}},\ve{p}_2,\ve{q}_2\right)$
		\STATE $\hat{\ve{v}} \gets \exp\left(\sum_{\ve{\alpha} \in 
			\ca_{v}}c_{v,\ve{\alpha}}\psi_{\ve{\alpha}}(\cx)\right)$
		\STATE $\hat{\ve{c}}_{\mu} \gets 
		\WLS\left(\cx,\cy,\ca_{\mu},\hat{\ve{v}}\right)$
		\ENDFOR
		\STATE $ i^* = \arg\min \acc{\varepsilon^i_{\rm LOO} : 
			i=1,\ldots,N_{\rm FGLS}}$
		\STATE Output: $\ca_{\mu}$, $\hat{\ve{c}}^{i^*}_{\mu}$, 
		$\ca^{i^*}_{v}$, 
		$\hat{\ve{c}}^{i^*}_{v}$
	\end{algorithmic}
\end{algorithm}
\par
\Cref{alg:MFGLS} presents the modified FGLS. Instead of using OLS, we apply the 
\emph{adaptive ordinary least-squares} with degree and $q$-norm adaptivity 
(referred to as $\AOLS$) \cite{UQdoc_13_104}. This algorithm builds 
a series of PCEs, each of which is obtained by applying OLS with the truncation 
set $\ca^{p,q,M}$ defined by a particular combination of $p \in \ve{p}$ 
and $q \in \ve{q}$. Then, it selects the truncation scheme for which the 
associated PCE has the lowest \emph{leave-one-out} error. In the modified FGLS, 
the truncation set $\ca_{\mu}$ for $\mu(\ve{x})$ is 
selected only once (before the loop), whereas several truncation schemes 
$\acc{\ca^i_{v}: i = 1,\ldots,N_{\rm FGLS}}$ are obtained. We 
select the one corresponding to the smallest leave-one-out error on the 
expansion of the variance as the truncation set $\ca_{v}$ for $v(\ve{x})$. 
After running \Cref{alg:MFGLS}, we apply the two-round optimization strategy 
described in the previous section to build the GLaM corresponding to the 
selected truncation schemes.
\par
There are several parameters to be determined in \Cref{alg:MFGLS}. 
In the following examples and applications, we set the candidate degrees 
$\ve{p}_1 = \acc{0,\ldots,10}$ for $\lambda^{\PC}_1(\ve{x})$, and $\ve{p}_2 = 
\acc{0,\ldots,5}$ for $\lambda^{\PC}_2(\ve{x})$. $\ve{p}_1$ contains 
high degrees to approximate possibly highly nonlinear mean functions, the 
accuracy of which is crucial for basis selections for $\lambda_2(\ve{x})$ in 
\Cref{alg:MFGLS}. $\ve{p}_2$ is set to have degrees up to 5, allowing 
relatively complex variations. The lists of $q$-norms 
are $\ve{q}_1 =\ve{q}_2 = \acc{0.4,0.5,0.6,0.7,0.8,0.9,1}$, which 
contains the full basis. The total number of FGLS iterations is set to 
$N_{\rm FGLS} = 10$ which, according to our experience, is enough to 
find an appropriate truncated set for $\lambda^{\PC}_2(\ve{x})$.

\section{Application examples}
\label{sec:examples}
In this section, we validate the proposed algorithm on two analytical examples and two case studies in mathematical finance and epidemiology. In the four cases, the response distributions do not belong to a single parametric family, so as to test the flexibility of the proposed method. In addition, we compare the performance of GLaMs with the nonparametric kernel conditional density estimator from the package \texttt{np} \cite{Hayfield2008} implemented in R. The latter performs a thorough leave-one-out cross-validation with a multistart strategy to choose the bandwidths \cite{Hall2004}, which is one of the state-of-the-art kernel estimation methods. The surrogate model built by this method is referred to as the kernel conditional density estimator (KCDE). 

Alongside GLaM and KCDE, another surrogate model, the heteroskedastic Gaussian process (denoted by GP), is also considered. This model assumes that the response distribution is Gaussian, and the mean and variance functions are represented  by Gaussian processes. We apply the method proposed by Binois et. al. \cite{Binois2019} which adopts a sequential design strategy to actively balance the trade-off between replications and explorations. The algorithm is available in the package \texttt{hetGP} in R. However, due to the sequential design (the new points are added one by one), building such a surrogate can be very time-consuming (cf. \Cref{sec:ex2} for details). Consequently, we present the comparisons with \texttt{hetGP}  only for the first two examples.

Moreover, for comparison purposes, we consider another ``Gaussian'' surrogate model where we represent the response distribution with a normal distribution. The associated mean and variance, which are functions of the input $\ve{x}$, are not fitted to data but set to the \emph{true} values of the simulator. In other words, this surrogate model should represent the ``oracle'' of Gaussian-type mean-variance surrogate models, such as the ones presented in \cite{Marrel2012,Binois2018}.
\par
We use Latin hypercube sampling \cite{McKay1979} to generate the experimental design for GLaM and KCDE. The stochastic simulator is only evaluated once for each vector of input parameters. The associated QoI values are used to construct 
surrogate models with the proposed estimation procedure in \Cref{sec:local}. In contrast, the construction of the GP relies on a sequential design strategy which adaptively find new points to evaluate \cite{Binois2019}. Hence, we use Latin hypercube sampling of $20\%$ of the total number of model runs to initiate the process. Then, the algorithm proceeds by iteratively looking for points to evaluate and updating the surrogate.
\par
To quantitatively assess the performance of the surrogate model, we define an 
error measure between the underlying model and the emulator by 
\begin{equation}\label{eq:Rlevel1} 
	\varepsilon = \Esp{d\left(Y(\ve{X}),\hat{Y}(\ve{X})\right)},
\end{equation}
where $Y(\ve{X})$ is the model response, $\hat{Y}(\ve{X})$ corresponds to
that of the surrogate, $d\left(Y_1,Y_2\right)$ denotes the contrast 
measure between the probability distributions of $Y_1$ and $Y_2$, and 
the expectation is taken with respect to $\ve{X}$.
In this study, we use the \emph{normalized Wasserstein distance}, defined by
\begin{equation}
	d\left(Y_1,Y_2\right) = \frac{d_{\rm 
			WS}\left(Y_1,Y_2\right)}{\sigma\left(Y_1\right)},
\end{equation}
where $d_{\rm WS}$ is the \emph{Wasserstein distance of order two} 
\cite{Villani2008} defined by
\begin{equation}\label{eq:WS}
	d_{\rm WS}\left(Y_1,Y_2\right) \eqdef \norm{Q_1 - Q_2}_2 =  
	\sqrt{\int_{0}^{1}\left(Q_1(u) - Q_2(u)\right)^2\D u}\, ,
\end{equation}
where $Q_1$ and $Q_2$ are the quantile functions of $Y_1$ and $Y_2$, 
respectively. As a summary, by combining \cref{eq:Rlevel1} and \cref{eq:WS} the global error reads
	\begin{equation}\label{eq:5.4}
		\varepsilon = \int_{\cd_{\ve{X}}} \sqrt{\int_0^1 \left(Q_{Y(\ve{x})}(u) - Q_{\hat{Y}(\ve{x})} (u) \right)^2 \, \D u} \, \frac{f_{\ve{X}} (\ve{x})}{\sqrt{\Var{Y(\ve{x})}}} \, \D\ve{x}
	\end{equation}

Following this definition, the standard deviation $\sigma_{Y_1}$ can be seen as 
the Wasserstein distance between the distribution of $Y_1$ and a degenerate 
distribution concentrated at the mean value $\mu_{Y_1}$. As a result, the 
Wasserstein distance normalized by the standard deviation can be interpreted as 
the ratio of the error related to emulating the distribution of $Y_1$ by that 
of $Y_2$, and to using the mean value $\mu_{Y_1}$ as a proxy of $Y_1$.
\par
Because $d_{\rm WS}$ is invariant under translation, the normalized 
Wasserstein distance is invariant under both translation and scaling; that is,
\begin{equation}
	\forall a \in \Rr\setminus {0}, b \in \Rr \quad
	\frac{d_{\text{WS}}\left(a\,Y_1+b,a\,Y_2+b\right)}{\sigma(a\,Y_1+b)} = 
	\frac{d_{\text{WS}}\left(Y_1,Y_2\right)}{\sigma(Y_1)}.
\end{equation}
\par
To calculate the expectation in \cref{eq:Rlevel1}, we use Latin hypercube 
sampling to generate a test set $\cx_{\rm test}$ of size $N_{\rm test} = 
1{,}000$ in the input space. The normalized Wasserstein distance is calculated 
for each $\ve{x} \in \cx_{\rm test}$ and then averaged by $N_{\rm test}$. 
\par
For the last two case studies, the analytical response distribution of $Y(\ve{x})$ is unknown. To characterize it, we repeatedly evaluate the model $10^4$ times for $\ve{x}$. In addition, we also compare some summarizing statistical quantity $b(\ve{x})$ of the model response $Y(\ve{x})$, such as the mean $\Esp{Y(\ve{x})}$ or variance $\Var{Y(\ve{x})}$, depending on the focus of the application. Note that $b(\ve{x})$ is a deterministic function of input variables, and we define the normalized mean-squared error by \begin{equation}\label{eq:Berror}
	\varepsilon_b= \frac{\sum_{i=1}^{N_{\rm test}} \left(b^{(i)}_{S} - \hat{b}^{(i)}\right)^2}{\sum_{i=1}^{N_{\rm test}}\left(\hat{b}^{(i)} - \bar{\hat{b}}\right)^2 }, \text{ with }
	\bar{\hat{b}} = \frac{1}{N_{\rm test}} \sum_{i=1}^{N_{\rm test}} \hat{b}^{(i)},
\end{equation}
where $b^{(i)}_{S}$ is the value predicted by the surrogate for $\ve{x}^{(i)} \in \cx_{\rm test}$, and $\hat{b}^{(i)}$ denotes the quantity estimated from $10^4$ replicated runs of the original stochastic simulator for $\ve{x}^{(i)}$. The error $\varepsilon_b$ defined in \cref{eq:Berror} indicates how much of the variance of $b(\ve{X})$ cannot be explained by $b_{S}(\ve{X})$ estimated from surrogate model.
\par
Experimental designs of various size $N \in 
\acc{250;500;1{,}000;2{,}000;4{,}000}$ are investigated to study the convergence of 
the proposed method. Each scenario is run 50 times with independent experimental designs to account for statistical uncertainty in the 
random design for GLaM and KCDE. For GP, $N$ corresponds to the total number of model runs. We repeat 10 times for each value of $N$ (i.e., 10 heteroskedastic Gaussian processes are built using the same number of model runs). As a consequence, error estimates for each $N$ are represented by box plots.

\subsection{Example 1: a two-dimensional simulator}
\label{sec:ex1}
The first example is the \emph{Black--Scholes} model used for stock prices 
\cite{McNeil2005}:
\begin{equation}
	\D S_t = x_1\, S_t\,\D t + x_2\, S_t \,\D W_t,  \label{eq:GBM} 
\end{equation}
where $\ve{x} = (x_1,x_2)^T$ are the input parameters, corresponding to the 
expected return rate and volatility of a stock, respectively. $W_t$ is a 
standard Wiener process, which represents the source of stochasticity. 
\Cref{eq:GBM} is a stochastic differential equation whose solution 
$S_t(\ve{x})$ is a stochastic process for given parameters $\ve{x}$. 
Note that we explicitly express 
$\ve{x}$ in $S_t(\ve{x})$ to emphasize that $\ve{x}$ are input parameters, but 
the stochastic equation is defined with respect to time. Without loss of 
generality, we set the initial condition to $S_0(\ve{x}) = 1$.
\par
In this example, we are interested in $Y(\ve{x}) = S_1(\ve{x})$, which 
corresponds to the stock value in one year \ie $t=1$. We set $X_1 \sim 
\cu(0,0.1)$ and $X_2\sim \cu(0.1,0.4)$ to represent the input uncertainty, 
where the ranges are selected based on parameters calibrated from real data 
\cite{Reddy2016}.
\par
The solution to \cref{eq:GBM} can be derived using It\^o calculus 
\cite{Shreve2004}: $Y(\ve{x})$ follows a lognormal distribution defined by
\begin{equation}\label{eq:GMB_solu}
	Y(\ve{x}) \sim \cl\cn\left(x_1-\frac{x^2_2}{2},x_2\right).
\end{equation}
As the distribution of $Y(\ve{x})$ is known, it is not necessary to simulate 
the whole process $S_t(\ve{x})$ with time integration to evaluate 
$S_1(\ve{x})$. Instead, we can directly generate samples from the distribution 
defined in \cref{eq:GMB_solu}. 

\begin{figure}[!htbp]
	\centering
	\begin{subfigure}{.45\linewidth}
		\centering
		\includegraphics[height=0.7\linewidth, keepaspectratio]{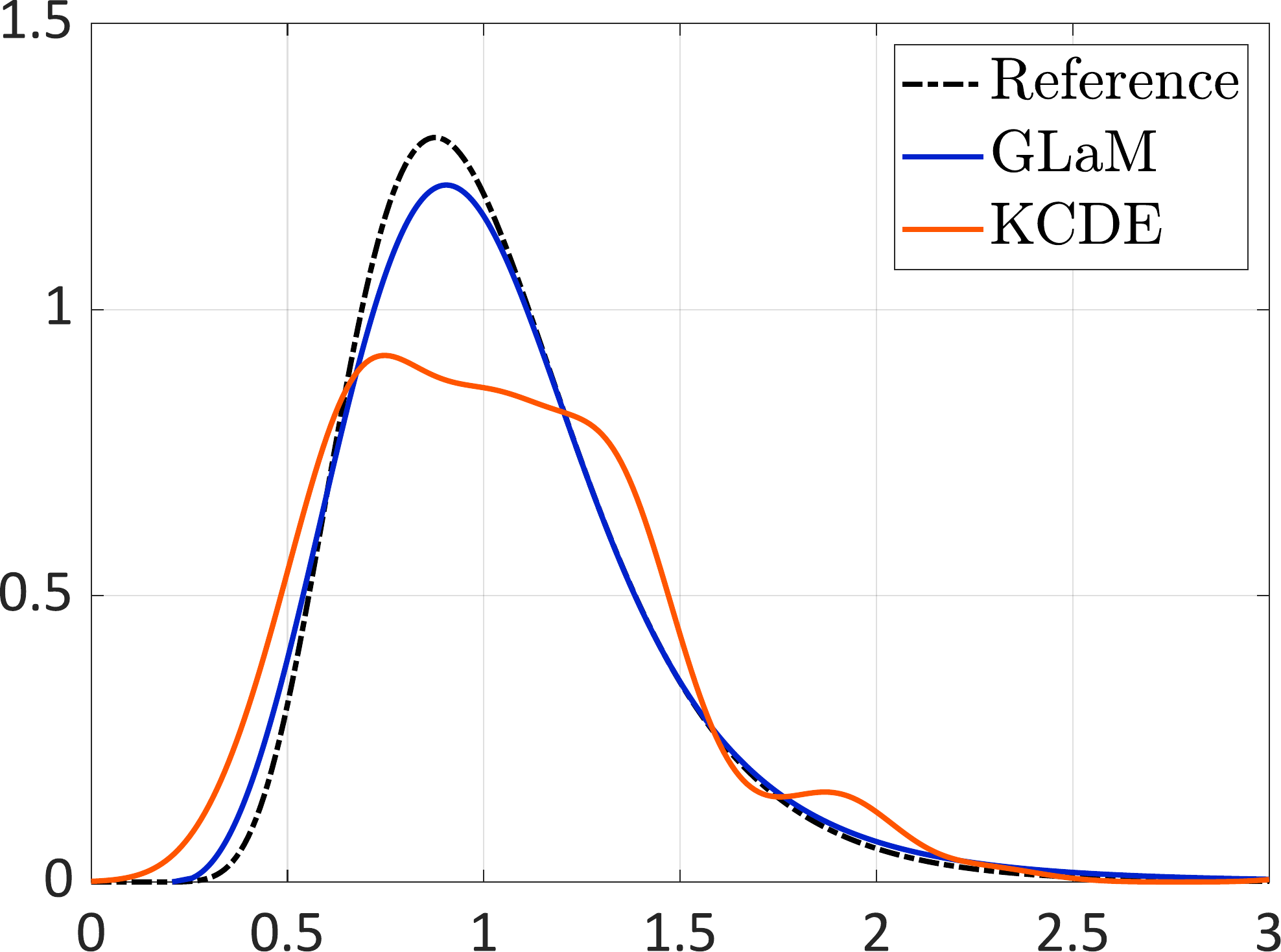}
		\caption{PDF for $\ve{x} = (0.03,0.33)^T$}
	\end{subfigure}
	\hspace{0.5cm}
	\begin{subfigure}{.45\linewidth}
		\centering
		\includegraphics[height=0.7\linewidth, keepaspectratio]{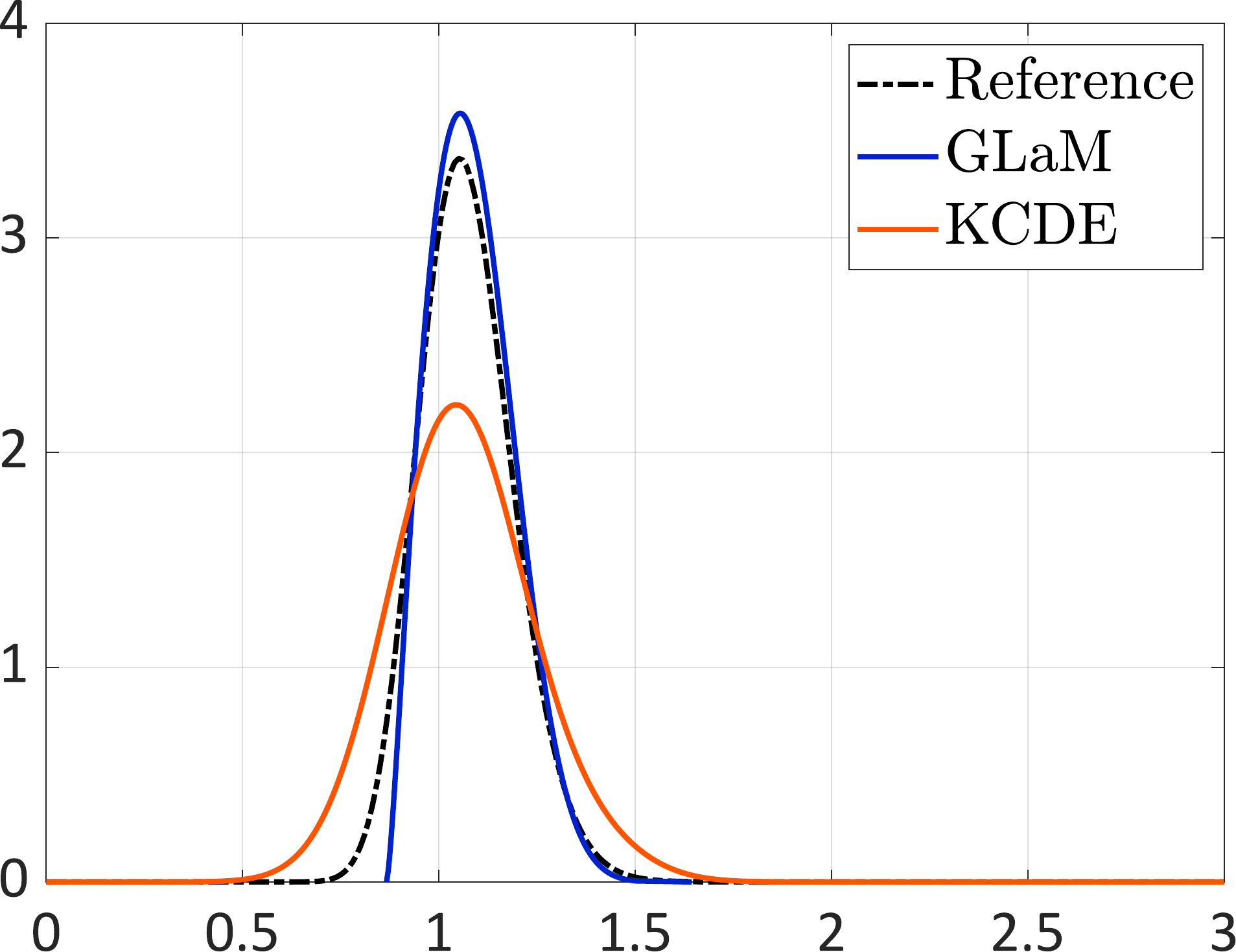}
		\caption{PDF for $\ve{x} = (0.07,0.11)^T$}
	\end{subfigure}
	\caption{Example 1 --- Comparisons of the emulated PDF, $N=500$.}
	\label{fig:GBMpdf}
\end{figure}

\begin{figure}[!htbp]
	\centering
	\begin{subfigure}{.44\linewidth}
		\centering
		\includegraphics[height=0.82\linewidth, keepaspectratio]{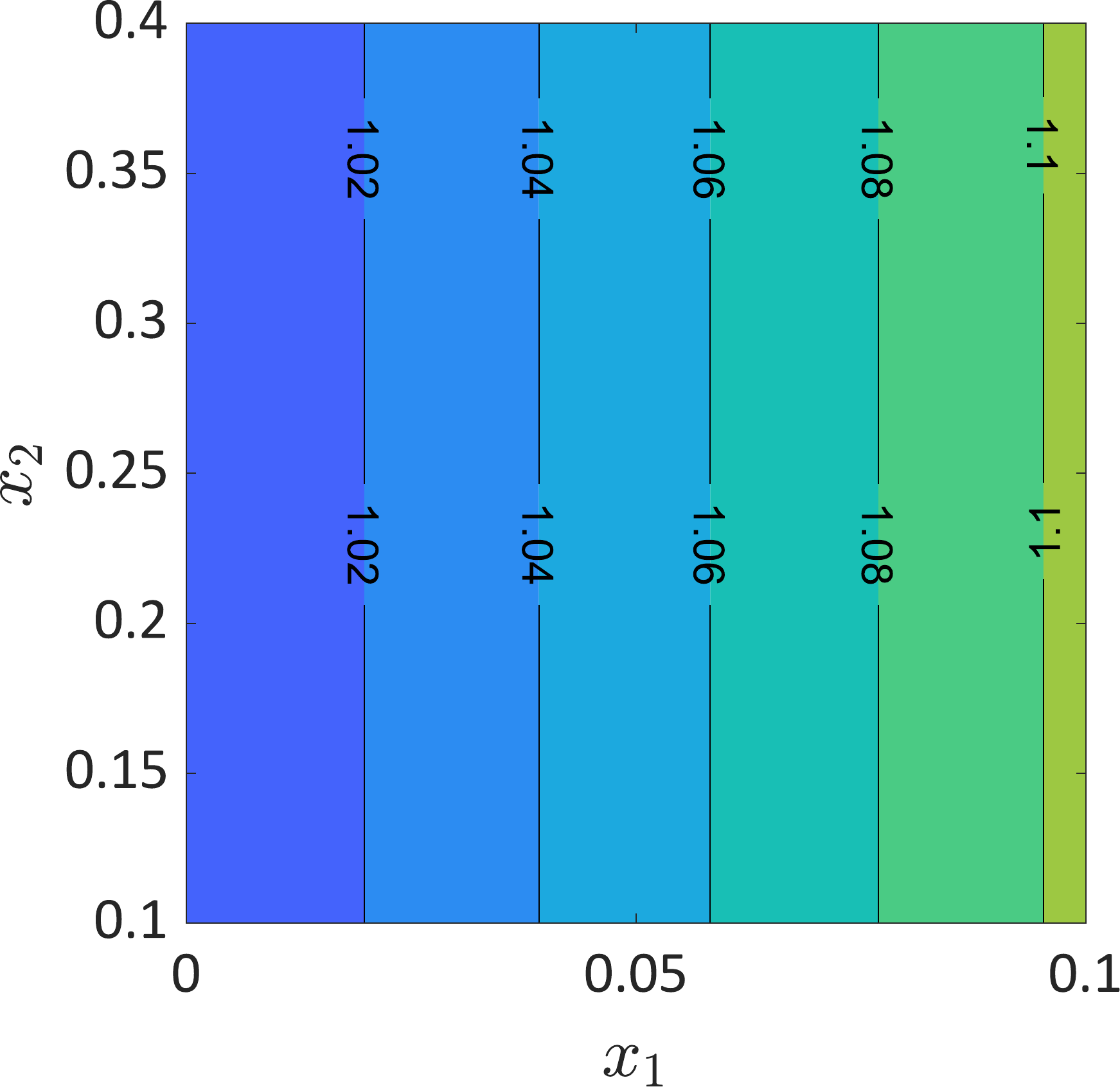}
		\caption{Reference}
	\end{subfigure}
	\begin{subfigure}{.44\linewidth}
		\centering
		\includegraphics[height=0.82\linewidth, keepaspectratio]{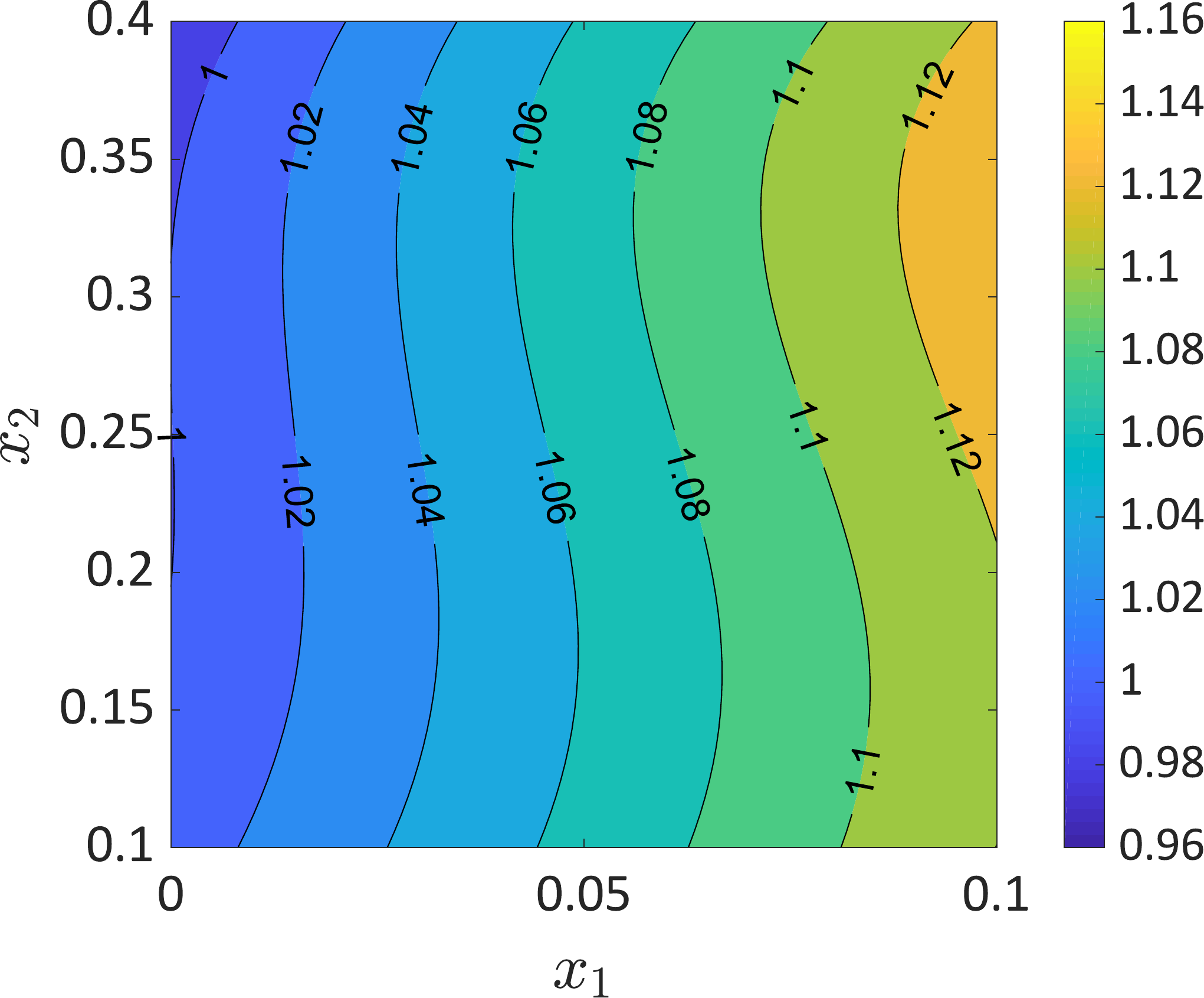}
		\caption{GLaM}
	\end{subfigure}
	\begin{subfigure}{.44\linewidth}
		\centering
		\includegraphics[height=0.82\linewidth, keepaspectratio]{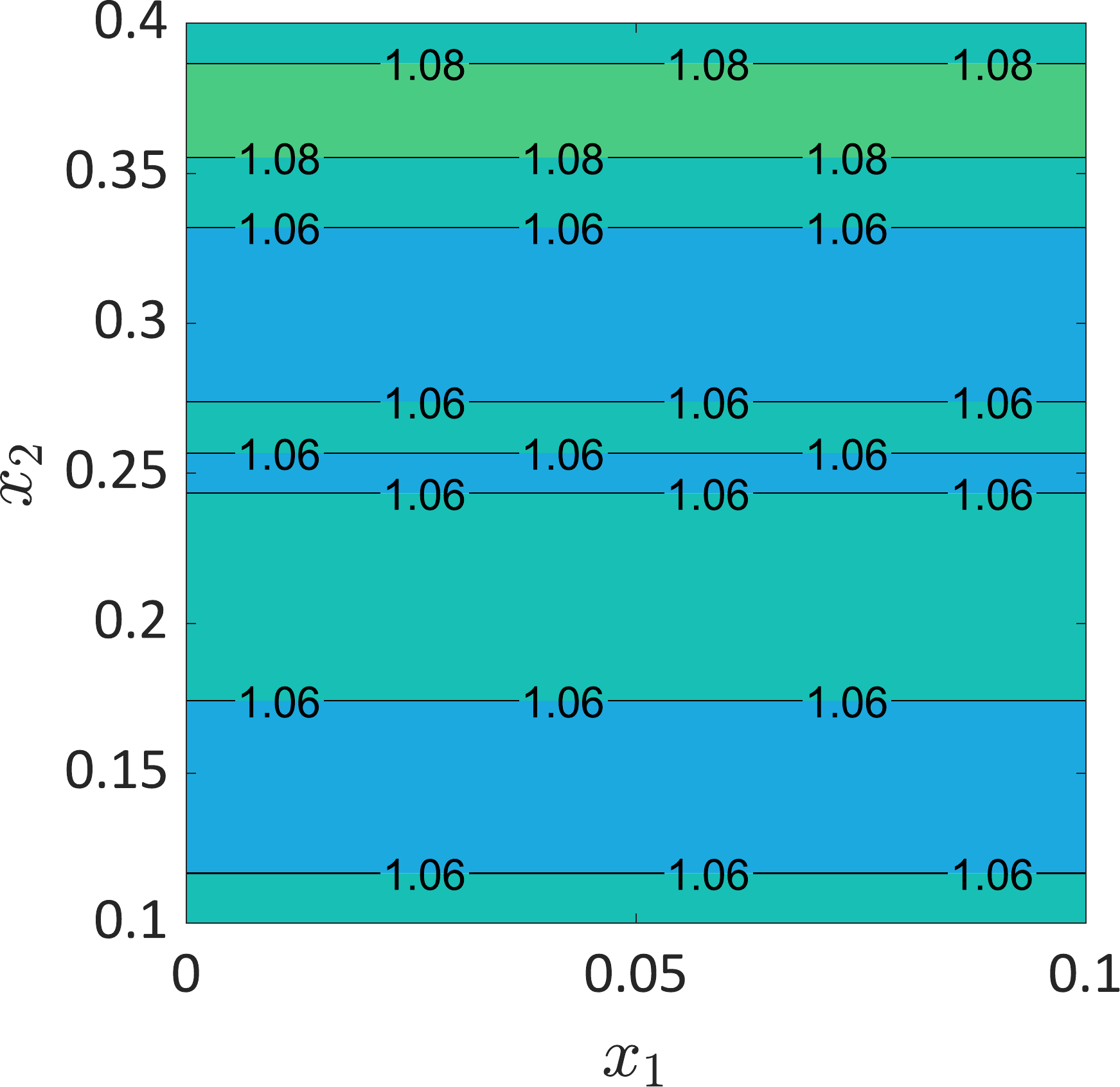}
		\caption{KCDE}
	\end{subfigure}
	\begin{subfigure}{.44\linewidth}
		\centering
		\includegraphics[height=0.82\linewidth, keepaspectratio]{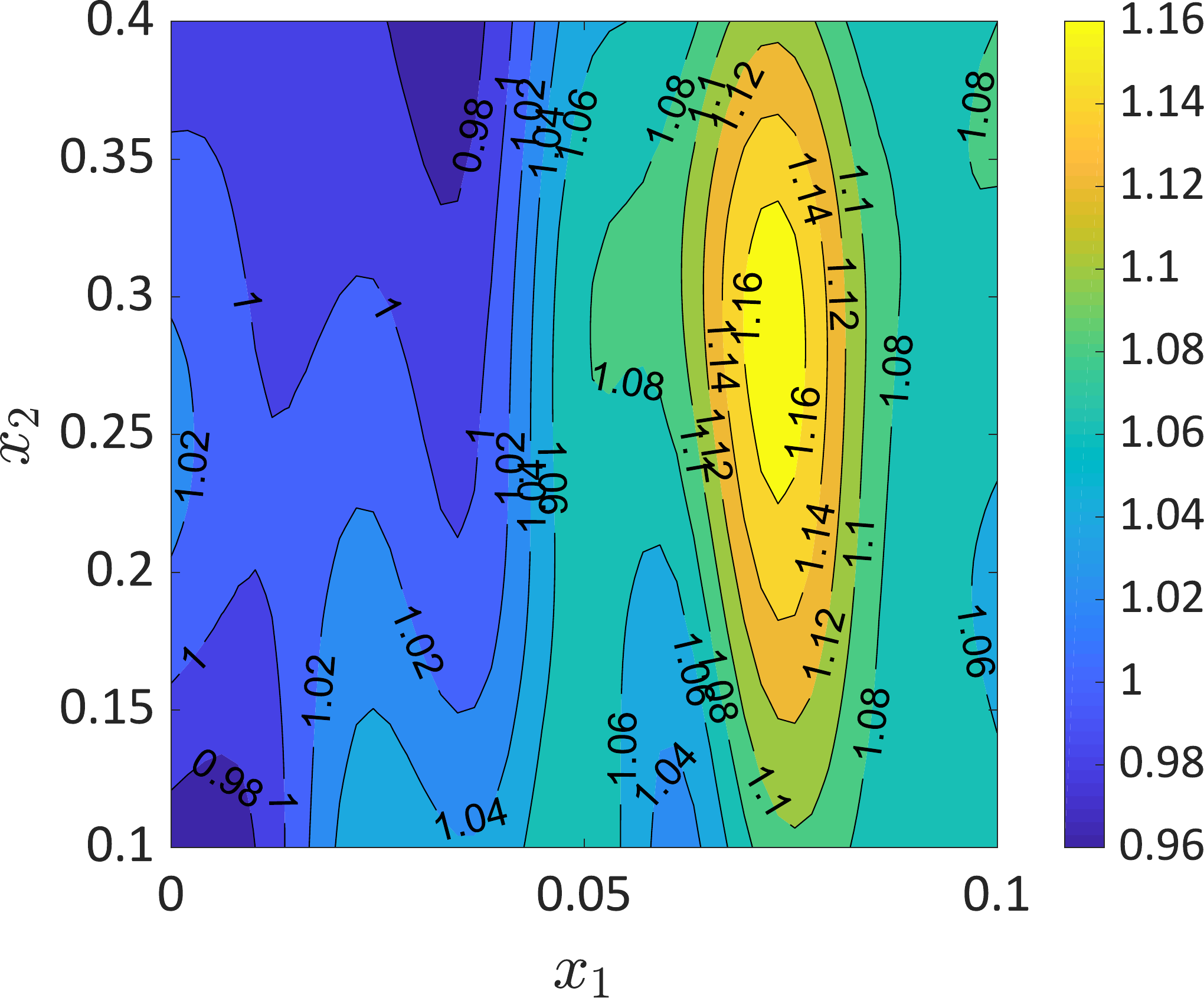}
		\caption{GP}
	\end{subfigure}
	\caption{Example 1 --- Comparisons of the mean function estimation, $N=500$.}
	\label{fig:GBMm}
\end{figure}

\begin{figure}[!htbp]
	\centering
	\begin{subfigure}{.44\linewidth}
		\centering
		\includegraphics[height=0.82\linewidth, keepaspectratio]{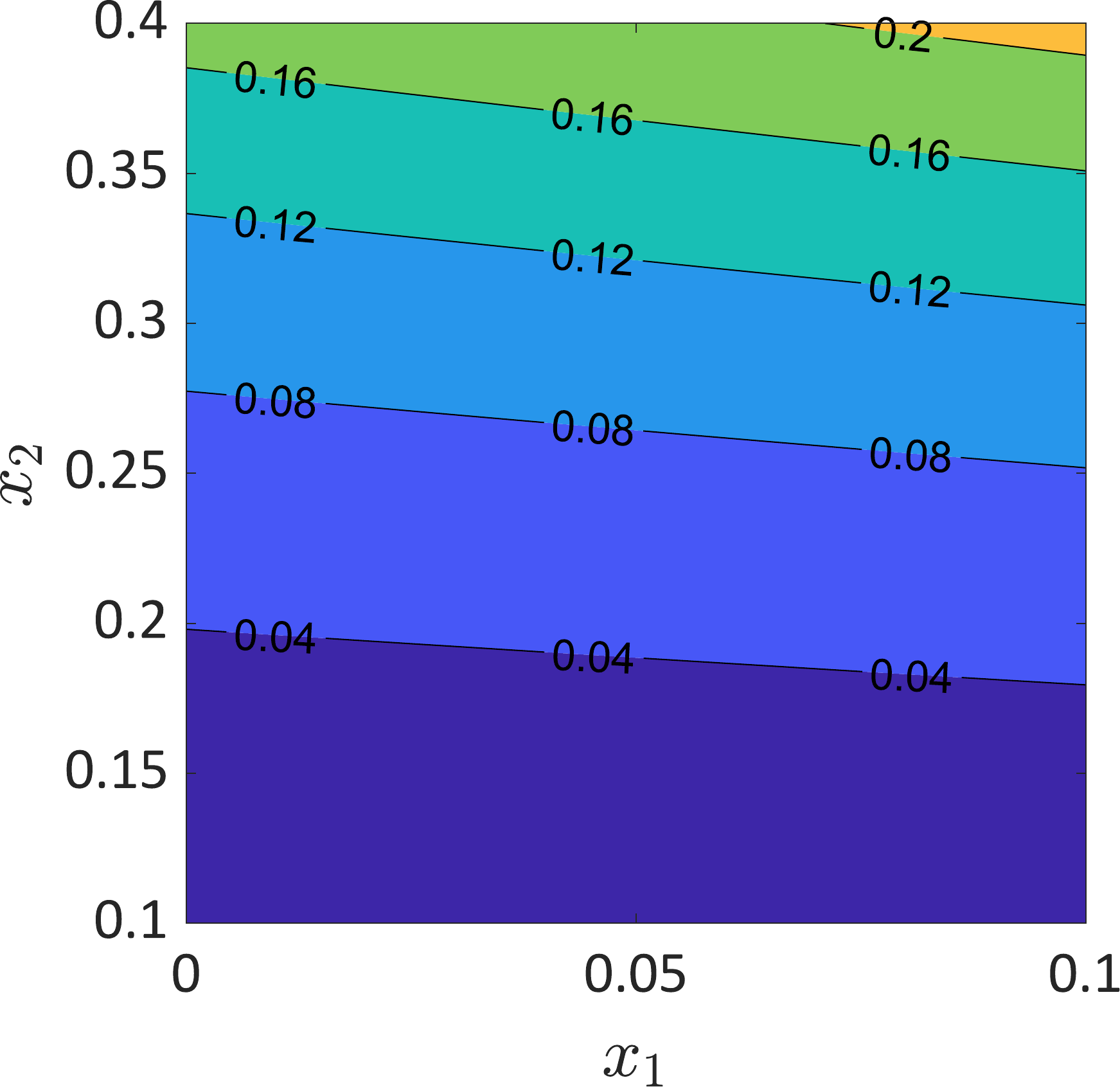}
		\caption{Reference}
	\end{subfigure}
	\begin{subfigure}{.44\linewidth}
		\centering
		\includegraphics[height=0.82\linewidth, keepaspectratio]{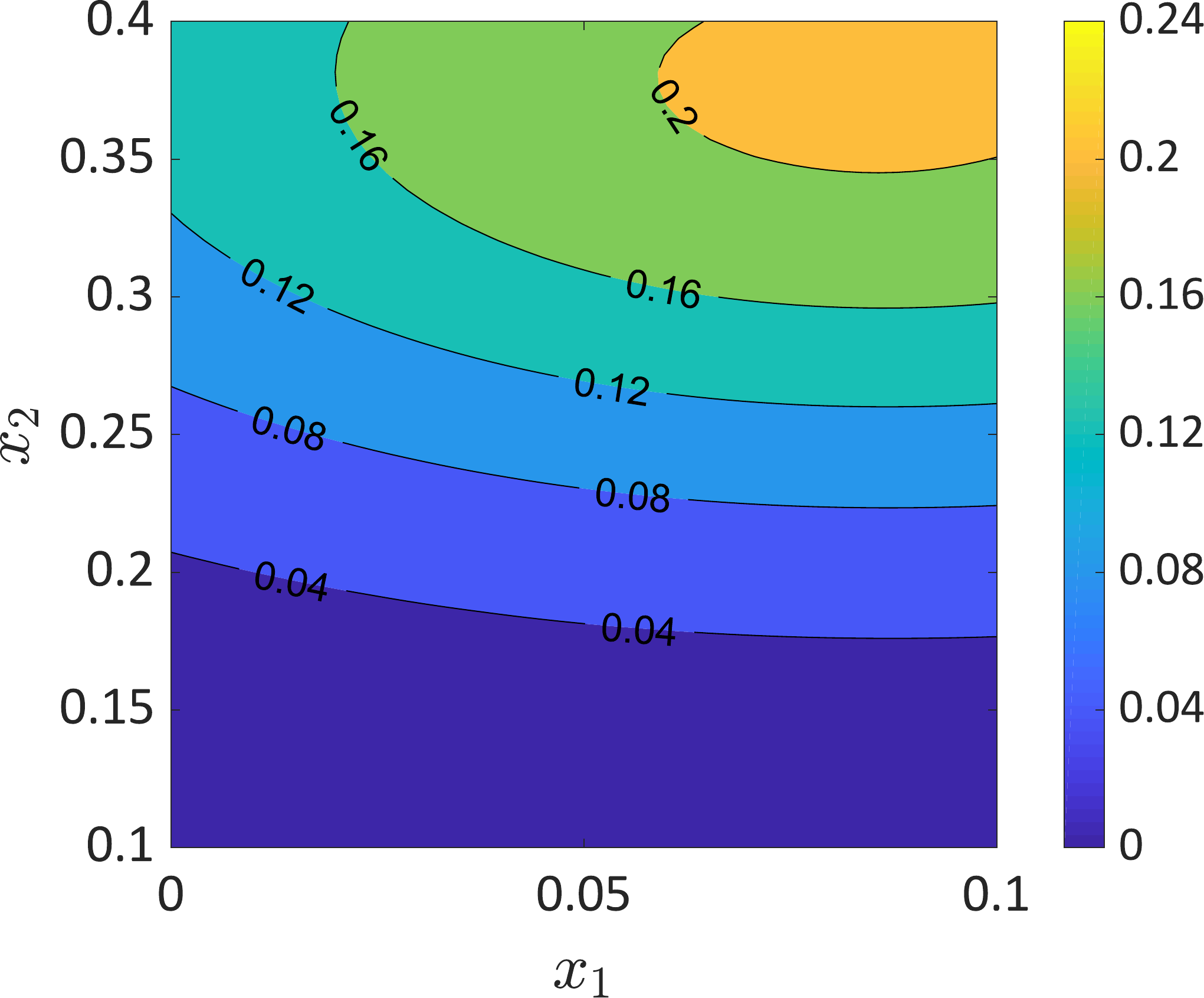}
		\caption{GLaM}
	\end{subfigure}
	\begin{subfigure}{.44\linewidth}
		\centering
		\includegraphics[height=0.82\linewidth, keepaspectratio]{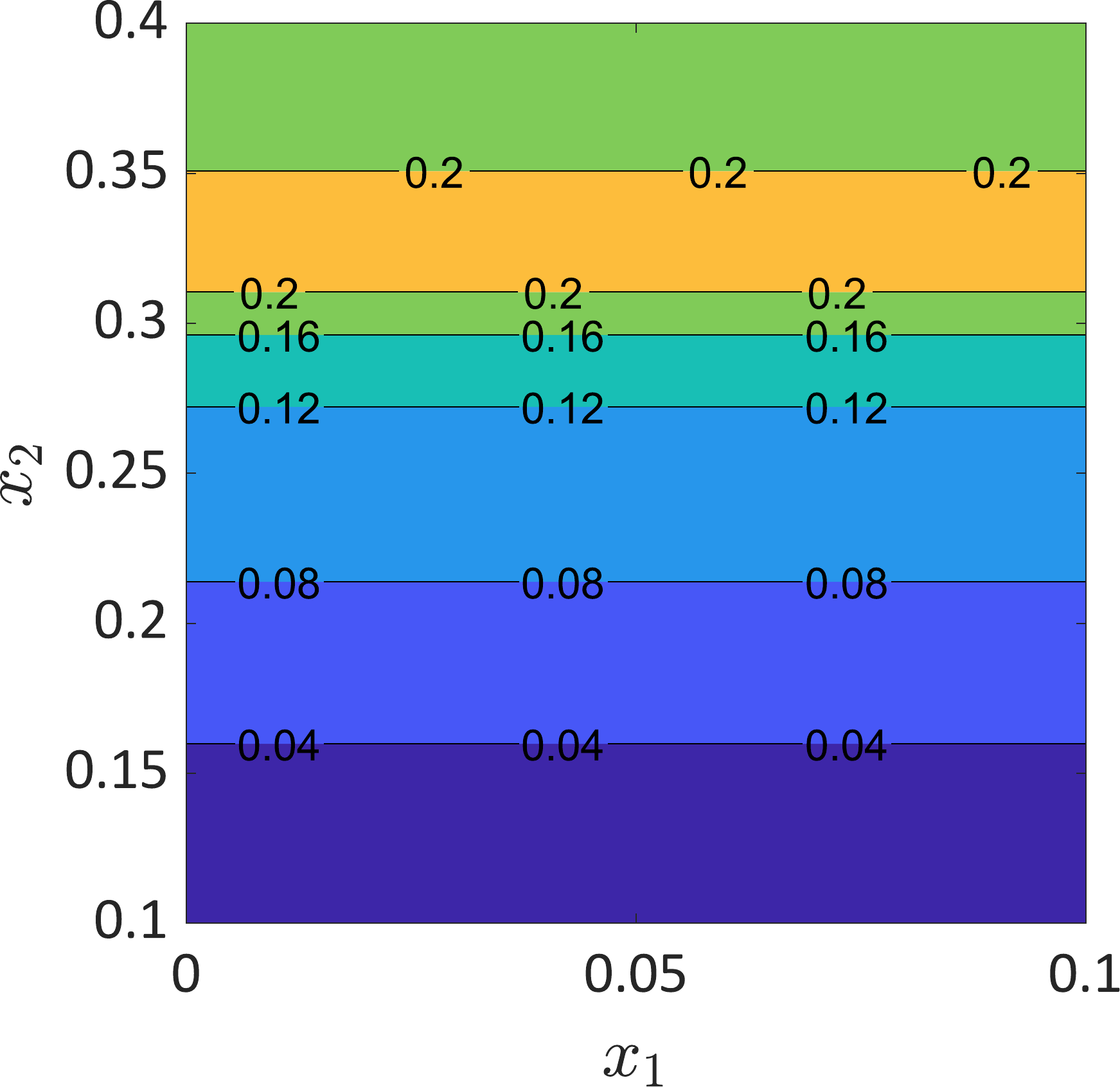}
		\caption{KCDE}
	\end{subfigure}
	\begin{subfigure}{.44\linewidth}
		\centering
		\includegraphics[height=0.82\linewidth, keepaspectratio]{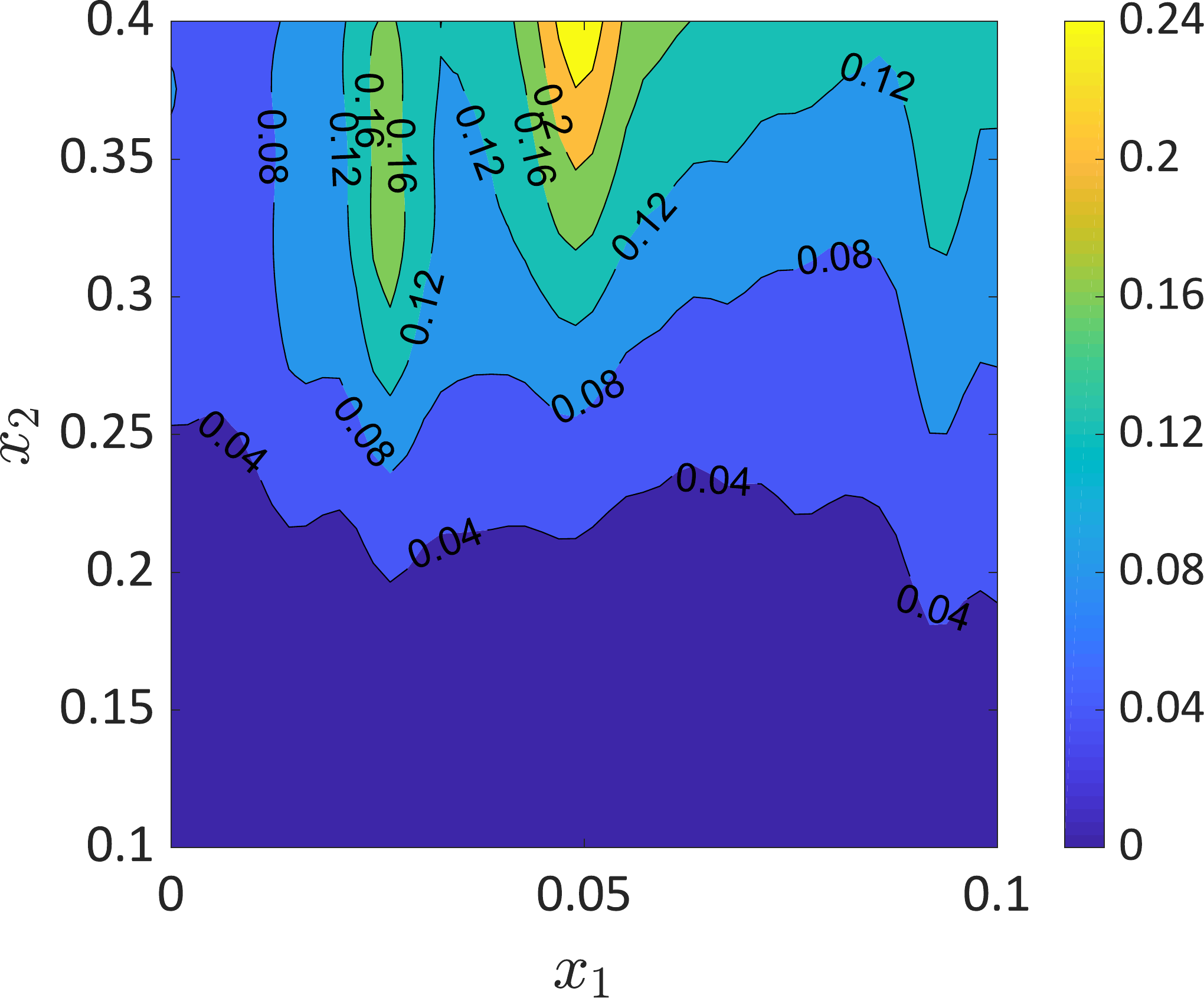}
		\caption{GP}
	\end{subfigure}
	\caption{Example 1 --- Comparisons of the variance function estimation, $N=500$.}
	\label{fig:GBMv}
\end{figure}

\Cref{fig:GBMpdf} shows two PDFs predicted by a GLaM and a KCDE built on an experimental design of size $N=500$. We observe that with $500$ model runs, the KCDE yields PDFs with spurious oscillations and demonstrates relatively poor representation of the bulk. In contrast, the GLaM can better approximate the underlying response PDF in terms of both magnitude and shape variations. \Cref{fig:GBMm,fig:GBMv} compare the mean and variance function predicted by the GLaM, KCDE, and GP. The analytical mean function following \cref{eq:GMB_solu} is $\exp(x_1)$, which only depends on the first variable. The GLaM gives an accurate estimate of the mean function, whereas the KCDE captures a wrong dependence, and GP produces a rather complex structure. For the variance function, the GLaM yields a more detailed trend than the KCDE and GP.
\par
For quantitative comparisons, \Cref{fig:GBM_WS} summarizes the error measure 
\cref{eq:Rlevel1} with respect to the size of experimental design. The 
accuracy of the oracle normal approximation is also reported (black dashed 
line). This error is only due to model misspecifications because we 
use the true mean and variance (however, the true response distribution is lognormal). The GP approach performs rather poorly and converges to the oracle normal approximation when the number of points in the experimental design increases. This means that it can accurately estimate the mean and variance functions for large data sets. However, due to the limitation of the Gaussian assumption, GP cannot further decrease the error. The average error of GLaMs built on $N=500$ model 
runs are smaller than that of the normal approximation. For $N>500$, GLaMs 
clearly provide more accurate results. KCDEs show a 
slow rate of convergence even in this example of dimension two. In contrast, 
GLaMs reveal high efficiency with a faster decrease of the errors. In terms of 
the average error, GLaMs outperform KCDEs for all sizes of experimental 
design. Furthermore, GLaMs yield an average error near $0.1$ for $N=1{,}000$, 
which can be hardly achieved by KCDEs even with four times more model runs.

\begin{figure}[!htbp]
	\centering
	\includegraphics[width=.65\linewidth, keepaspectratio]{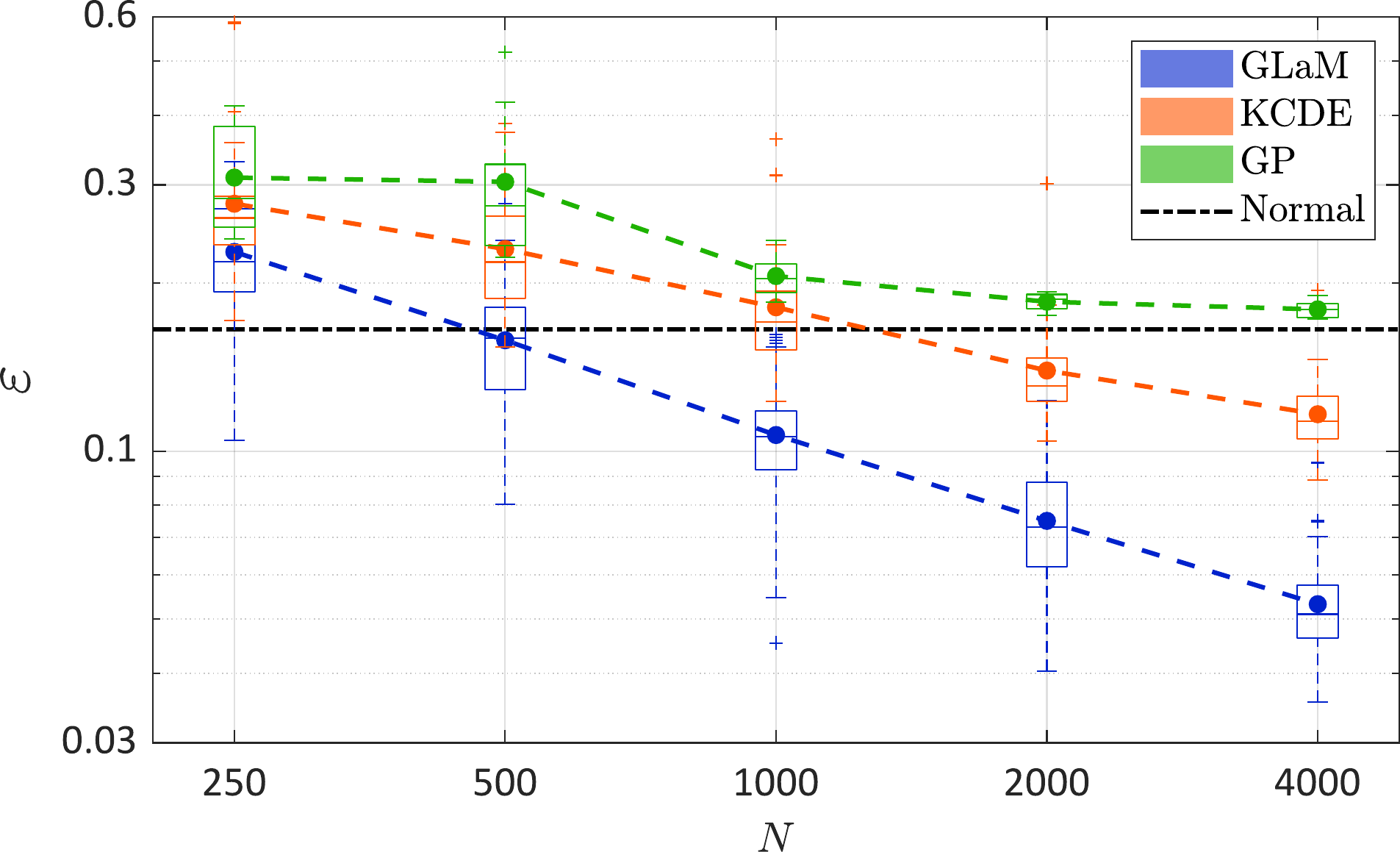}
	\caption{Example 1 --- Comparison of the convergence between GLaMs and KCDEs in terms of the normalized Wasserstein distance as a function of the size of the experimental design. The dashed lines denote the average value over 50 repetitions of the full analysis. The green box plots and associated dashed lines correspond to the errors of the heteroskedastic Gaussian Process with sequential design (10 repetitions for each size of the experimental design). The black dash-dotted line represents the error of the model assuming that the response distribution is normal with the true mean and variance.}
	\label{fig:GBM_WS}
\end{figure}

\subsection{Example 2: a five-dimensional simulator}
\label{sec:ex2}
The second example is given by
\begin{equation}
	Y(\ve{x})=\cm_s(\ve{x},\omega) = \mu(\ve{x})+\sigma(\ve{x}) \cdot Z(\omega),
\end{equation}
where $\ve{X} \sim \cu\left([0,1]^5\right)$ are the input variables, and $Z 
\sim \cn(0,1)$ is the latent variable that introduces the stochasticity. The 
simulator has an input dimension of $M=5$, which is used to show the 
performance of the proposed method in a moderate-dimensional problem. By 
definition, $Y(\ve{x})$ is a Gaussian random variable with mean 
$\mu(\ve{x})$ and standard deviation $\sigma(\ve{x})$ which are defined by
\begin{equation}\label{eq:ex2mv}
	\begin{split}
		\mu(\ve{x}) &= 3-\sum_{j=1}^{5}j\,x_j + \frac{1}{5}\sum_{j=1}^{5}j\,x^3_j +  
		\frac{1}{15}\sum_{j=1}^{5}j\,\log\left((x^2_j+x^4_j)\right) + x_1 \, x^2_2 - 
		x_5 \, x_3 + x_2 \, x_4, \\
		\sigma(\ve{x})& = \exp\left(\frac{1}{10}\sum_{j=1}^{5} j\,x_j\right) ,
	\end{split}
\end{equation}
Thus, this example has a nonlinear mean function and a strong 
heteroskedastic effect: the variance varies between 1 and 20.
\par
\begin{figure}[!htbp]
	\centering
	\begin{subfigure}[b]{.45\linewidth}
		\centering
		\includegraphics[height=0.7\linewidth, keepaspectratio]{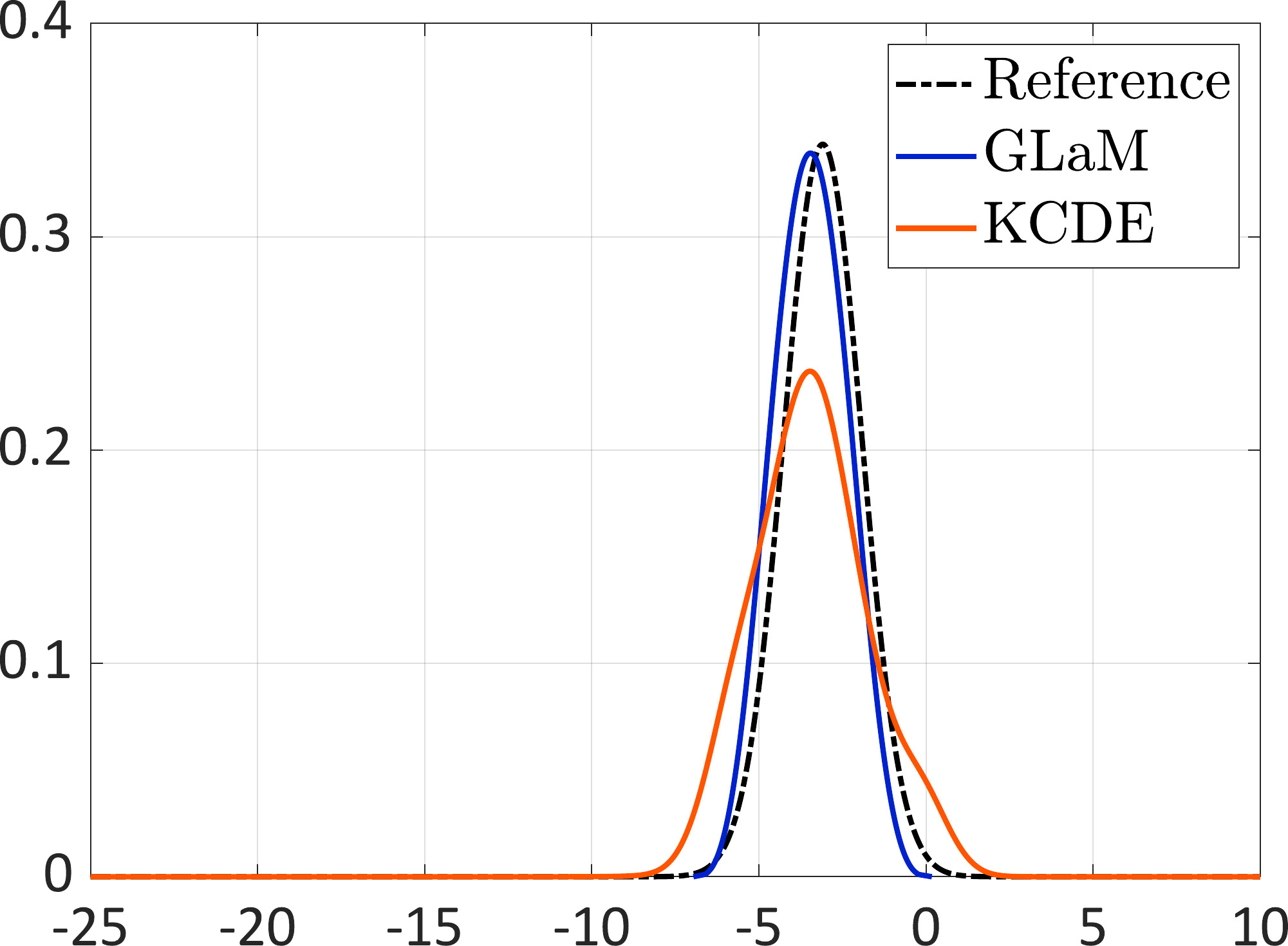}
		\caption{PDF for $\ve{x} = (0.1,0.1,0.1,0.1,0.1)^T$}
		\label{fig:normpdf1}
	\end{subfigure}
	\hspace{0.5cm}
	\begin{subfigure}[b]{.45\linewidth}
		\centering
		\includegraphics[height=0.7\linewidth, keepaspectratio]{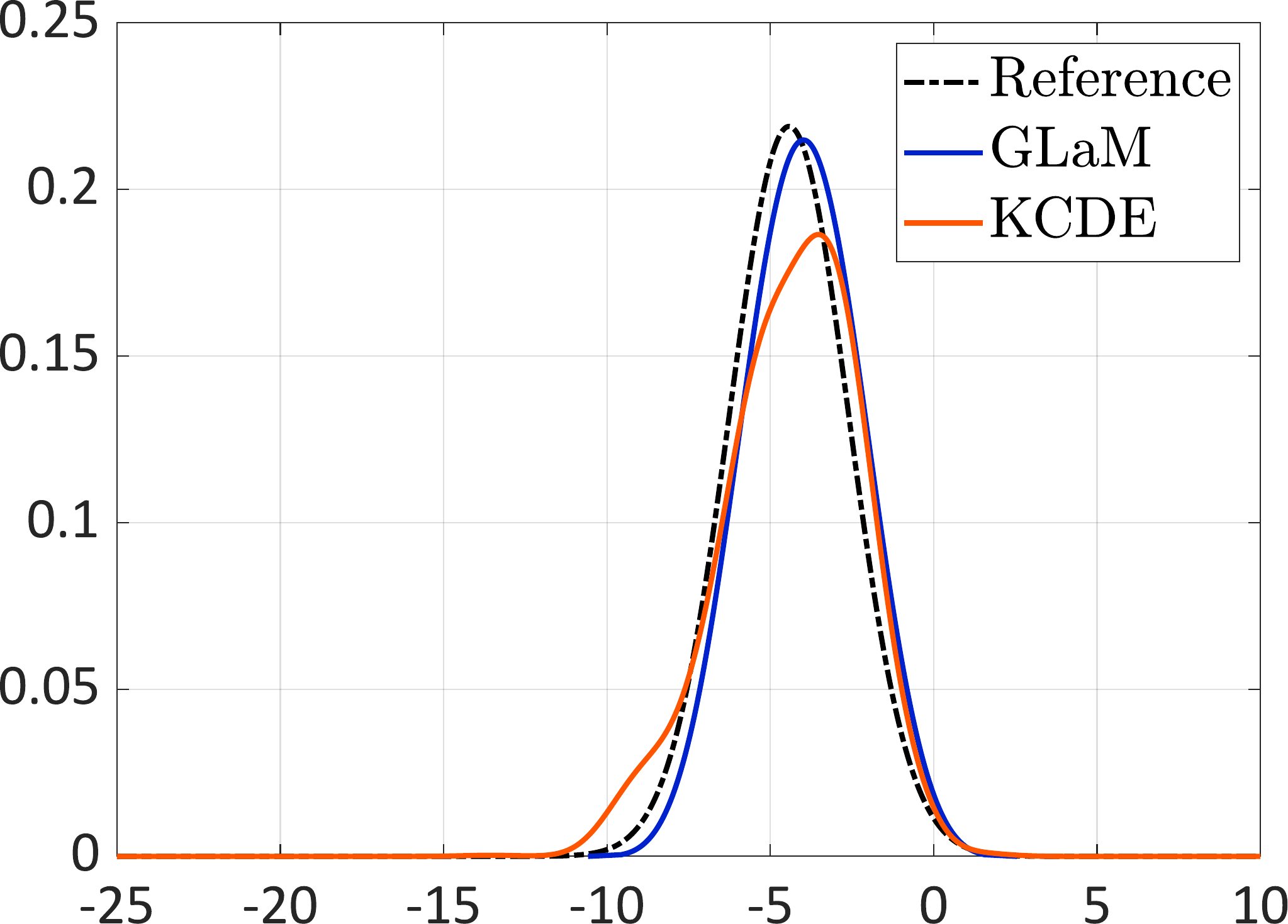}
		\caption{PDF for $\ve{x} = (0.4,0.4,0.4,0.4,0.4)^T$}
		\label{fig:normpdf2}
	\end{subfigure}
	\begin{subfigure}[b]{.45\linewidth}
		\centering
		\includegraphics[height=0.7\linewidth, keepaspectratio]{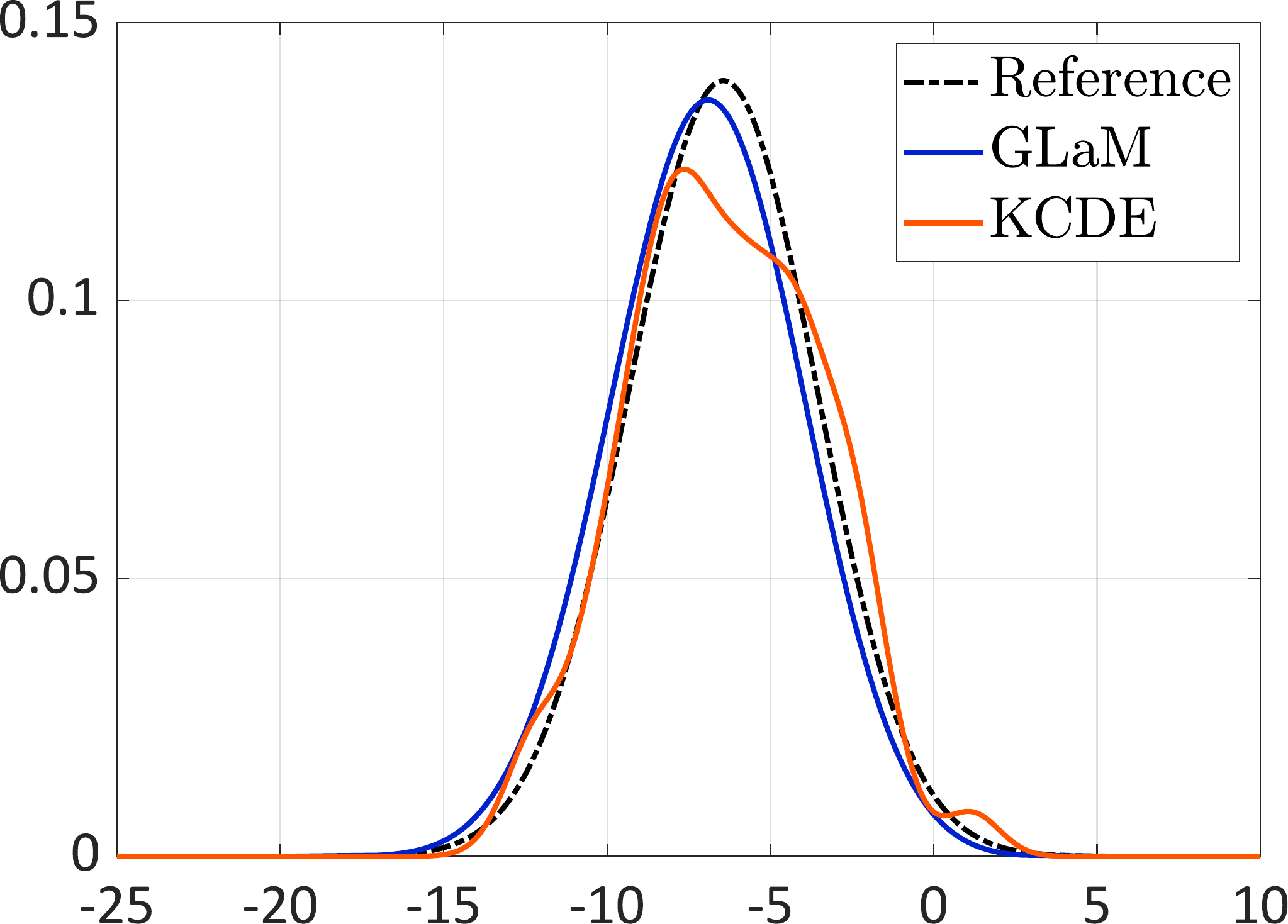}
		\caption{PDF for $\ve{x} = (0.7,0.7,0.7,0.7,0.7)^T$}
		\label{fig:normpdf3}
	\end{subfigure}
	\hspace{0.5cm}
	\begin{subfigure}[b]{.45\linewidth}
		\centering
		\includegraphics[height=0.7\linewidth, keepaspectratio]{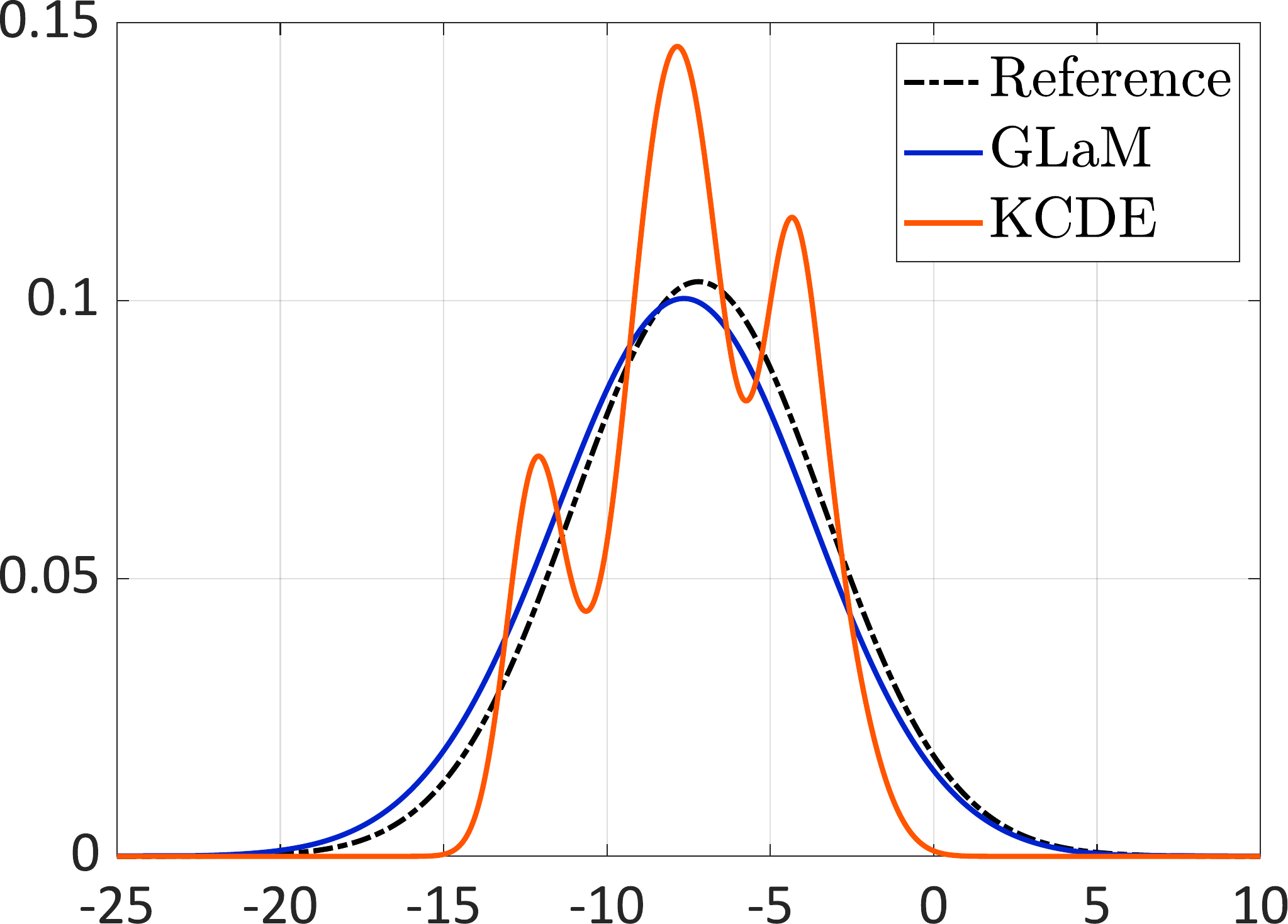}
		\caption{PDF for $\ve{x} = (0.9,0.9,0.9,0.9,0.9)^T$}
		\label{fig:normpdf4}
	\end{subfigure}
	\caption{Example 2 --- Comparisons of the emulated PDF, $N=1{,}000$. Variance 
		values $1.35$, $3.32$, $8.17$, $14.88$ from (a) to (d)}
	\label{fig:normpdf}
\end{figure}
\begin{figure}[!htbp]
	\centering
	\begin{subfigure}{.44\linewidth}
		\centering
		\includegraphics[height=0.82\linewidth, keepaspectratio]{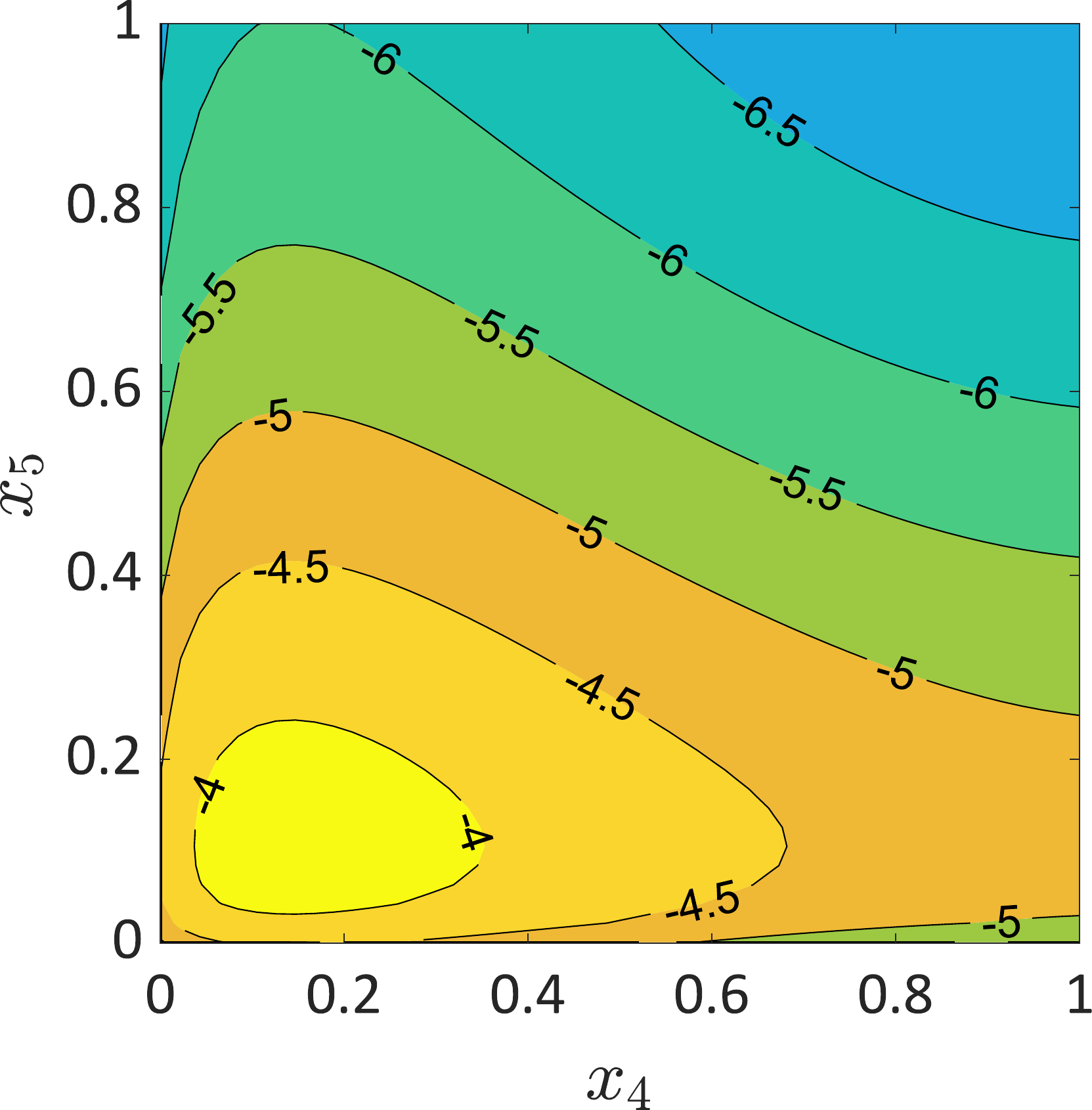}
		\caption{Reference}
	\end{subfigure}
	\begin{subfigure}{.44\linewidth}
		\centering
		\includegraphics[height=0.82\linewidth, keepaspectratio]{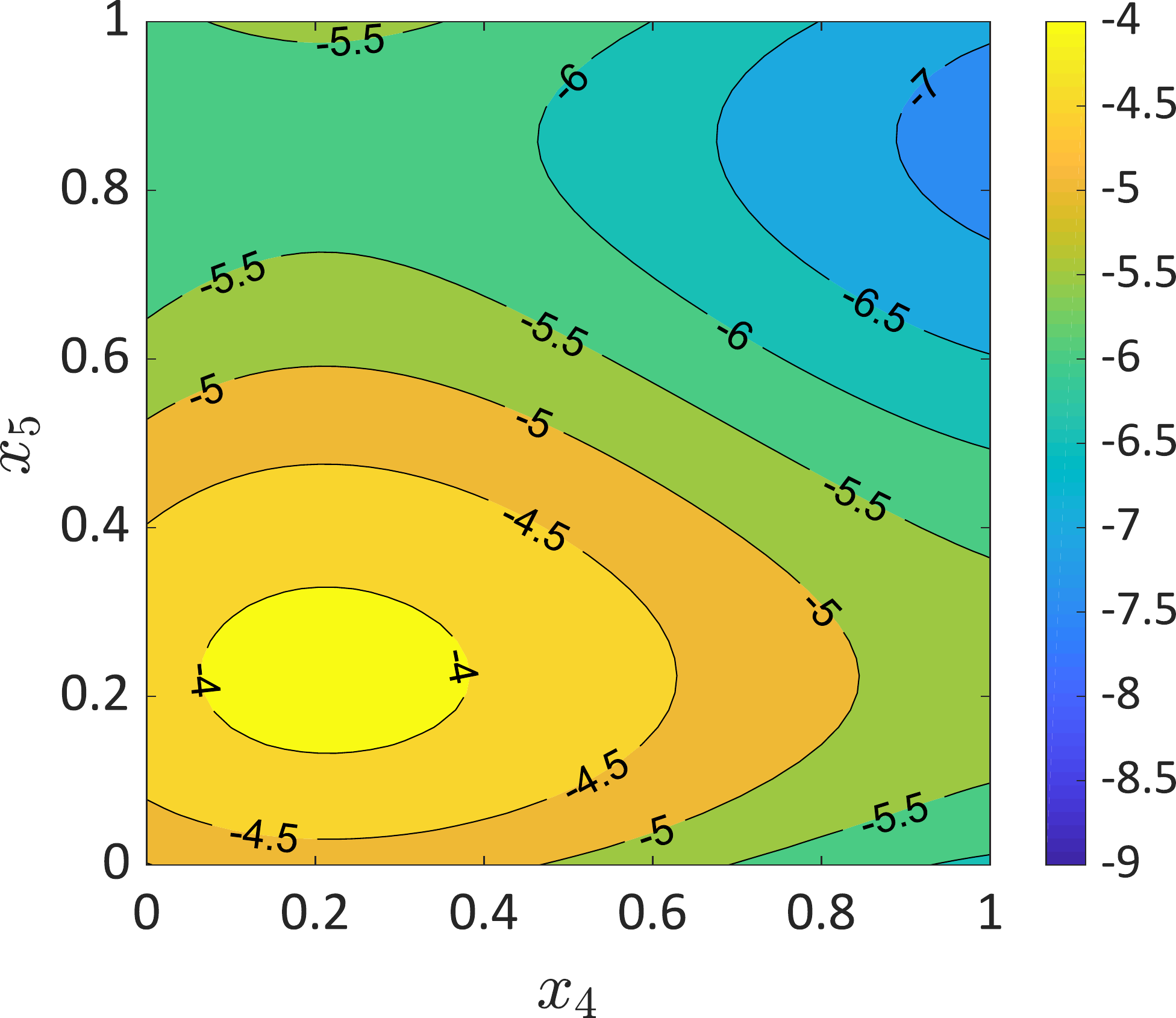}
		\caption{GLaM}
	\end{subfigure}
	\begin{subfigure}{.44\linewidth}
		\centering
		\includegraphics[height=0.82\linewidth, keepaspectratio]{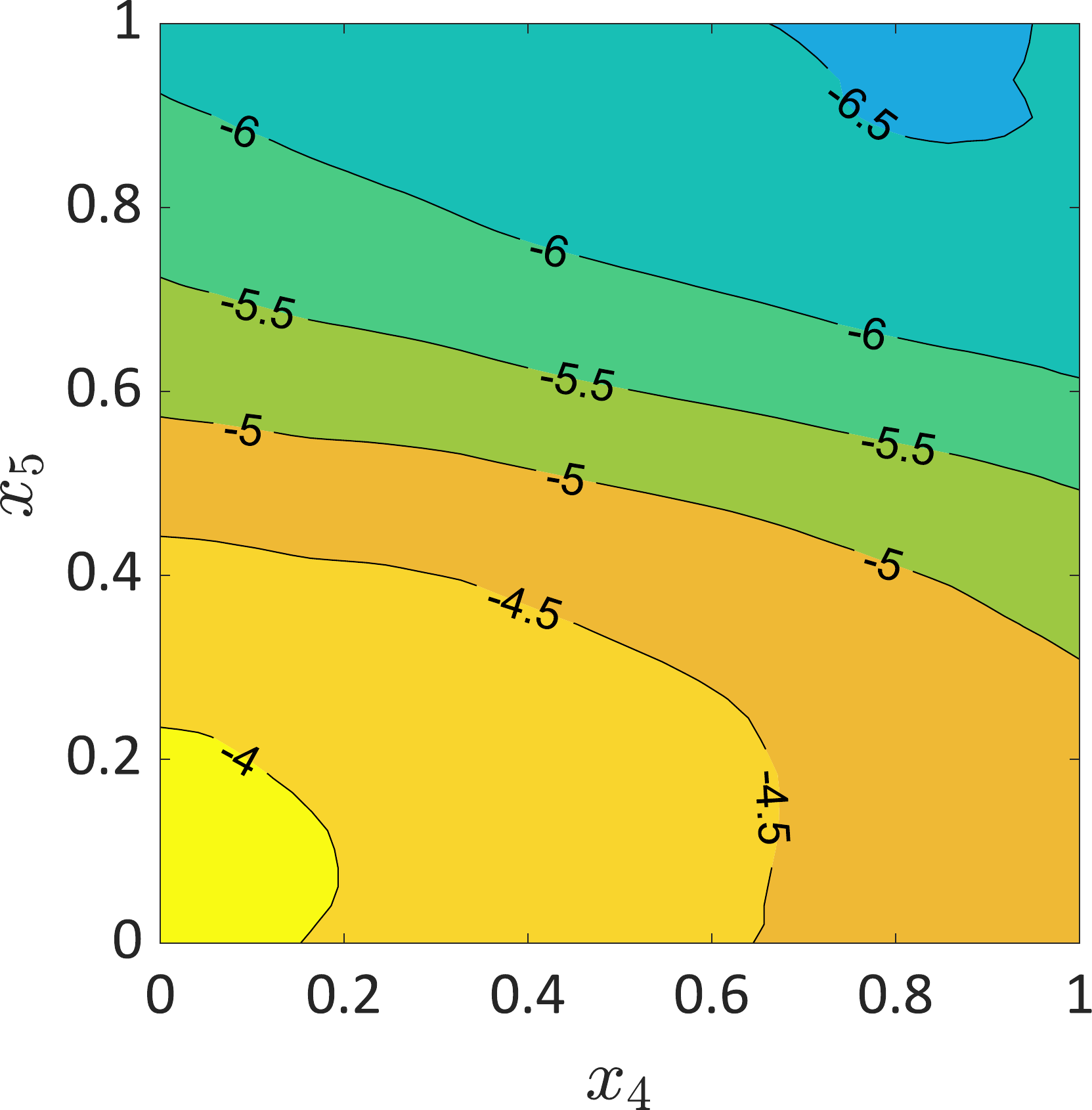}
		\caption{KCDE}
	\end{subfigure}
	\begin{subfigure}{.44\linewidth}
		\centering
		\includegraphics[height=0.82\linewidth, keepaspectratio]{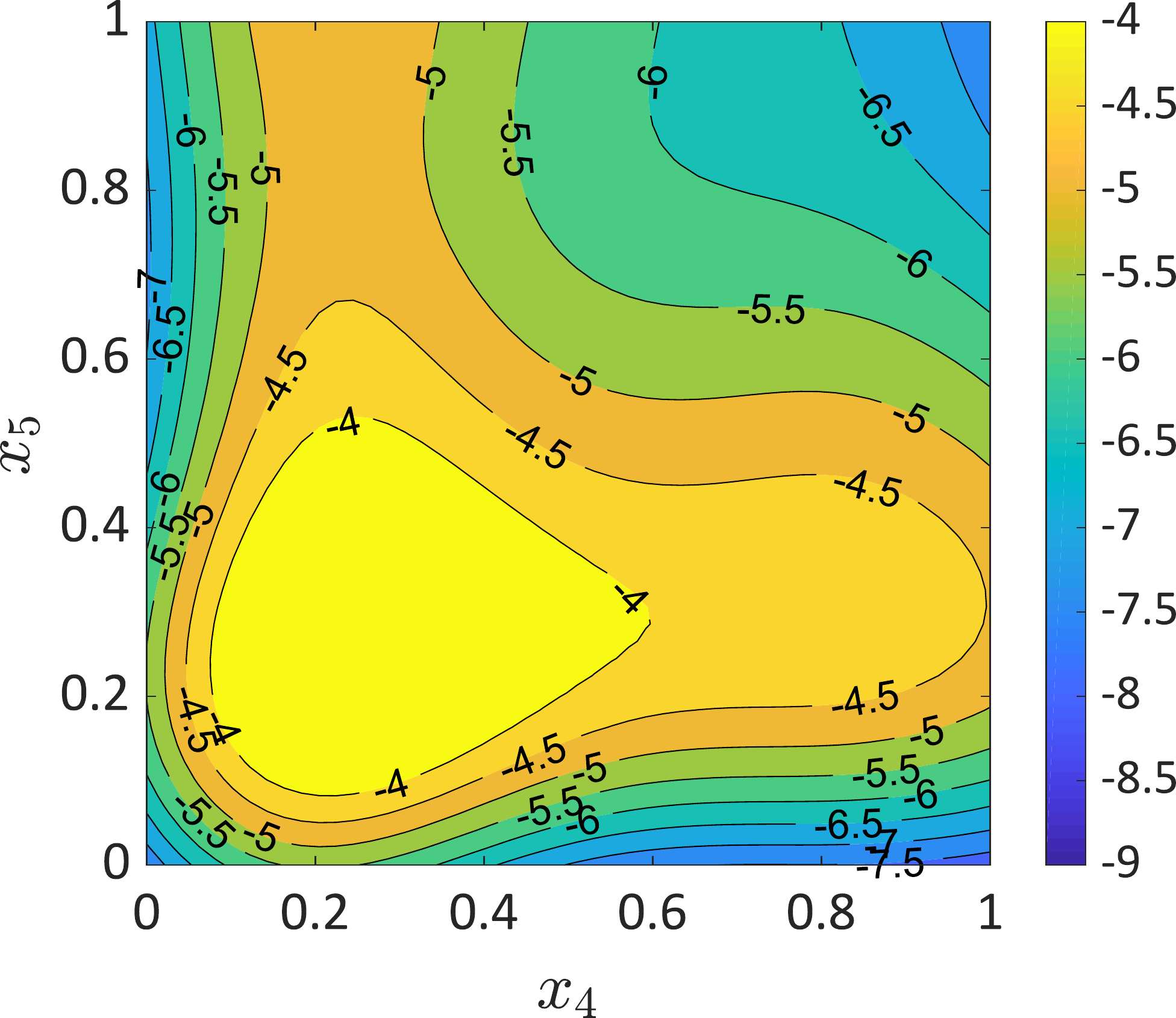}
		\caption{GP}
	\end{subfigure}
	\caption{Example 2 --- Comparisons of the mean function estimation in the plan $x_4-x_5$ with all the other input fixed at their expected value. The surrogate models are fitted to an ED with $N=1{,}000$.}
	\label{fig:normm}
\end{figure}
\begin{figure}[!htbp]
	\centering
	\begin{subfigure}{.44\linewidth}
		\centering
		\includegraphics[height=0.82\linewidth, keepaspectratio]{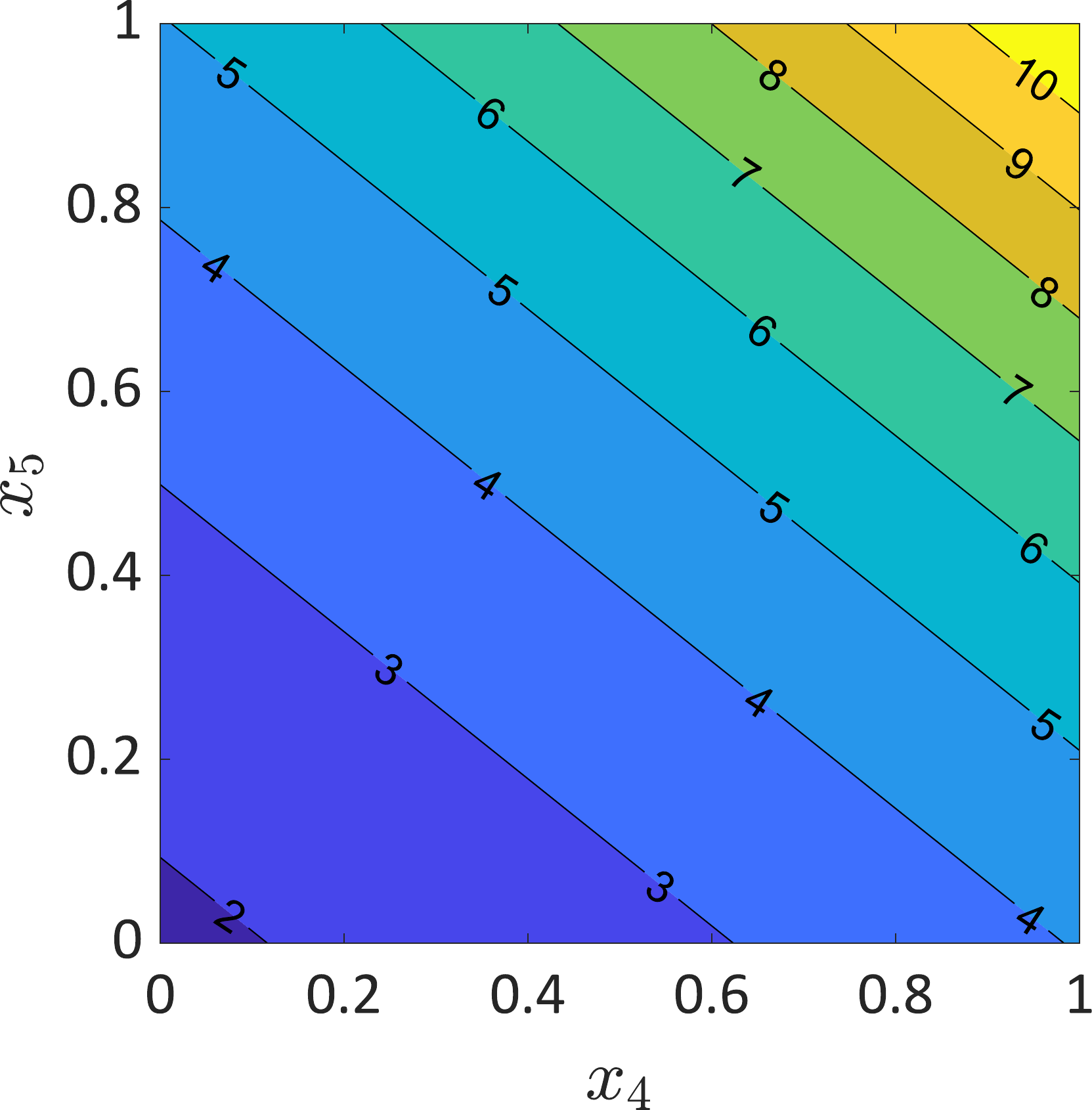}
		\caption{Reference}
	\end{subfigure}
	\begin{subfigure}{.44\linewidth}
		\centering
		\includegraphics[height=0.82\linewidth, keepaspectratio]{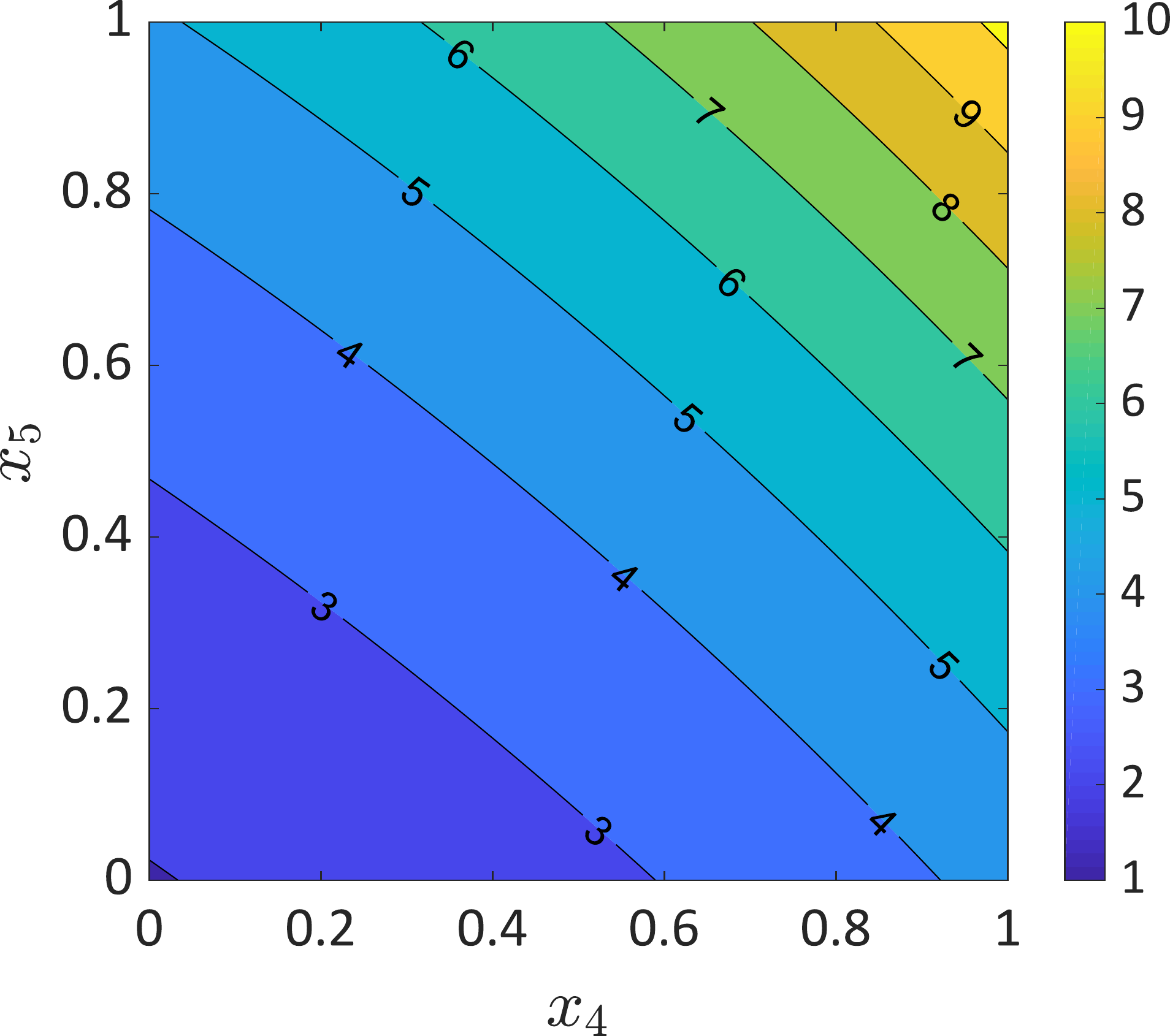}
		\caption{GLaM}
	\end{subfigure}
	\begin{subfigure}{.44\linewidth}
		\centering
		\includegraphics[height=0.82\linewidth, keepaspectratio]{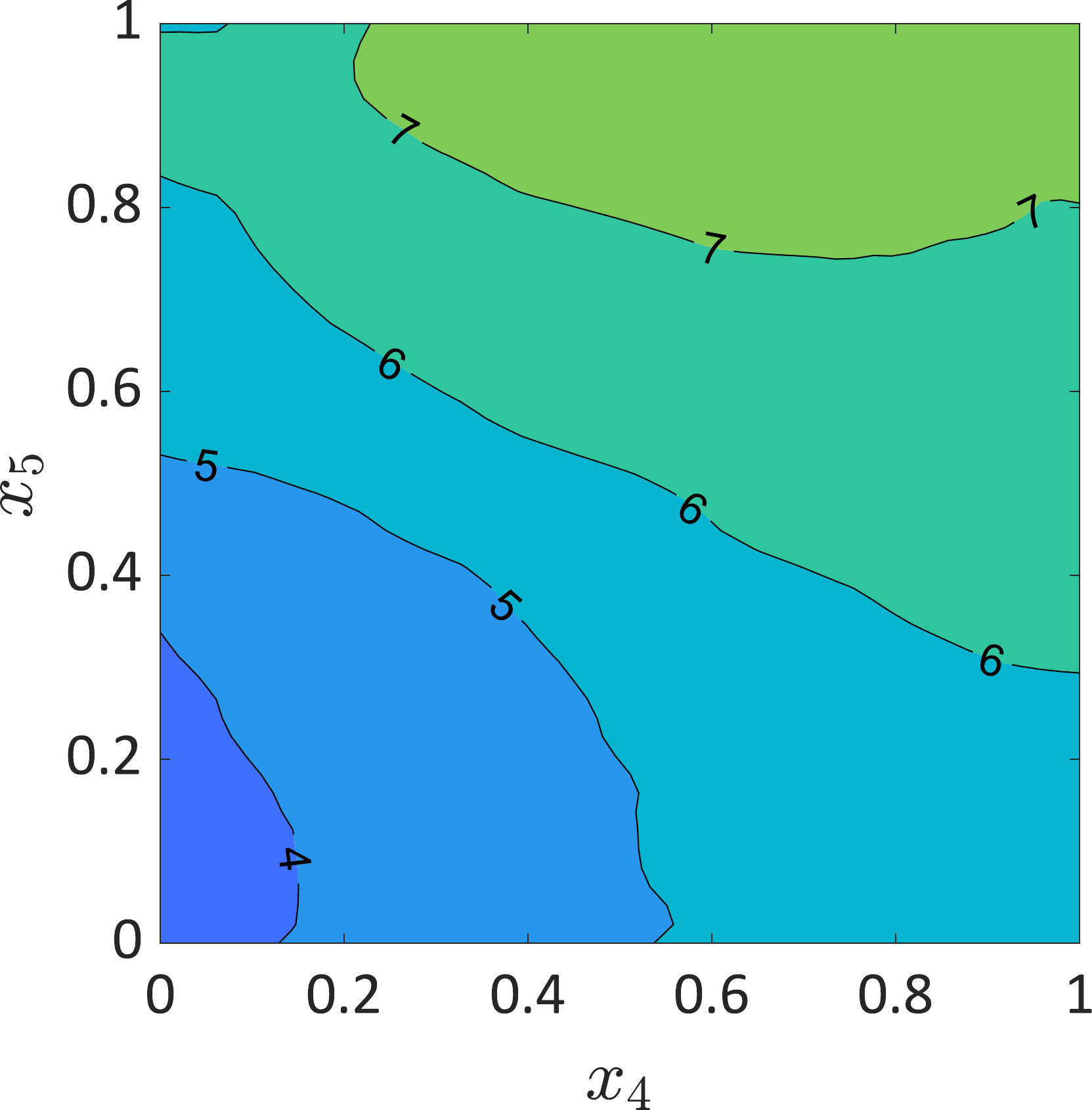}
		\caption{KCDE}
	\end{subfigure}
	\begin{subfigure}{.44\linewidth}
		\centering
		\includegraphics[height=0.82\linewidth, keepaspectratio]{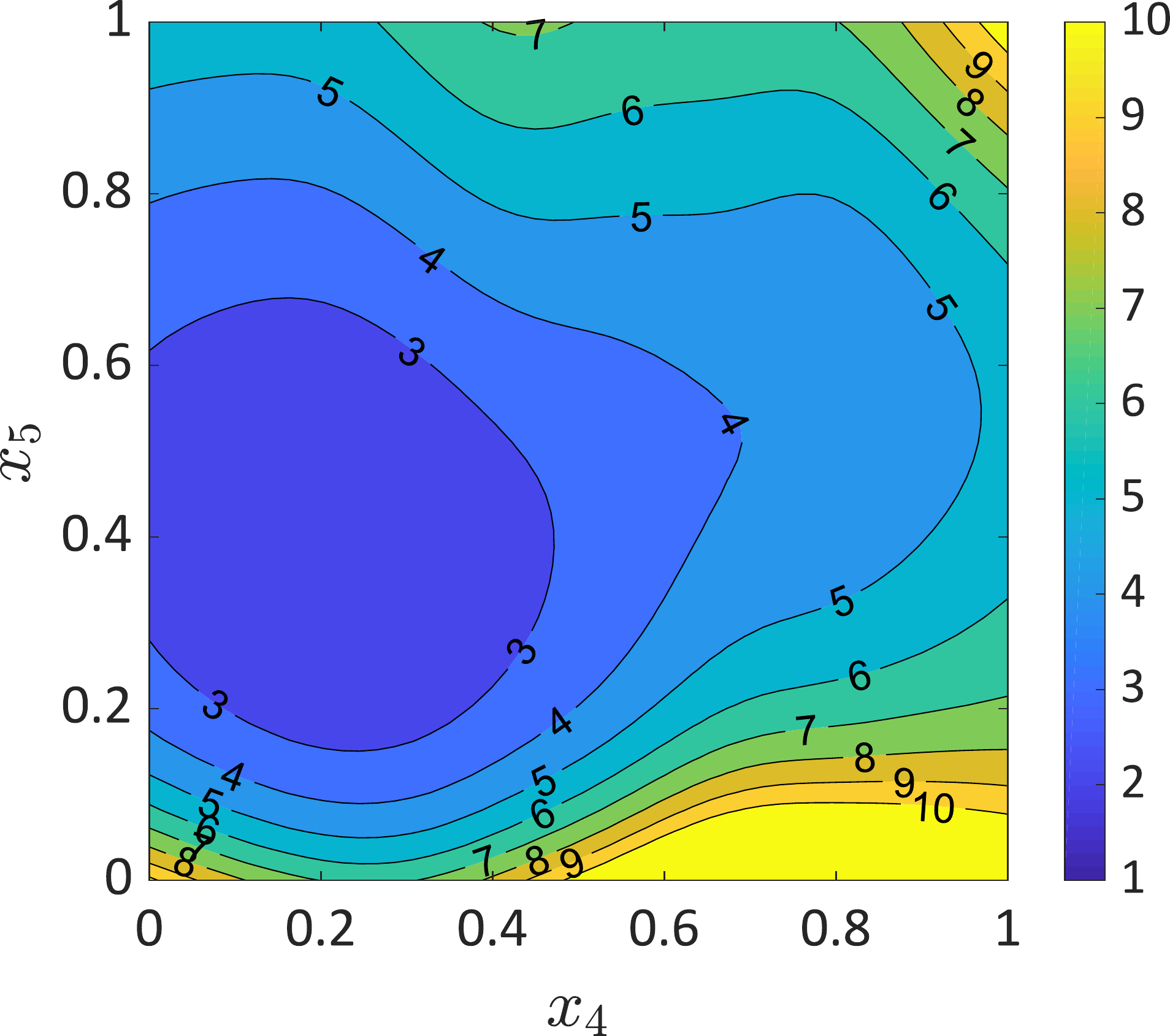}
		\caption{GP}
	\end{subfigure}
	\caption{Example 2 --- Comparisons of the variance function estimation in the plan $x_4-x_5$ with all the other input fixed at their expected value. The surrogate models are fitted to an ED with $N=1{,}000$.}
	\label{fig:normv}
\end{figure}
\par
\Cref{fig:normpdf} compares the model response PDFs (with different variances) 
for four input values with those predicted by a GLaM and a KCDE built upon 
$1{,}000$ model runs. The results show that the GLaM correctly identifies the 
shape of the underlying normal distribution among all possible shapes of the 
GLD. Moreover, it yields a better approximation to the reference PDF, whereas 
KCDE tends to ``wiggle'' in \Cref{fig:normpdf4} (high variance) and 
overestimate the spread in \Cref{fig:normpdf1} (low variance). \Cref{fig:normm,fig:normv} illustrate the mean and variance function predicted by the GLaM, KCDE, and GP in the $x_4-x_5$ plan with all the other variables fixed at their expected value. The results show that the GLaM provides more accurate estimates for both functions.
\par
\begin{figure}[!htbp]
	\centering
	\includegraphics[width=.65\linewidth, keepaspectratio]{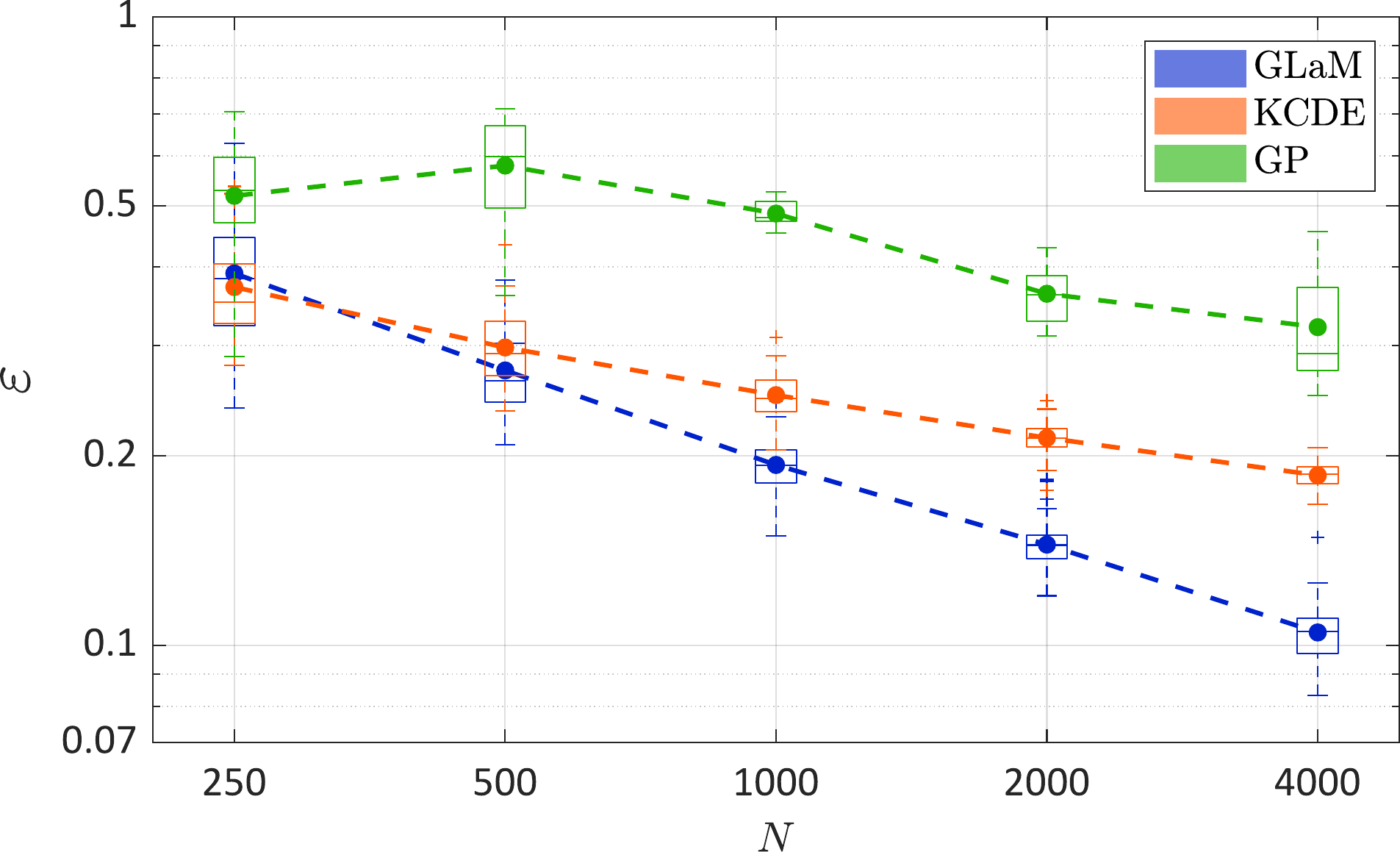}
	\caption{Example 2 --- Comparison of the convergence between GLaMs and
		KCDEs in terms of the normalized Wasserstein distance as a function of 
		the size of the experimental design. The dashed lines denote the 
		average value over 50 repetitions of the full analysis. The green box plots and associated dashed lines correspond to the errors of the heteroskedastic Gaussian Process with sequential design (10 repetitions for each size of the experimental design). The ``oracle'' normal model has an error $\varepsilon = 0$ that is not plotted here.}
	\label{fig:norm_WS}
\end{figure}
\par
Similar to the first example, we perform a convergence study for $N 
\in \{250;\allowbreak 500;\allowbreak 1{,}000; \allowbreak 2{,}000; 
\allowbreak4{,}000\}$, the results of which are 
shown in \cref{fig:norm_WS}. The underlying response distribution is Gaussian, and thus the oracle normal approximation has $\varepsilon=0$, which is not reported in the figure. Surprisingly, GP gives the worst results. This may be understood as follows: the updating criterion of the sequential design targets at minimizing the \emph{integrated mean-squared error}. The latter mainly focuses  on improving the mean estimation (as illustrated in \cref{fig:normm,fig:normv}), yet both the mean and variance contribute to the Wasserstein distance \cref{eq:WS}. Also, this example is a five-dimensional problem, which results in more parameters to estimate for GP.
In the case of small $N$, namely $N=250$, both the GLaMs and KCDEs perform poorly, with the GLaMs showing a similar average error but higher variability. This is explained as follows. Because of the use of $\AOLS$ in the modified FGLS procedure, we observe that the total number of coefficients of GLaMs to be estimated varies between 19 to 39 for $N=250$. Since the GLD is very flexible, a relatively large data set is necessary to provide enough evidence of the underlying PDF shape. Consequently, a small $N$ can lead to overfitting for high-dimensional $\ve{c}$, but good surrogates can be obtained for more parsimonious models. In contrast, KCDE always performs a thorough leave-one-out cross-validation strategy to select the bandwidths. Therefore, KCDEs show a slightly more stable estimate for $N=250$. With $N$ increasing, however, GLaMs converge much faster and outperform KCDEs for $N \geq 500$ both in terms of the mean and median of the errors. For $N\geq 1{,}000$, the average performance of GLaM is even better than the best KCDE model among the 50 repetitions.

In this example of moderate dimensionality, building a GP with sequential design is surprisingly time-consuming, especially for large experimental designs. This is probably due to the sequential design of experiments, which adds new points one by one and updates the surrogate after each enrichment. The associated simulations were performed on the ETH Euler cluster, and the average CPU time varied from 463 seconds for $N=250$ to over 9 days for $N=4{,}000$ to build a single GP. For KCDE, it took about 20 CPU seconds for $N=250$ up to $30$ minutes for $N=4{,}000$ on a standard laptop. In comparison, constructing a GLaM is always on the order of seconds: around 8 seconds for both $N=250$ and $N=4{,}000$ on a standard laptop.


\subsection{Effect of replications}
As pointed in \Cref{rem:rep}, the proposed method can also work with a data set containing replicates. The latter are simply treated as separate points in the ED. In this section, we analyze the effect of replications using the previous two analytical examples. To this end, we generate data by replicating $R\in\acc{5;10;25;50}$ for each set of input parameters in the ED. We keep the total number of simulations the same as nonreplicated cases by reducing the size of the ED accordingly. For instance, a data set of total $N=1{,}000$ model evaluations with 10 replications consists of 100 different sets of input parameters, each of which is simulated 10 times.

For quantitative comparisons, we investigate a convergence study similar to \Cref{sec:ex1,sec:ex2}: the total number of runs $N$ varies in $\acc{250;500;1{,}000;2{,}000;4{,}000}$, and each scenario is repeated 50 times.

\Cref{fig:GBM_rep,fig:norm_rep} summarize the error defined in \cref{eq:Rlevel1} averaged over the 50 repetitions for each $R\in\acc{5;10;25;50}$. In the first example, replications do not have a strong effect for $R \in \acc{5;10;25}$. This is because the expansions for $\ve{\lambda}(\ve{x})$ contain only a few terms. Therefore, as long as we have enough ED points, exploring the input space and performing replications bring similar improvements to the surrogate accuracy. However, a large number of replications, i.e., $R=50$, gives too few ED points for small values of $N$, which yields GLaMs of poor performance. 

In the second example, we observe a clear negative effect of replications: for the same total amount of model runs, the surrogate quality deteriorates when increasing the number of replications / decreasing the size of the experimental design. 

In summary, homogeneous replications (i.e., those with the same number of replicates for each point of the experimental design) do not necessarily bring additional accuracy and may even lead to a ``waste'' of computational budget for the proposed GLaM method. Nevertheless, this does not imply that replications are always useless. On the one hand, for methods that explore the usage of replications, there is a trade-off between replications and exploration \cite{Binois2019}. On the other hand, an adaptive selection of different numbers of replications for each point in the experimental design could possibly improve the performance of the proposed method. However, unlike the heteroskedastic GP, GLaM not only estimates the mean and the variance but also produces the whole PDF. As a result, sequential design strategies for building GLaMs remain to be developed in future study and are outside the scope of the paper.
\begin{figure}[!htbp]
	\centering
	\includegraphics[width=.65\linewidth, keepaspectratio]{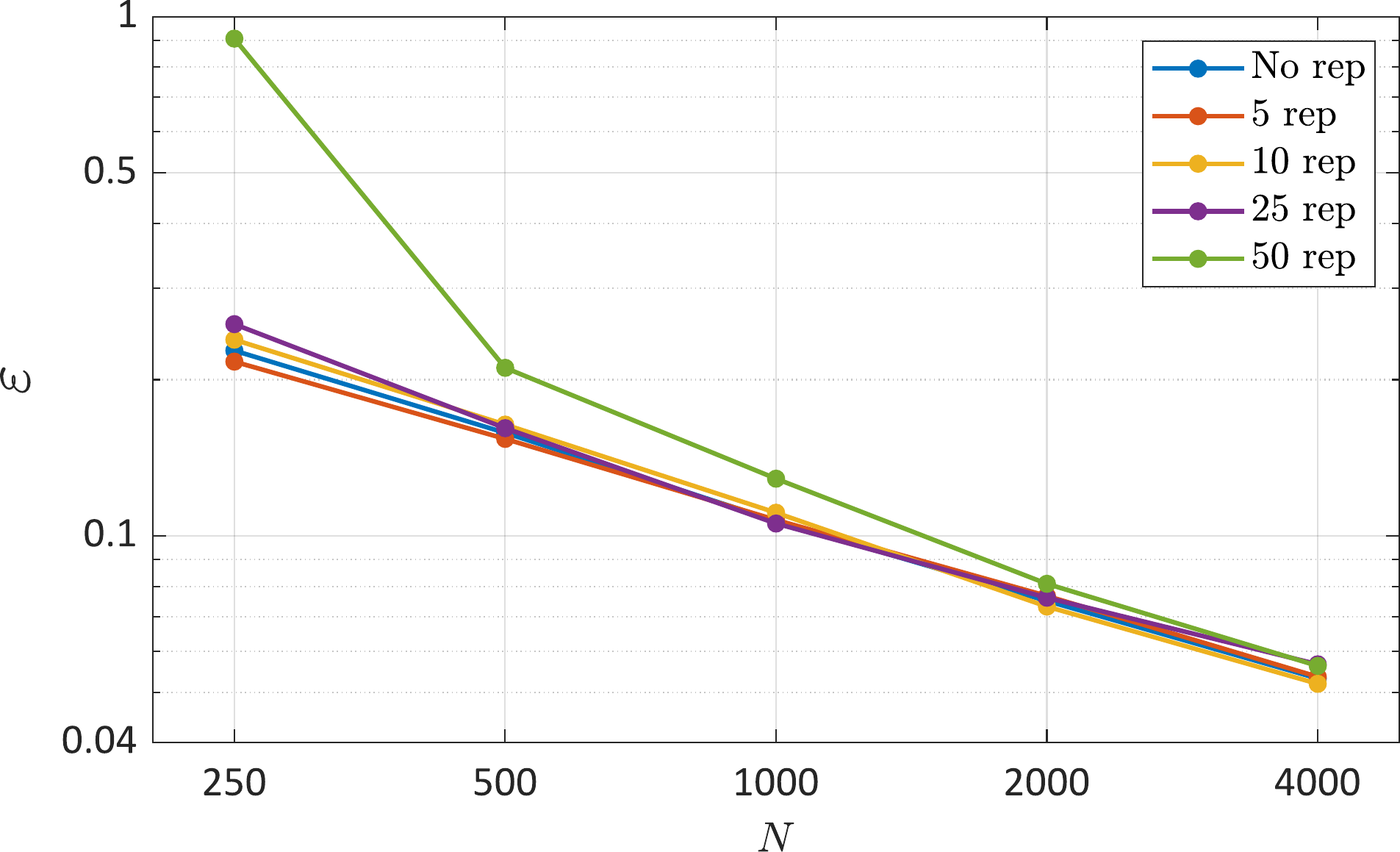}
	\caption{Example 1 --- Comparison of the GLaMs built on data with different number of replications. The curves corresponds to the mean error over the 50 repetitions.}
	\label{fig:GBM_rep}
\end{figure}

\begin{figure}[!htbp]
	\centering
	\includegraphics[width=.65\linewidth, keepaspectratio]{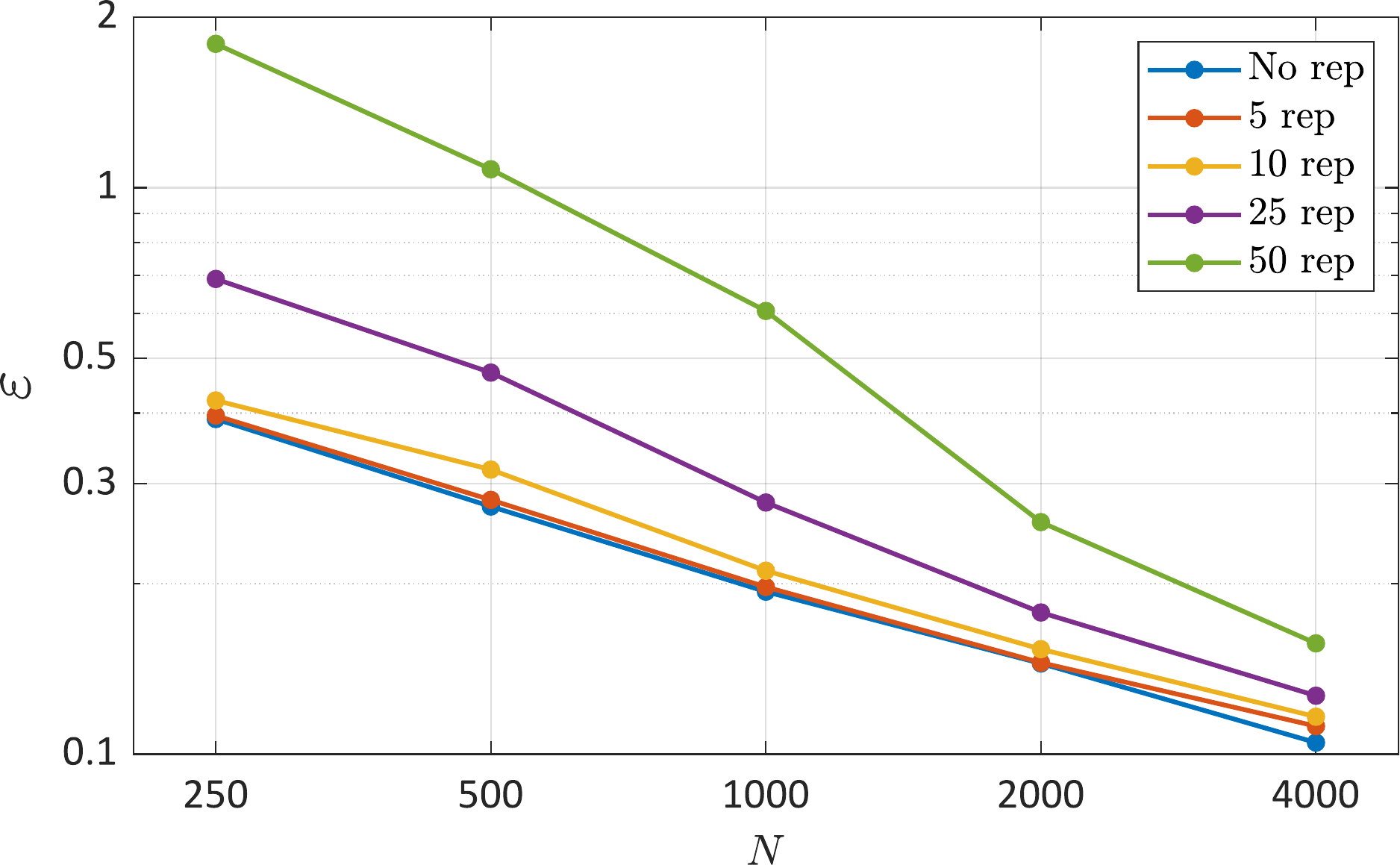}
	\caption{Example 2 --- Comparison of the GLaMs built on data with different number of replications. The curves corresponds to the mean error over the 50 repetitions.}
	\label{fig:norm_rep}
\end{figure}

\subsection{Example 3: Asian options}
\label{sec:app1}
In this third example, we apply the proposed method to a financial case study, namely an \emph{Asian option} \cite{Kemna1990}. Such an option, a.k.a. average value option, is a derivative contract, the payoff of which is contingent on the average price of the underlying asset over a certain fixed time period. Due to the path-dependent nature, an Asian option has complex behavior, and its valuation is not straightforward, as opposed to European options. 
\par
Recall the Black--Scholes model defined in \cref{eq:GBM} that represents the evolution of a stock price $S_t(\ve{x})$. Instead of relying on the stock price on the maturity date $t=T$, the payoff of an Asian call option reads
\begin{equation}\label{eq:V-A}
	C(\ve{x}) = \max\acc{A_T(\ve{x})-K,0}, \text{ with } A_t(\ve{x}) = \frac{1}{t} \int_{0}^{t} S_u(\ve{x}) \D u.
\end{equation}
where $A_t(\ve{x})$ is called the \emph{continuous average process}, and $K$ denotes the \emph{strike price}. Because $A_T(\ve{x})$ plays an important role in the Asian option modeling \cref{eq:V-A}, the PDF of $A_T(\ve{x})$ is of interest in this case study. As in \Cref{sec:ex1}, we set $T=1$, which corresponds to a one-year inspection period. We choose $X_1 \sim \cu(0,0.1)$ and $X_2\sim \cu(0.1,0.4)$ for the two input random variables. Unlike $S_1(\ve{x})$, the distribution of $A_1(\ve{x})$ cannot be derived 
analytically. It is necessary to simulate the trajectory of $S_t(\ve{x})$ to compute $A_1(\ve{x})$. Based on the Markovian and lognormal properties of $S_t(\ve{x})$, we apply the following recursive equations for the path simulation with a time step $\Delta t = 0.001$:
\begin{equation*}
	\begin{split}
		S_0(\ve{x}) &= 1, \\
		S_{t+\Delta t}(\ve{x}) \mid S_t(\ve{x}) &\sim \cl\cn \left( 
		\log\left(S_t(\ve{x})\right) + \left(x_1 - 
		\frac{x^2_2}{2}\right)\Delta t, x_2\sqrt{\Delta t}\right).
	\end{split}
\end{equation*}
Finally, the continuous average defined in \cref{eq:V-A} is approximated by the 
arithmetic mean, that is,
\begin{equation*}
	A_1(\ve{x}) = \frac{\sum_{k=1}^{1{,}000}S_{k\Delta t}(\ve{x})}{1{,}000}
\end{equation*}
\par
\begin{figure}[!htbp]
	\centering
	\begin{subfigure}[b]{0.45\textwidth}
		\centering
		\includegraphics[height=0.7\linewidth, keepaspectratio]{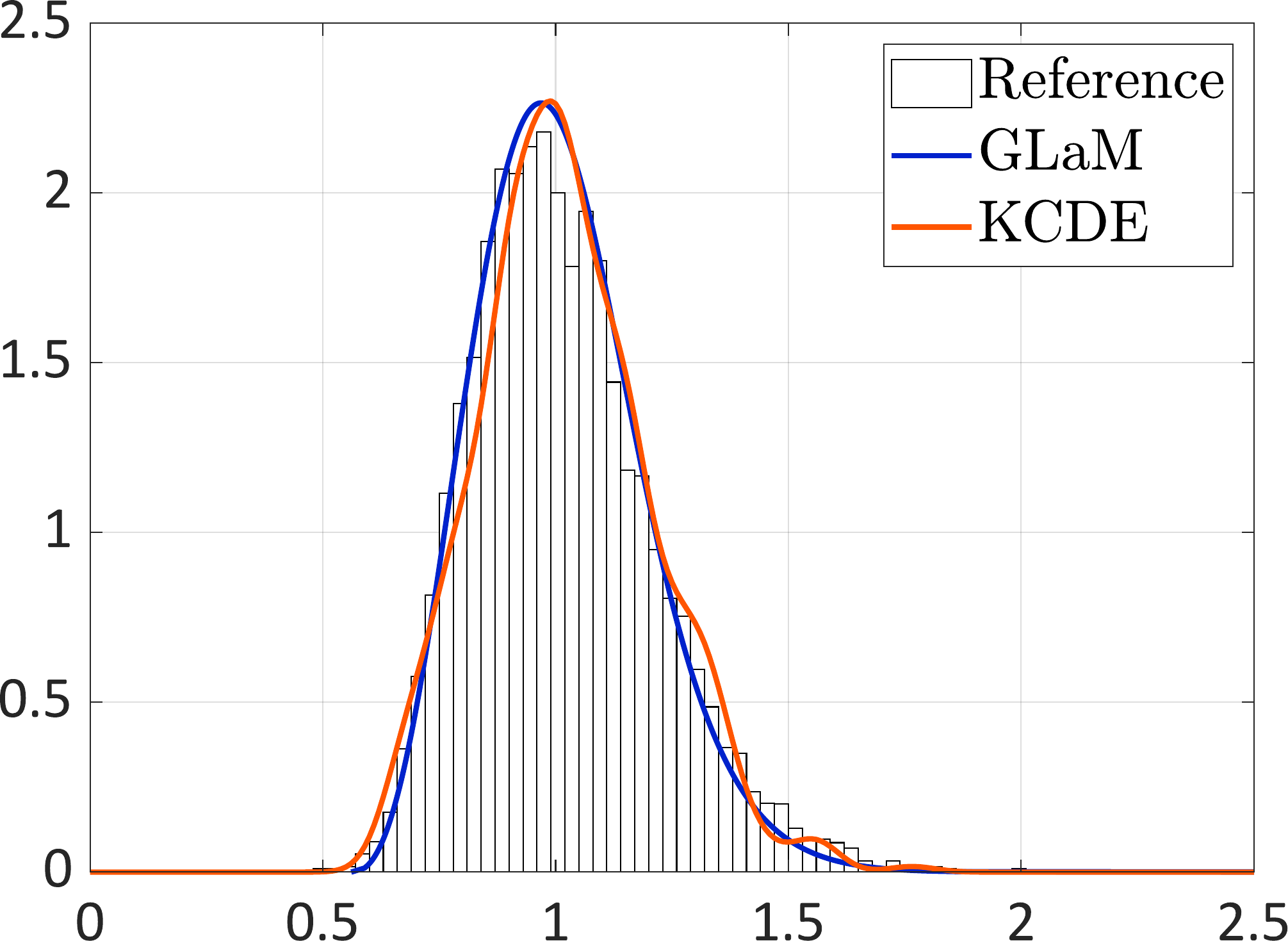}
		\caption{PDF for $\ve{x} = (0.03,0.33)^T$}
		\label{fig:asianpdf1}
	\end{subfigure}
	\hspace{0.5cm}
	\begin{subfigure}[b]{0.45\textwidth}
		\centering
		\includegraphics[height=0.7\linewidth, keepaspectratio]{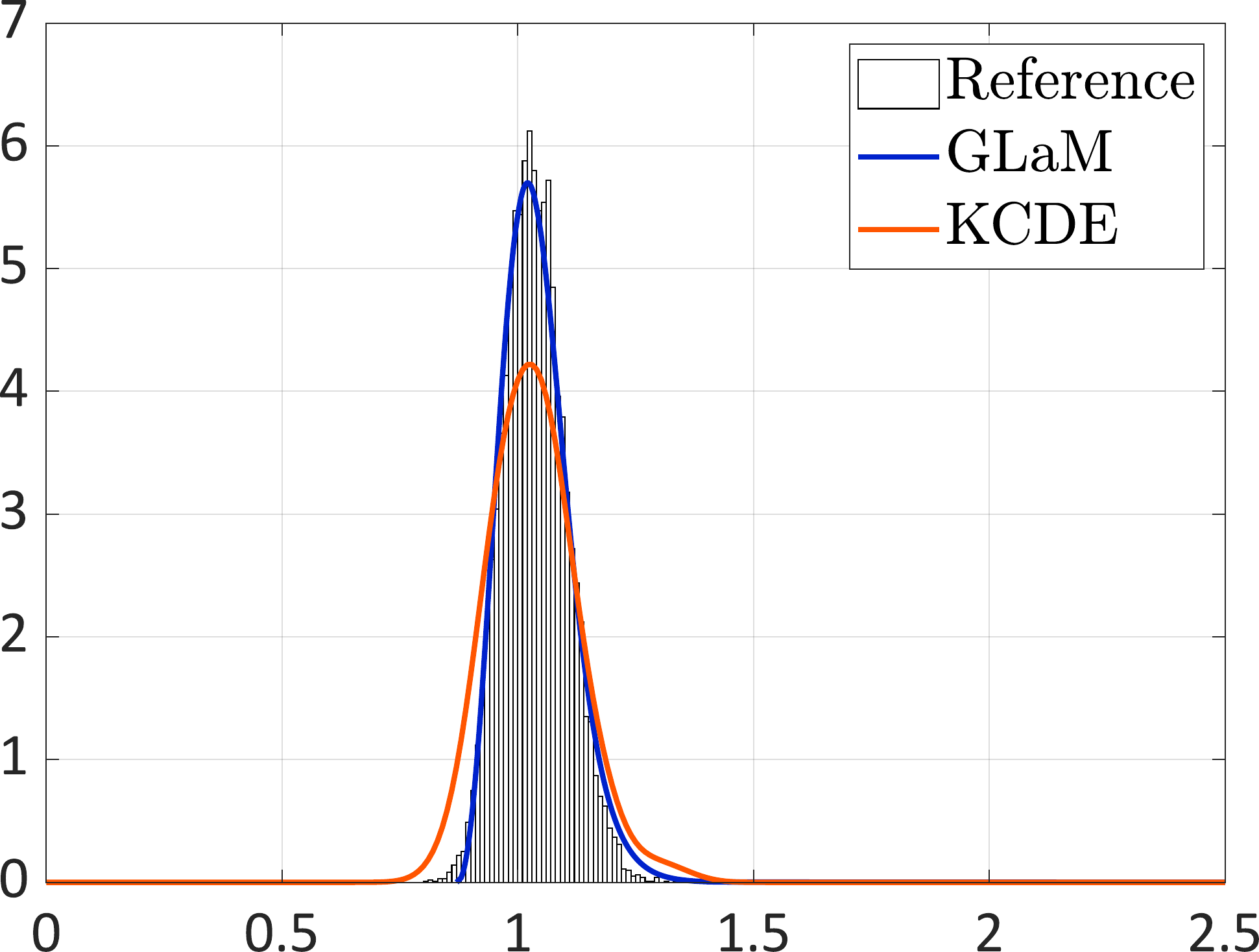}
		\caption{PDF for $\ve{x} = (0.07,0.11)^T$}
		\label{fig:asianpdf2}
	\end{subfigure}
	\caption{Asian option --- Comparisons of the emulated PDF, $N=500$}
	\label{fig:asianpdf}
\end{figure}
\begin{figure}[!htbp]
	\centering
	\begin{subfigure}{.33\linewidth}
		\centering
		\includegraphics[height=0.84\linewidth, keepaspectratio]{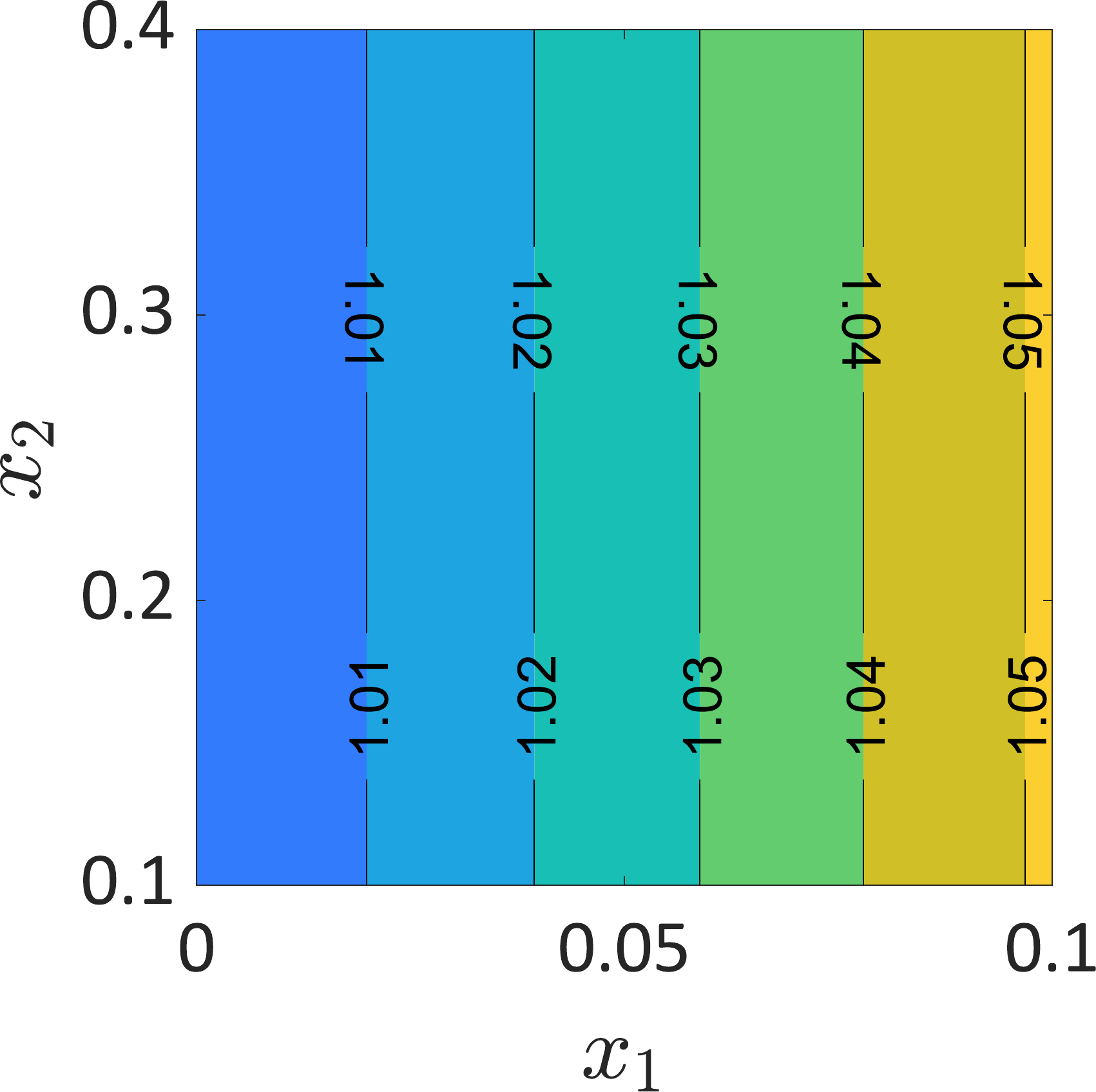}
		\caption{Reference}
	\end{subfigure}
	\hspace*{-6mm}
	\begin{subfigure}{.33\linewidth}
		\centering
		\includegraphics[height=0.84\linewidth, keepaspectratio]{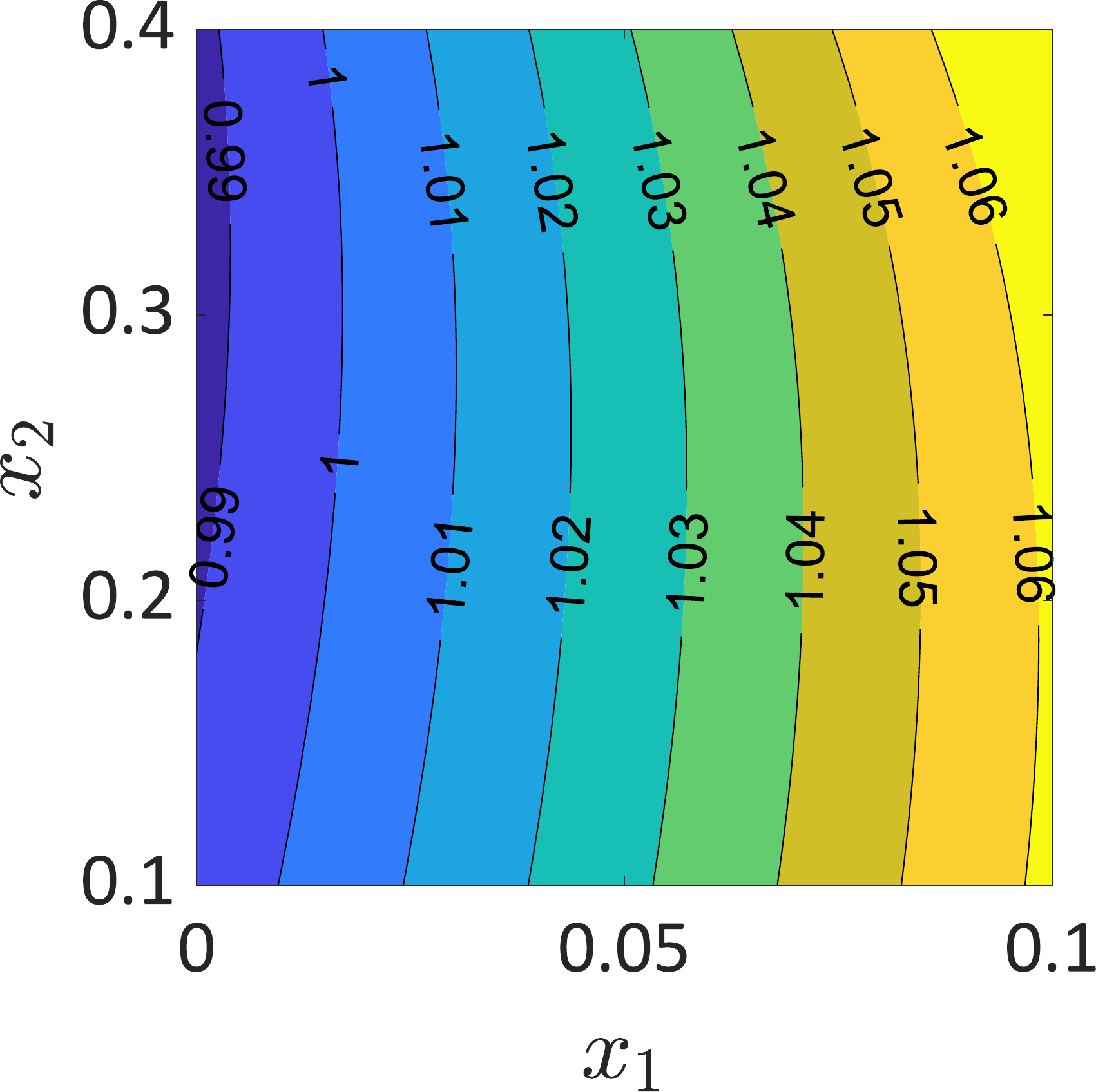}
		\caption{GLaM}
	\end{subfigure}
	\hspace*{-3mm}
	\begin{subfigure}{.33\linewidth}
		\centering
		\includegraphics[height=0.84\linewidth, keepaspectratio]{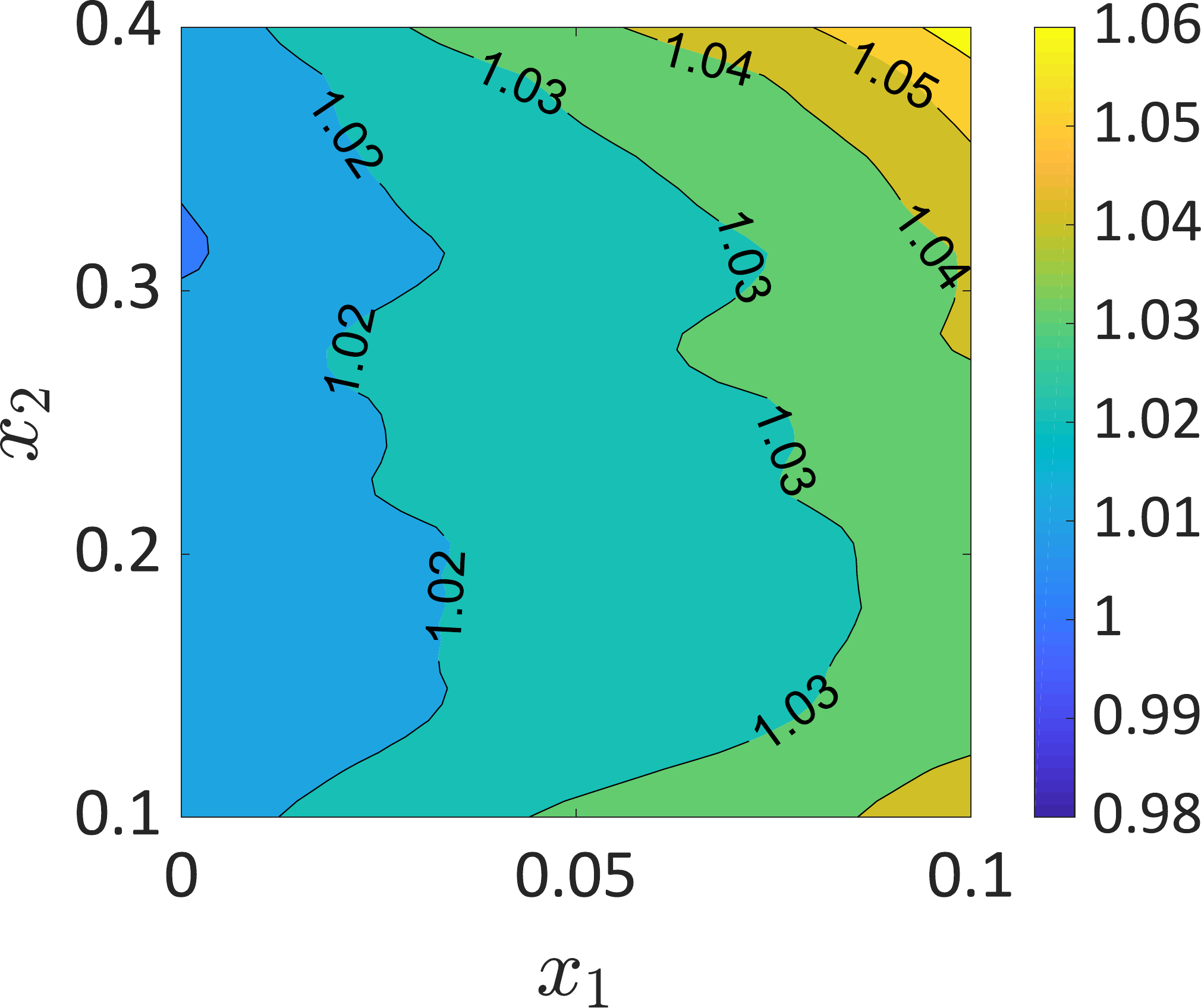}
		\caption{KCDE}
	\end{subfigure}
	\caption{Asian option --- Comparisons of the mean function estimation, $N=500$.}
	\label{fig:asianm}
\end{figure}
\begin{figure}[!htbp]
	\centering
	\begin{subfigure}{.33\linewidth}
		\centering
		\includegraphics[height=0.82\linewidth, keepaspectratio]{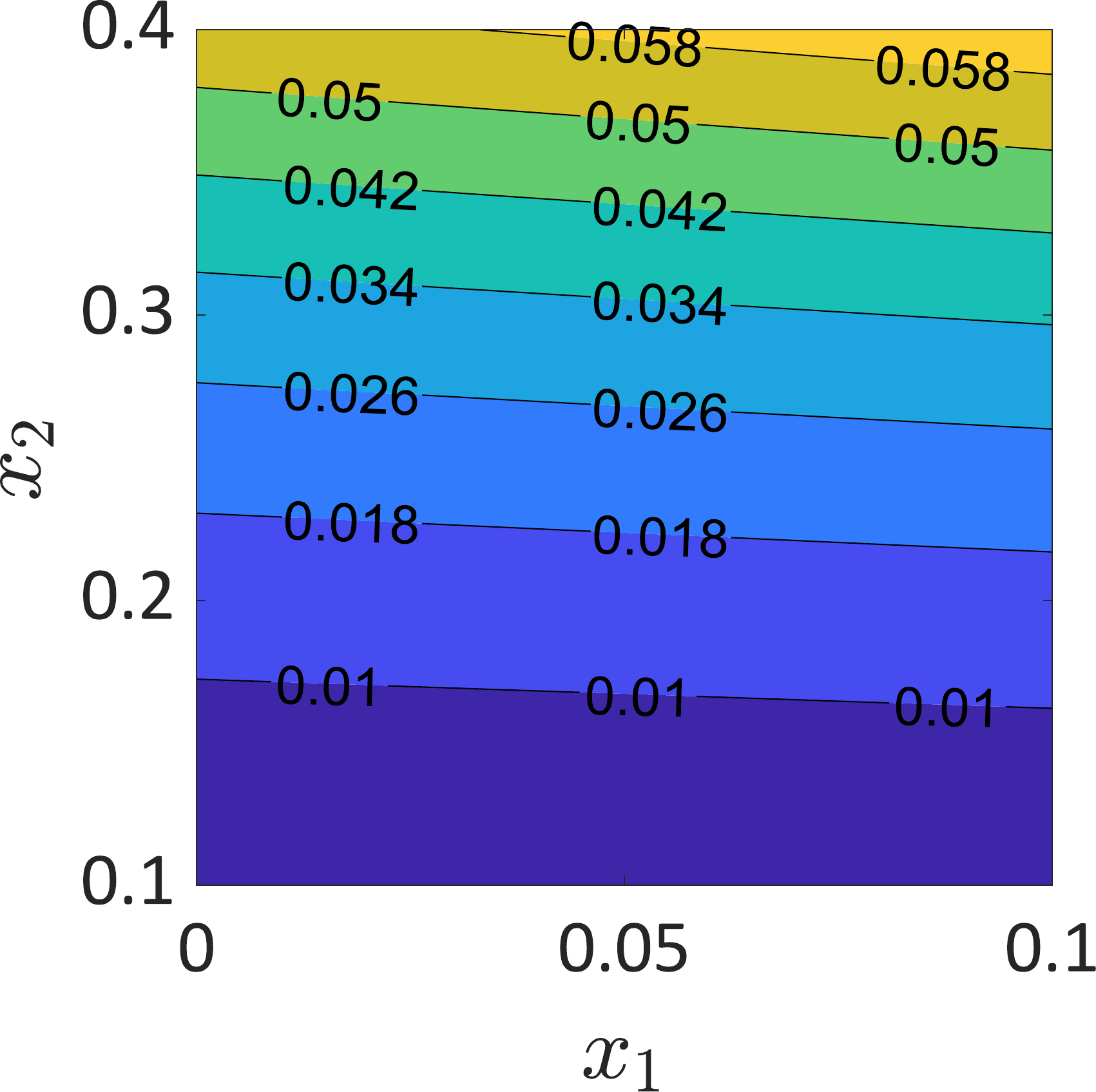}
		\caption{Reference}
	\end{subfigure}
	\hspace*{-6mm}
	\begin{subfigure}{.33\linewidth}
		\centering
		\includegraphics[height=0.82\linewidth, keepaspectratio]{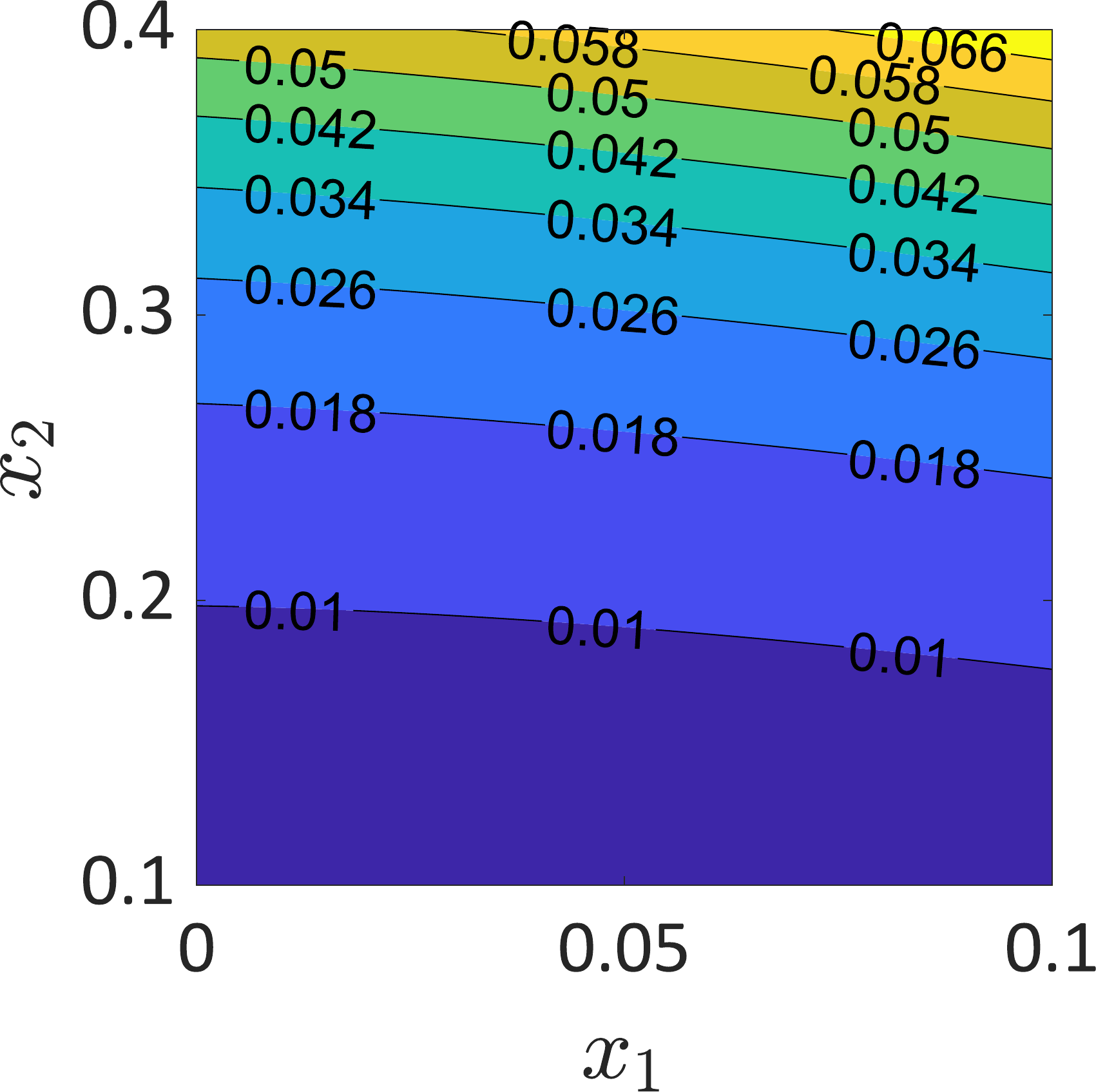}
		\caption{GLaM}
	\end{subfigure}
	\hspace*{-3mm}
	\begin{subfigure}{.33\linewidth}
		\centering
		\includegraphics[height=0.82\linewidth, keepaspectratio]{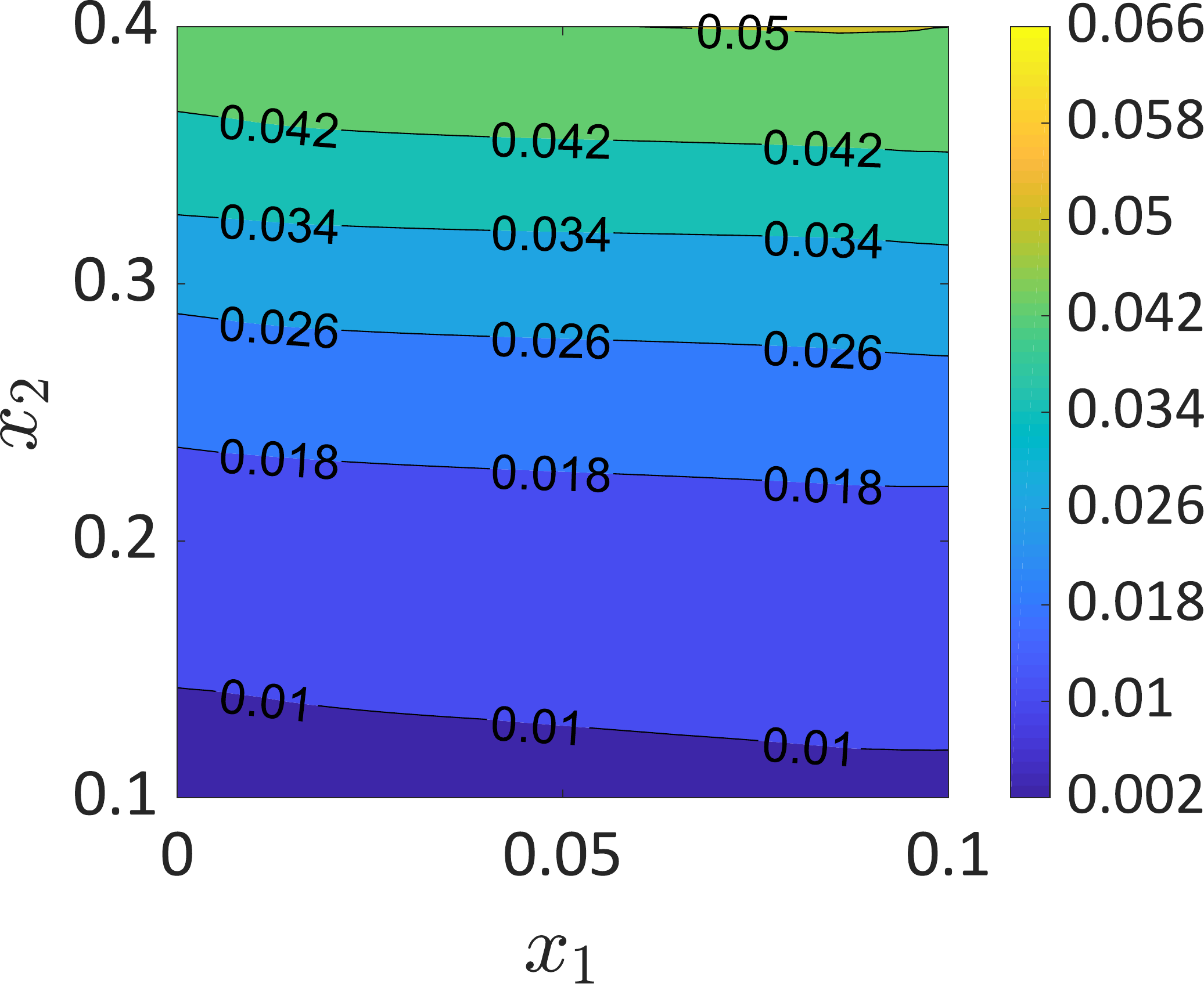}
		\caption{KCDE}
	\end{subfigure}
	\caption{Asian option --- Comparisons of the variance function estimation, $N=500$.}
	\label{fig:asianv}
\end{figure}
\par
\Cref{fig:asianpdf} shows two response PDFs predicted by the two surrogate models constructed on an experimental design of $N=500$. The reference histograms are calculated from $10^4$ repeated runs of the simulator for each set of input parameters. We observe that the KCDE exhibits slight fluctuations at the right tail for high volatility (in \Cref{fig:asianpdf1}) and does not well approximate the bulk of the response distribution for low volatility (in \Cref{fig:asianpdf2}). In comparison, the GLaM can well represent the PDF shape in both cases and also more accurately approximates the tails. \Cref{fig:asianm,fig:asianv} shows the mean and variance function, where the reference values can be obtained by applying It\^o's calculus. For the experimental design of $N=500$, the GLaM more accurately predicts the two functions. Finally, quantitative comparisons in \Cref{fig:Asian_WS} confirm the superiority of GLaMs to KCDEs: GLaMs yield smaller average error for all $N \in \acc{250;500;1{,}000;2{,}000;4{,}000}$ and demonstrate a better convergence rate. Moreover, for large experimental designs ($N \geq 2{,}000$), the average error of GLaMs is nearly half of that of KCDEs. The oracle Gaussian approximation in this case study has a similar error to GLaMs built on $1{,}000$ model runs. For $N\geq 2{,}000$, GLaMs fitted from data are much more accurate than the best possible Gaussian-type mean-variance model.
\par
\begin{figure}[!htbp]
	\centering
	\includegraphics[width=.65\linewidth, keepaspectratio]{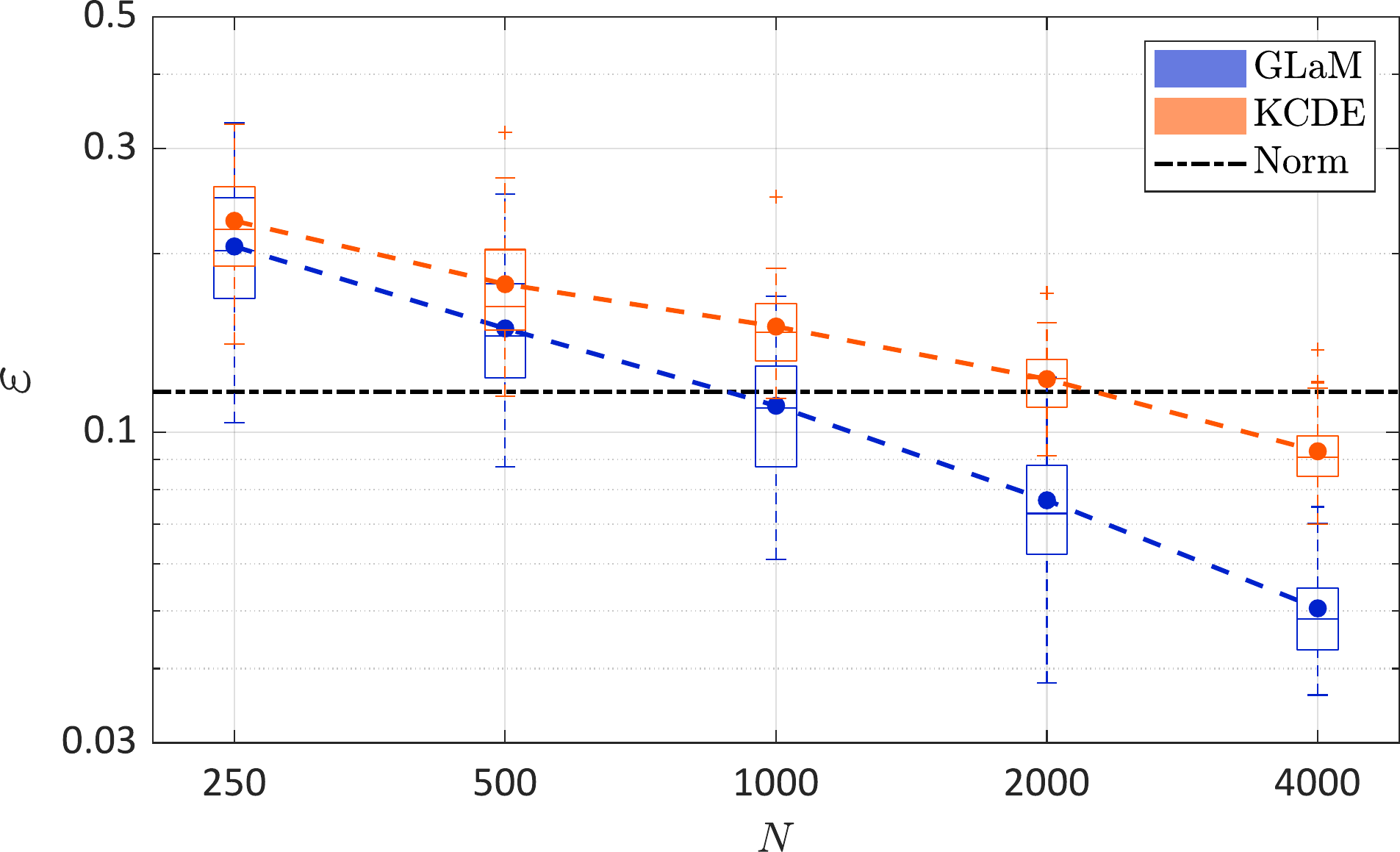}
	\caption{Asian option, average process $A_1(\ve{x})$ at $T=1$~year --- 
		Comparison of the convergence of GLaMs and KCDEs in terms of the 
		normalized Wasserstein distance as a function of the size of the 
		experimental design. The dashed lines denote the average value over $50$ 
		repetitions of the full analysis. The black dash-dotted line 
		represents the error of the model assuming that the response distribution 
		is normal with the true mean and variance}
	\label{fig:Asian_WS}
\end{figure} 
\par
As a second quantity of interest, we consider the expected payoff $\mu_C(\ve{x})=\Esp{C(\ve{x})}$. This quantity not only is important for making investment decisions but also has a very similar form to the option price \cite{Kemna1990}. The definition \cref{eq:V-A} implies that the payoff $C(\ve{x})$ is a mixed random variable, which has a probability mass at $0$ and a continuous PDF on the positive line depending on the strike price $K$. In the following analysis, $K$ is set to 1. 
\par
For GLaMs, $\mu_C(\ve{x})$ can be calculated by
\begin{equation}\label{eq:epoff}
	\mu_C(\ve{x}) = \left(\lambda_1 - \frac{1}{\lambda_2 \lambda_3} +	\frac{1}{\lambda_2 \lambda_4} - K\right)\,(1-u_K) + 
	\frac{1}{\lambda_2}\,\left(
	\frac{1-u_K^{\lambda_3+1}}{\lambda_3\,\left(\lambda_3+1\right)}
	-\frac{(1-u_K)^{\lambda_4+1}}{\lambda_4\,\left(\lambda_4+1\right)}
	\right)
\end{equation}
where $\lambda$'s are the distribution parameters at $\ve{x}$, 
and $u_K$ is the solution of the nonlinear equation
\begin{equation}
	Q(u_K;\ve{\lambda}) = K.
\end{equation}
with $Q$ being the quantile function defined in \cref{eq:FKML}.
\par
\Cref{fig:Asian_payoff} shows the convergence of estimations of $\mu_C(\ve{x})$ in terms of the error defined in \cref{eq:Berror}. The difference between the performance of GLaMs and KCDEs is not as significant as for the distribution estimation of $A_1(\ve{x})$ in \Cref{fig:Asian_WS}. For relatively small data sets, namely $N\leq 500$, both models work poorly: they are only able to explain on average no more than 70\% of the variance of $\mu_C(\ve{X})$. In addition, GLaMs demonstrate a higher variability of the errors. For larger experimental designs $N\geq 2{,}000$, however, the performance of GLaMs improves significantly more than that of KCDEs. For $N=4{,}000$, the average error of GLaMs is twice smaller than that of KCDEs, and the smallest error achieved by GLaMs is one order of magnitude smaller than the best KCDE.
\par
\begin{figure}[!htbp]
	\centering
	\includegraphics[width=.65\linewidth, keepaspectratio]{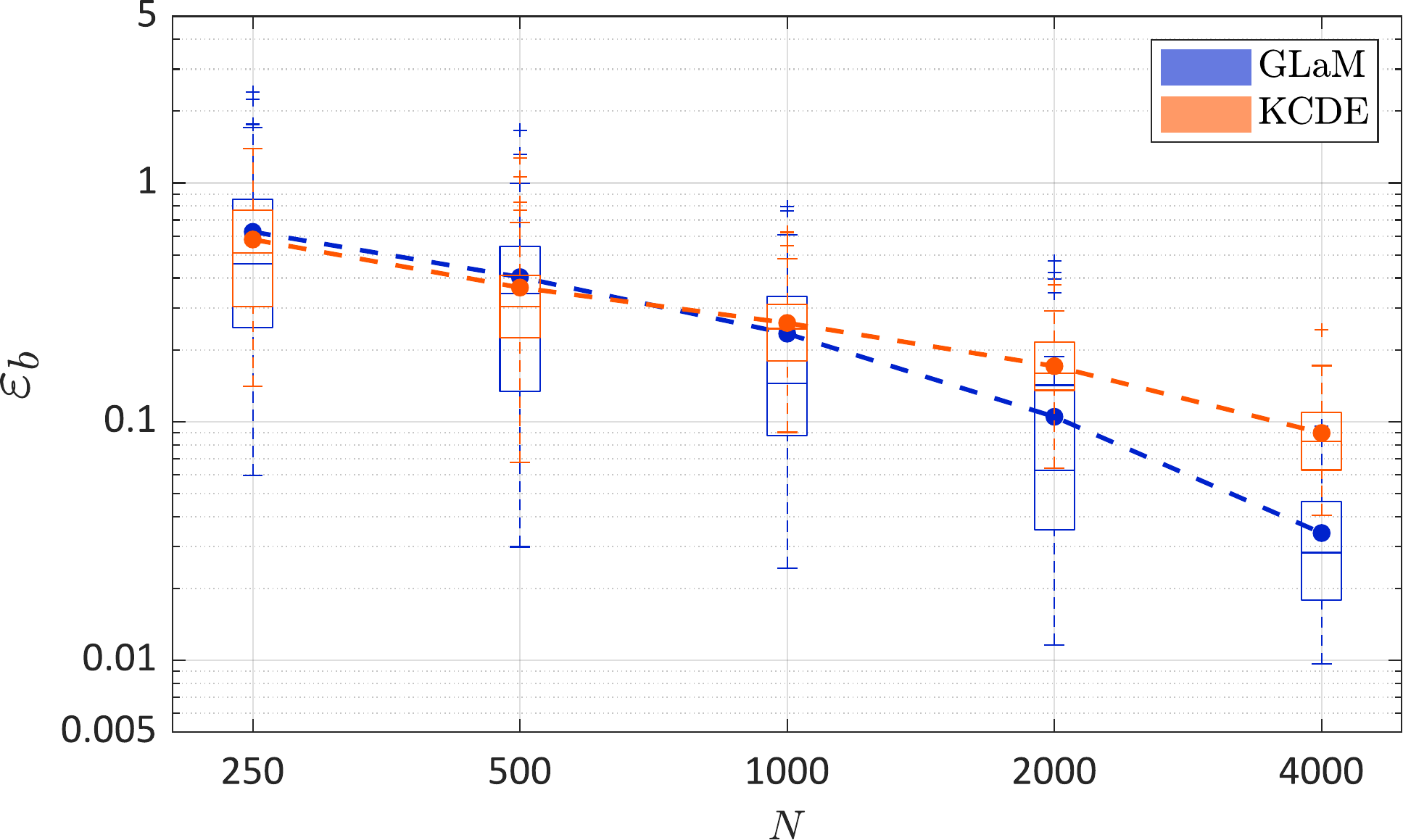}
	\caption{Asian option, expected payoff estimations --- Comparison of the 
		convergence of GLaMs and KCDEs in terms of the normalized mean squared 
		error as a function of the size of the experimental design. The dashed 
		lines denote the average value over $50$ repetitions of the full analysis.}
	\label{fig:Asian_payoff}
\end{figure}

\subsection{Example 4: Stochastic SIR model}
\label{sec:app2}

In this fourth example, we apply the proposed method to a \emph{stochastic 
	susceptible-infected-recovered} (SIR) model in epidemiology \cite{Britton2010}. 
This model simulates the spread of an infectious disease, which can help find 
appropriate epidemiological interventions to minimize social and ethical 
impacts during the outbreak.
\par
According to the standard SIR model, at time $t$ a population of size $P_t$ 
contains three groups of individuals: susceptible, infected, and recovered, the 
counts of which are denoted by $S_t$, $I_t$, and $R_t$, respectively. These 
three quantities fully characterize a population configuration at time $t$. 
Among the three groups, only susceptible individuals can get infected due to 
close contact with infected individuals, whereas an infected person can recover 
and becomes immune to future infections. We consider a fixed population without 
newborns and deaths, \ie the total population size is constant, $P_t = P$. As a 
result, $S_t$, $I_t$, and $R_t$ satisfy the constraint $S_t+I_t+R_t = P$, and 
only the time evolution of $(S_t,I_t)$ is necessary to characterize the 
spread of a disease. 
\par
To account for random recoveries and interactions among individuals, stochastic 
SIR models are usually preferred to represent the epidemic evolution. Without 
going into details, the model dynamics is briefly summarized as follows. The 
pair $\left(I_t,S_t\right)$ evolves as a continuous-time Markov 
process following mutual transition rates $\beta$ and $\gamma$, 
which denote the contact rate and recovery rate, respectively. The epidemic 
stops 
at time $t = T$ where $I_T=0$, indicating that no further infections can occur. 
The evolution process is simulated by the \emph{Gillespie algorithm} 
\cite{Gillespie1977}. The reader is referred to \cite{Britton2010} for a more 
detailed presentation of stochastic SIR models.
\par
In this case study, we set the total population equal to $P = 2{,}000$ and $\beta 
= \gamma = 0.5$ as in \cite{Binois2018}. The initial configuration 
$\ve{x}=(S_0,I_0)$ is the vector of input parameters. To account for different 
scenarios, the input variables $\ve{X}$ are modeled as $X_1 \sim 
\cu(1200,1800)$ (initial number of susceptible individuals) and $X_2 \sim 
\cu(20,200)$ (initial number of infected individuals). The QoI is the total 
number of newly infected individuals during the outbreak, \ie $Y(\ve{x}) = S_T 
- S_0$. 
\par
\begin{figure}[!htbp]
	\centering
	\begin{subfigure}[b]{.47\linewidth}
		\centering
		\includegraphics[height=0.7\linewidth, keepaspectratio]{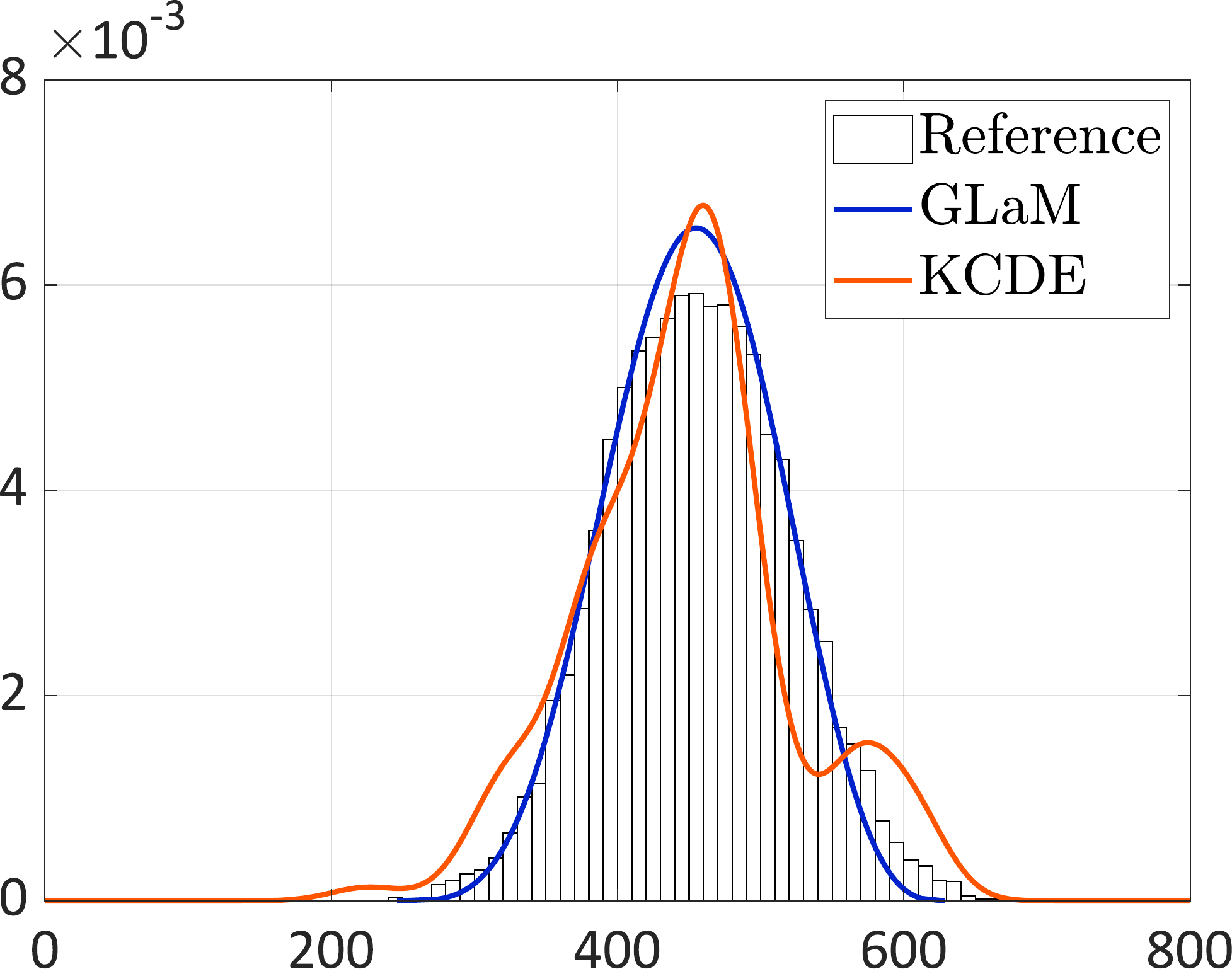}
		\caption{PDF for $\ve{x} = (1714,165)^T$}
	\end{subfigure}
	\hspace{0.5cm}
	\begin{subfigure}[b]{.47\linewidth}
		\centering
		\includegraphics[height=0.7\linewidth, keepaspectratio]{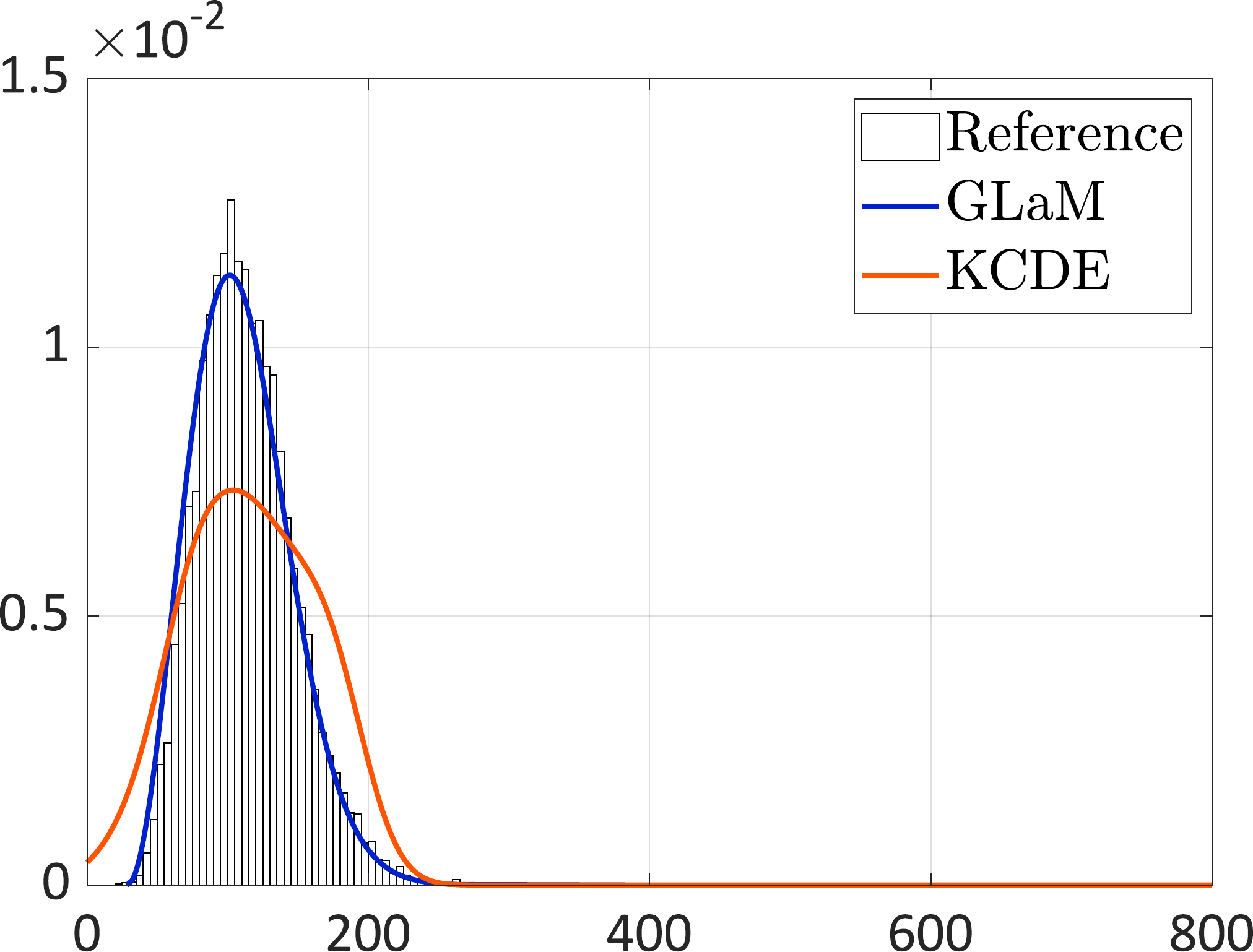}
		\caption{PDF for $\ve{x} = (1364,61)^T$}
	\end{subfigure}
	\caption{SIR model --- Comparisons of the emulated PDF, $N=500$}
	\label{fig:SIRpdf}
\end{figure}
\begin{figure}[!htbp]
	\centering
	\begin{subfigure}{.33\linewidth}
		\centering
		\includegraphics[height=0.8\linewidth, keepaspectratio]{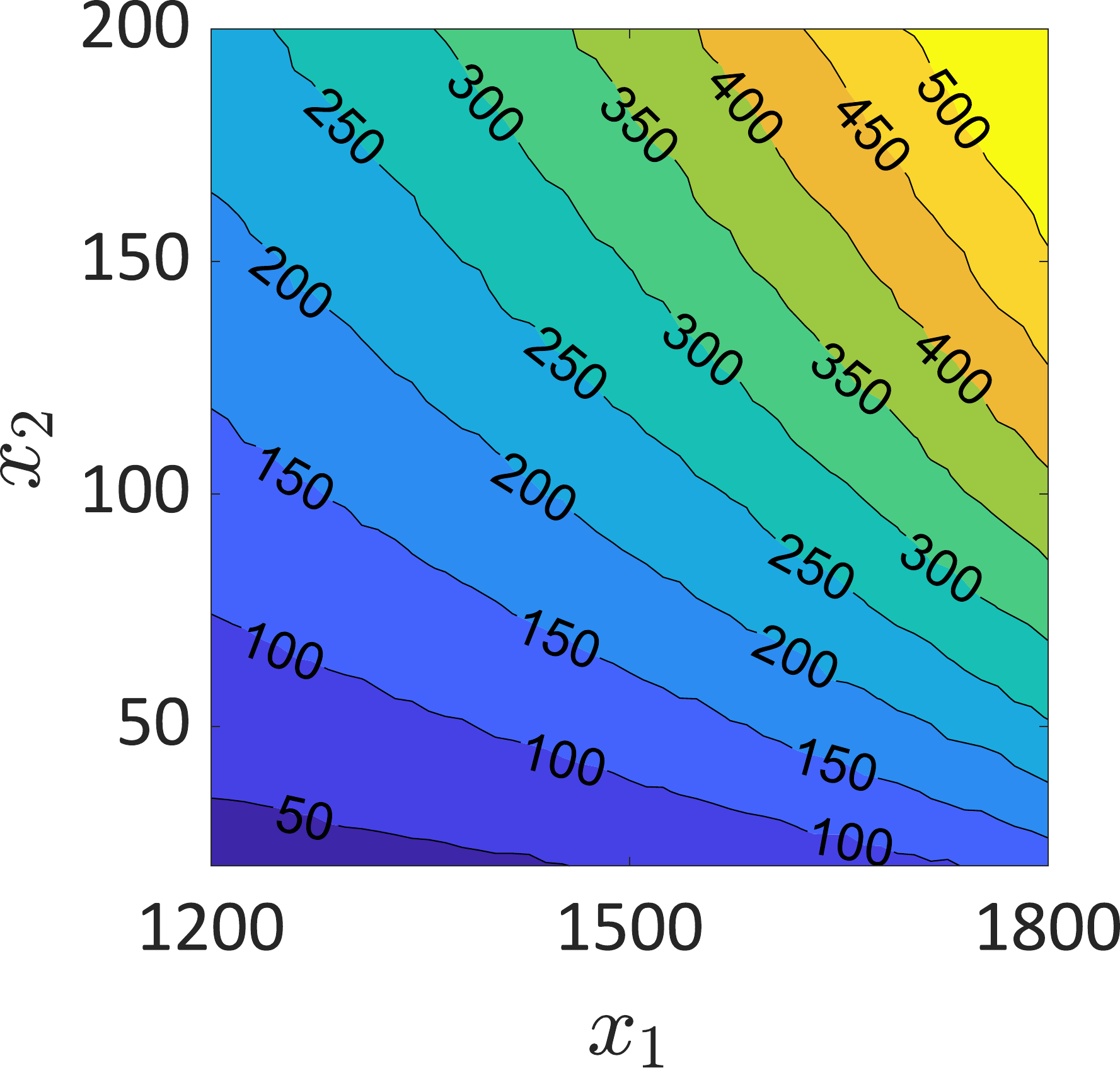}
		\caption{Reference}
	\end{subfigure}
	\hspace*{-6mm}
	\begin{subfigure}{.33\linewidth}
		\centering
		\includegraphics[height=0.8\linewidth, keepaspectratio]{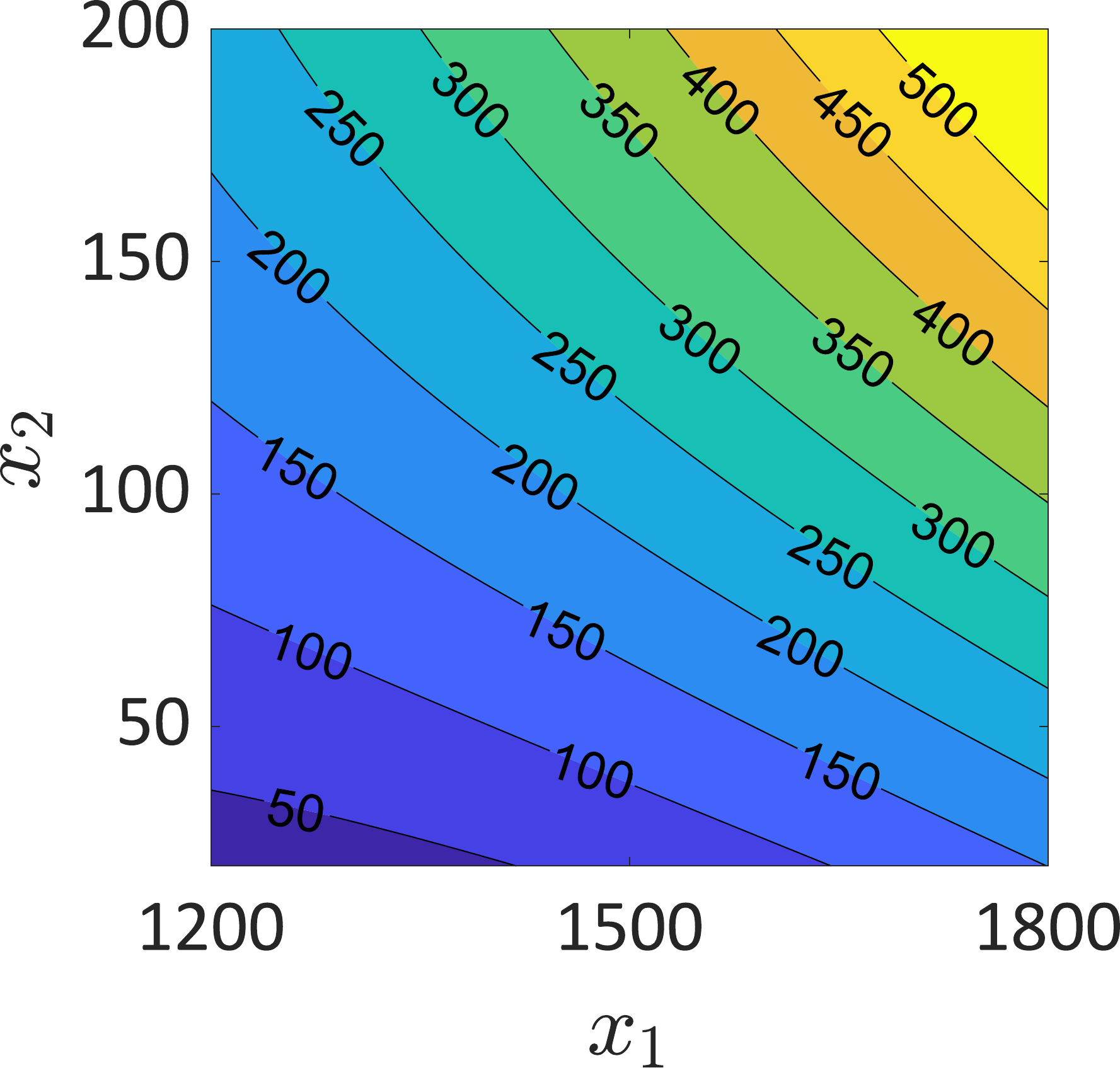}
		\caption{GLaM}
	\end{subfigure}
	\hspace*{-3mm}
	\begin{subfigure}{.33\linewidth}
		\centering
		\includegraphics[height=0.8\linewidth, keepaspectratio]{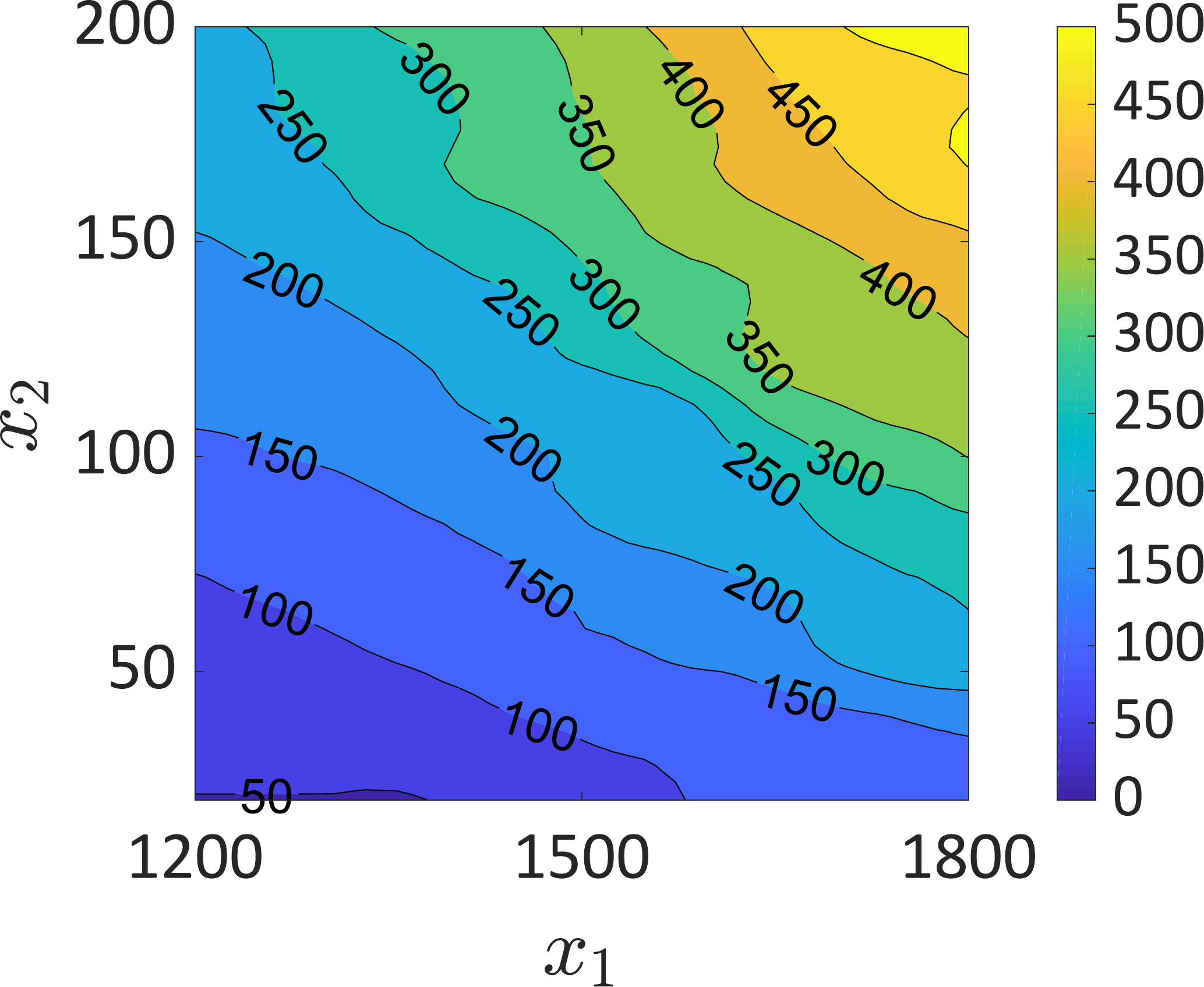}
		\caption{KCDE}
	\end{subfigure}
	\caption{SIR model --- Comparisons of the mean function estimation in the plan  $N=500$.}
	\label{fig:SIRm}
\end{figure}
\begin{figure}[!htbp]
	\centering
	\begin{subfigure}{.33\linewidth}
		\centering
		\includegraphics[height=0.8\linewidth, keepaspectratio]{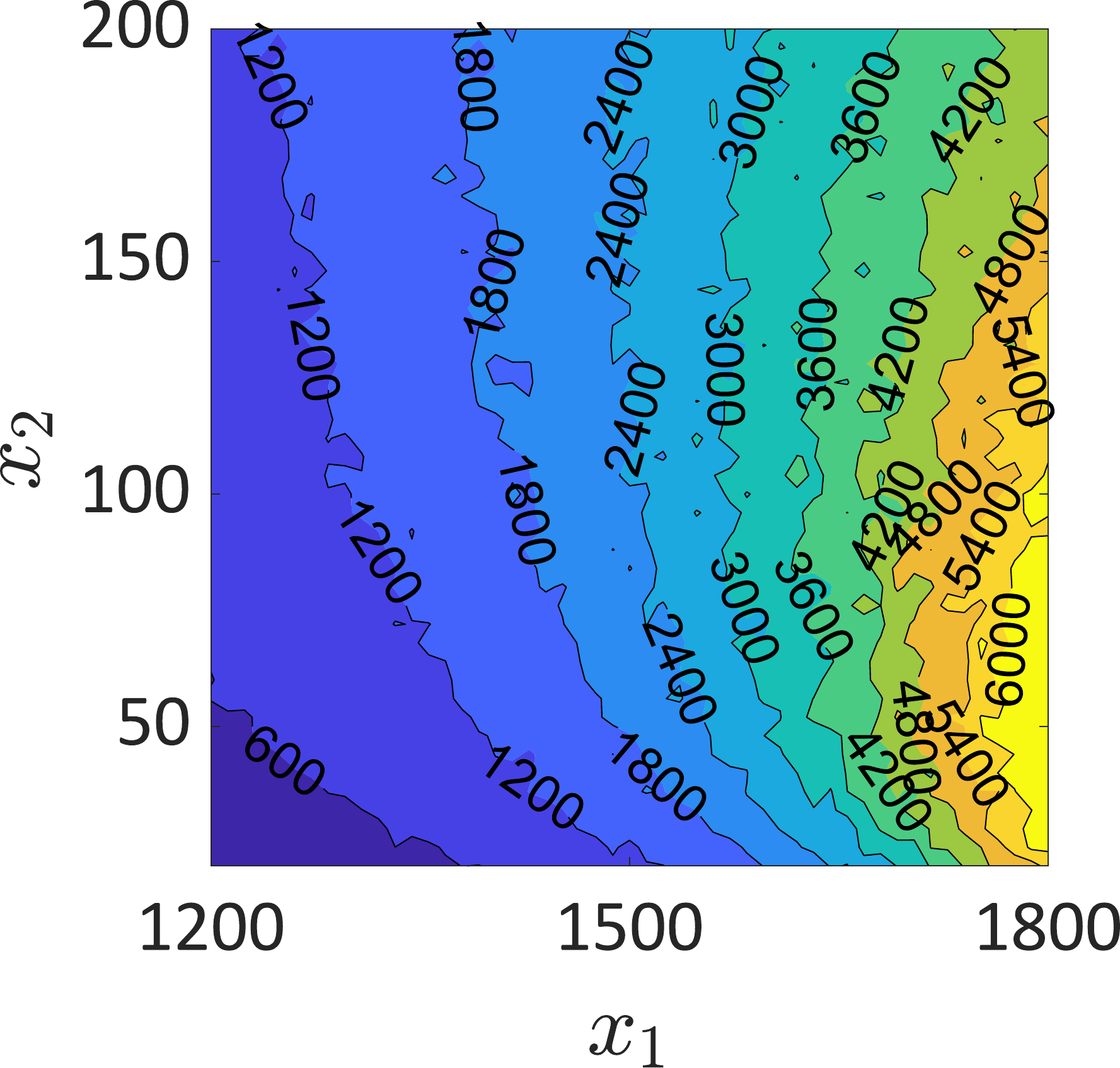}
		\caption{Reference}
	\end{subfigure}
	\hspace*{-6mm}
	\begin{subfigure}{.33\linewidth}
		\centering
		\includegraphics[height=0.8\linewidth, keepaspectratio]{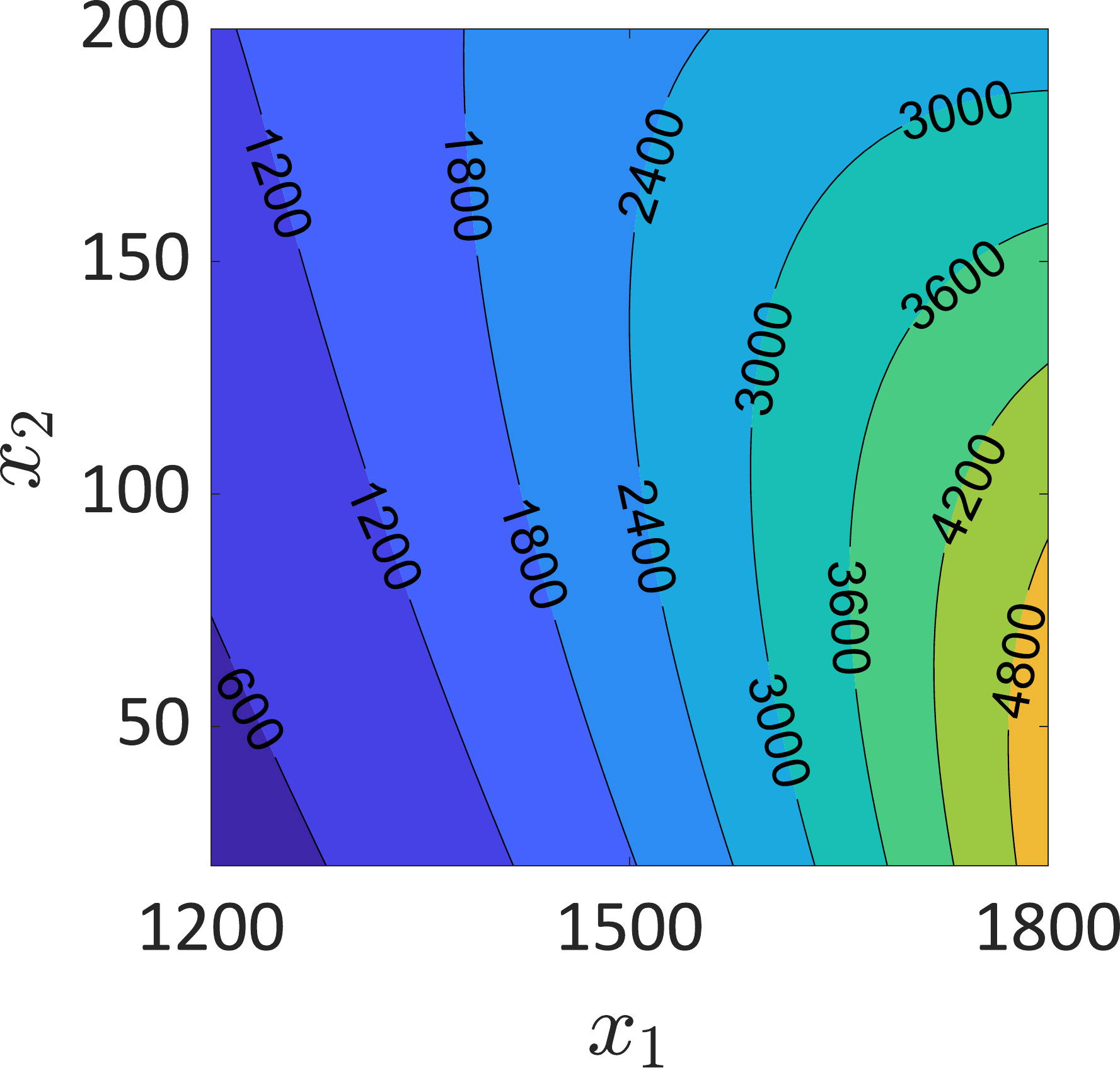}
		\caption{GLaM}
	\end{subfigure}
	\hspace*{-3mm}
	\begin{subfigure}{.33\linewidth}
		\centering
		\includegraphics[height=0.8\linewidth, keepaspectratio]{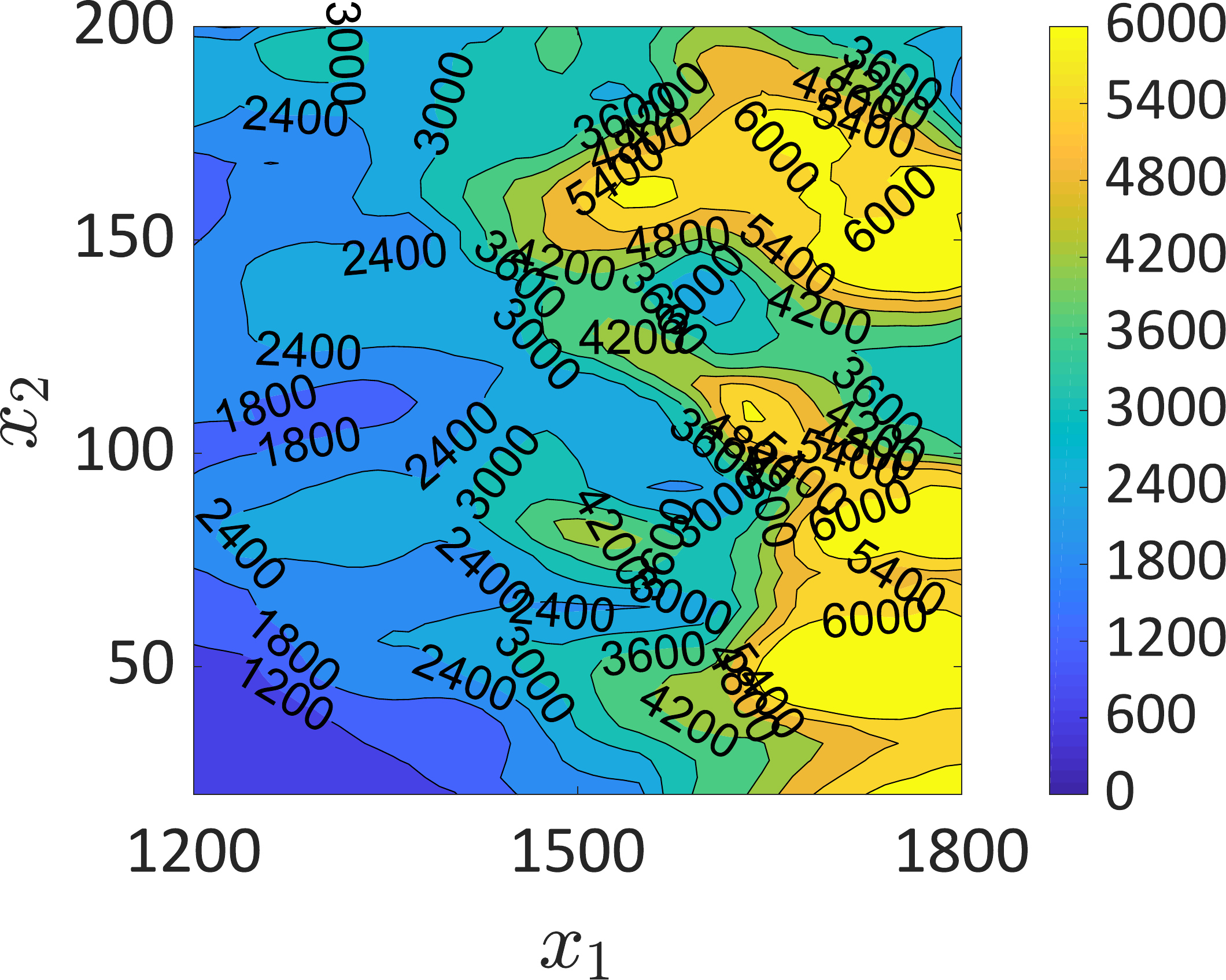}
		\caption{KCDE}
	\end{subfigure}
	\caption{SIR model --- Comparisons of the variance function estimation, $N=500$.}
	\label{fig:SIRv}
\end{figure}

\Cref{fig:SIRpdf} compares two response PDFs estimated by a GLaM and by a
KCDE for two sets of initial configurations, using an experimental design of 
size $N=500$. The reference histograms are 
obtained by $10^4$ repeated model runs for each $\ve{x}$. We observe that the 
PDF shape varies: it changes from symmetric to slightly right-skewed distributions 
depending on the input variables. The GLaM is able to accurately capture this shape variation, while KCDE exhibits relatively poor 
shape representations. 

\Cref{fig:SIRm,fig:SIRv} illustrate the mean and variance function. Because the analytical results are unknown for this simulator, we use $1{,}000$ replications to estimate these quantities for plotting. We observe that both functions vary nonlinearly in the input space. Compared with the KCDE, the GLaM is able to capture the trend of the two functions and provides more accurate estimates. More detailed comparisons of the surrogate models 
are shown in \Cref{fig:SIR_WS}. The error of the oracle Gaussian approximation is quite small. This implies that the response distribution for most of the input parameters in the input space is close to a Gaussian distribution. Nevertheless, GLaMs built on $N=4{,}000$ model runs still demonstrate better average behavior. For all sizes of experimental design, GLaMs clearly outperform KCDEs. For $N\geq 500$, the biggest error of GLaMs is smaller than the smallest error of KCDEs among the 50 repetitions. Finally, to achieve the same accuracy as GLaMs, KCDEs require around 7 times more model runs.
\par
\begin{figure}[!htbp]
	\centering
	\includegraphics[width=.65\linewidth, keepaspectratio]{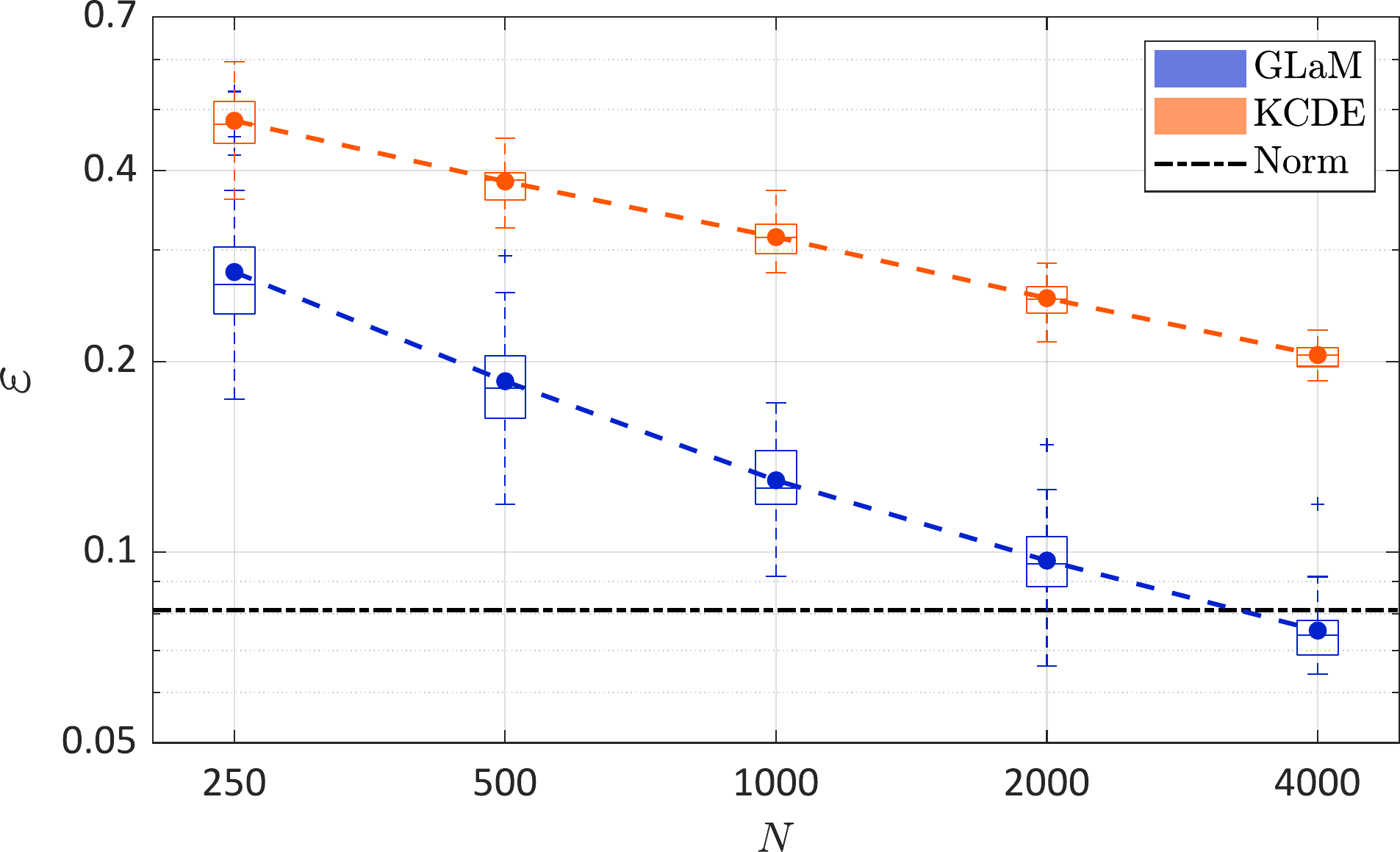}
	\caption{SIR model --- Comparison of the convergence 
		between GLaMs and KCDEs in terms of the normalized Wasserstein distance as 
		a function of the size of the experimental design. The dashed line denotes 
		the average value over $50$ repetitions of the full analysis. The 
		black dash-dotted line represents the error of the model assuming that the 
		response distribution is normal with the true mean and variance}
	\label{fig:SIR_WS}
\end{figure}
In epidemiological management, the expected value $\mu(\ve{x}) = 
\Esp{Y(\ve{x})}$ is crucial for decision making \cite{Merl2009}. Therefore, we 
investigate the accuracy of $\mu(\ve{x})$ estimations, and the results are in 
\Cref{fig:SIR_mean}. First, both GLaM and KCDE can explain more than 
90\% of the variance in $\mu(\ve{X})$ for $N=250$, which implies an overall 
accurate approximation to the mean function. With increasing $N$, GLaM shows a 
more rapid decay of the error. Furthermore, GLaMs built on $N=1{,}000$ have a 
similar (or even slightly better) performance to KCDEs with $N=4{,}000$.
\begin{figure}[!htbp]
	\centering
	\includegraphics[width=.65\linewidth, keepaspectratio]{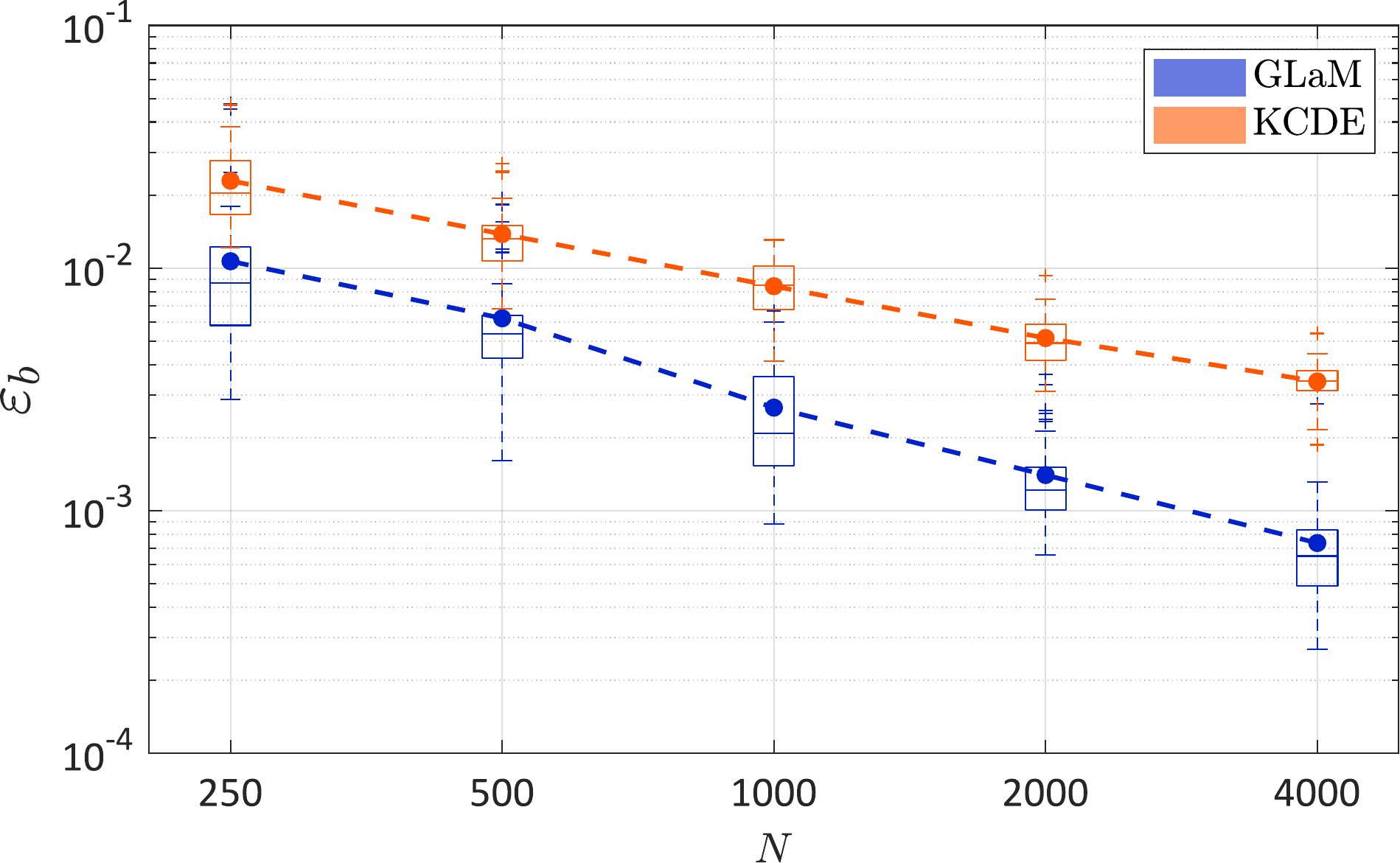}
	\caption{SIR model, mean value estimations --- Comparison of the convergence 
		between GLaMs and KCDEs in terms of the normalized mean-squared error 
		as a function of the size of the experimental design. The dashed line 
		denotes the average value over $50$ repetitions of the full analysis.}
	\label{fig:SIR_mean}
\end{figure}

\section{Conclusions}
\label{sec:conclusions}
This paper presents an efficient and accurate nonintrusive surrogate 
modeling method for stochastic simulators that does not require replicated runs 
of the latter. We follow the setting of Zhu and Sudret \cite{zhuIJUQ2020}, 
where the generalized lambda distribution is used to flexibly approximate 
the response probability density function. The distribution parameters, as 
functions of the input variables, are approximated by polynomial chaos 
expansions. In this paper, however, we do not require replicated runs of the 
stochastic simulator, which provides a more general and versatile approach. We 
propose the maximum conditional likelihood 
estimator to construct such a model for given basis functions. This estimation 
method is shown to be consistent and applicable to data with or without 
replications. In addition, we modify the feasible generalized least-squares 
algorithm to select suitable truncation schemes for the distribution 
parameters, which also provides a good starting point for the subsequent 
optimization of the likelihood function.
\par
The performance of the new method is illustrated on analytical examples 
and case studies in mathematical finance and epidemics. The results show that 
with a reasonable number of model runs, the developed algorithm can produce 
surrogate models that accurately approximate the response probability density 
function and capture the shape variations of the latter with $\ve{x}$. 
Considering the normalized Wasserstein distance as an error metric, generalized 
lambda models always show a better convergence rate than the 
nonparametric kernel conditional density estimator with adaptive bandwidth 
selections (from the package \texttt{np} in R). Furthermore, the proposed 
method generally yields more reliable estimates of certain important quantities.
\par
Quantifying the uncertainty of surrogate models that 
emulate the entire response distribution of a stochastic simulator remains to 
be developed in future work, especially when no or only a few replications are available. One possibility is to use cross-validation to calculate the expected loss. However, 
when the log-likelihood is used as the loss function such as \cref{eq:nloglh}, the 
resulting score is not intuitive and is difficult to interpret. Alternatively, 
with a given basis for $\ve{\lambda}(\ve{x})$ in GLaMs, one can use bootstrap 
\cite{Efron1982} to assess the uncertainty in the estimation of the 
coefficients. \Cref{fig:GBMboot} illustrates the PDF predictions of 100 bootstrapping GLaMs of a data set with $N=500$ of Example 1. Note that the associated theoretical aspects remain to be developed: it is necessary to prove the \emph{bootstrap consistency}, which is usually achieved by showing the asymptotic normality of 
the estimator. As a result, the asymptotic properties of the maximum likelihood 
estimator in \cref{eq:nloglh} need to be further investigated.
\begin{figure}[!htbp]
	\centering
	\begin{subfigure}{.48\linewidth}
		\centering
		\includegraphics[height=0.68\linewidth, keepaspectratio]{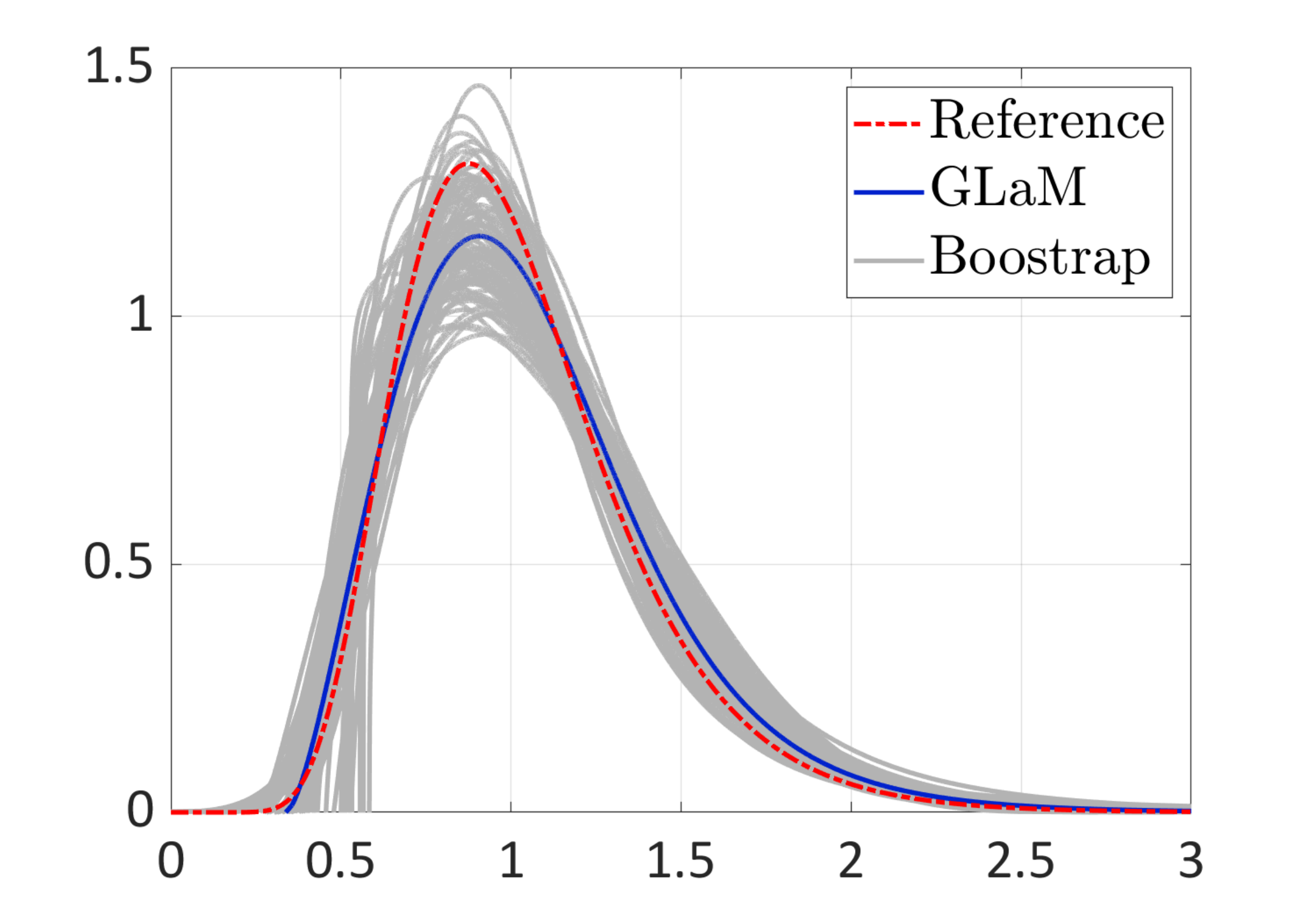}
		\caption{PDF for $\ve{x} = (0.03,0.33)^T$}
	\end{subfigure}
	\begin{subfigure}{.48\linewidth}
		\centering
		\includegraphics[height=0.68\linewidth, keepaspectratio]{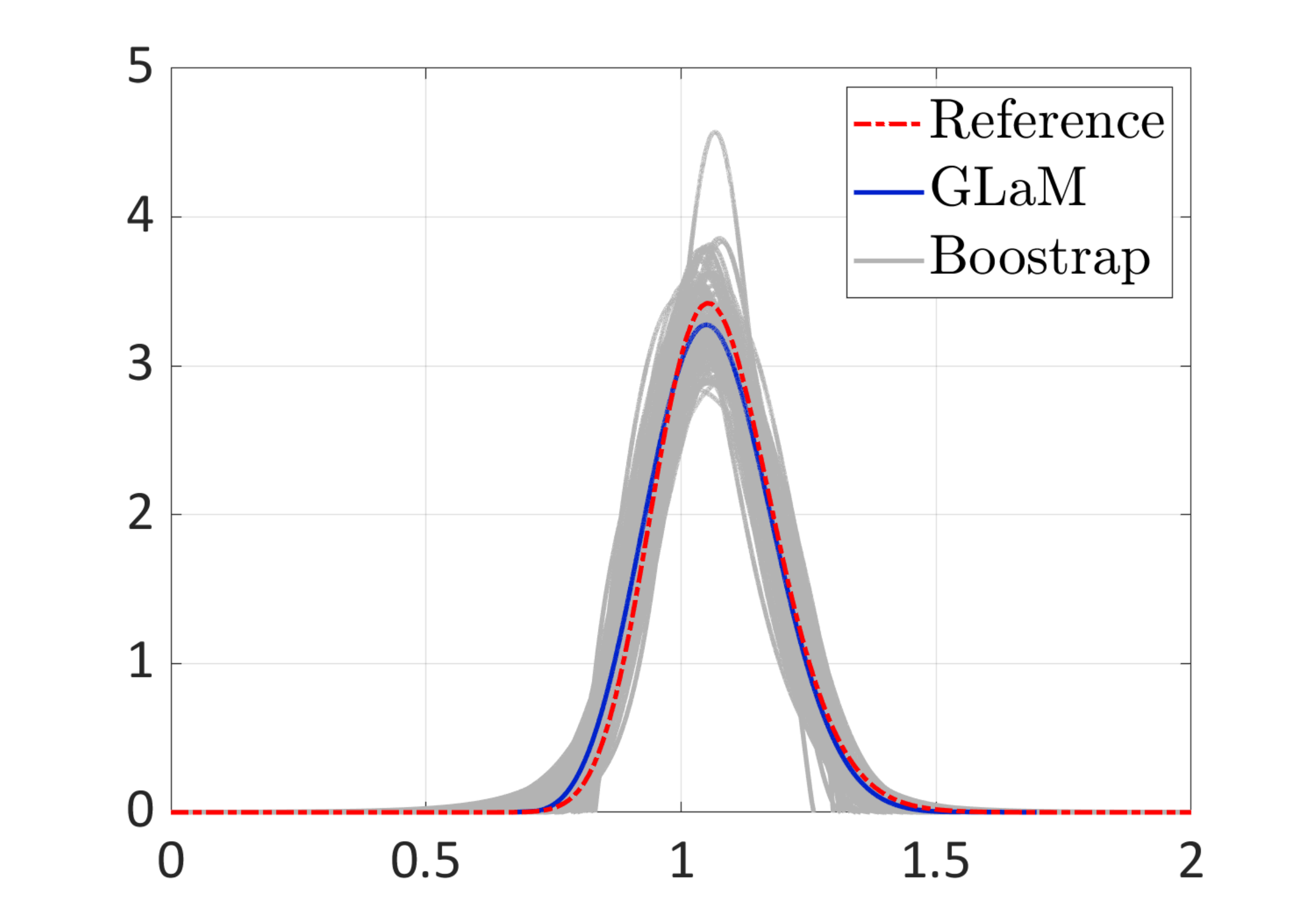}
		\caption{PDF for $\ve{x} = (0.07,0.11)^T$}
	\end{subfigure}
	\caption{Example 1 --- Uncertainty on the PDF predicted by GLaM for two values of the input parameters, using an experimental design of $N=500$. The blue line is the PDF predicted by GLaM from the 500~data points. The grey lines correspond to 100 PDFs generated by GLaM using bootstrapped experimental designs.}
	\label{fig:GBMboot}
\end{figure}

\par
Possible interesting applications of the proposed method to be investigated in 
future studies include reliability analysis and sensitivity analysis 
\cite{ZhuRESS2021}. To improve the performance of the generalized lambda surrogate model 
for small data sets, we plan to develop algorithms that select only important 
basis functions based on appropriate model selection criteria. 
Finally, since the generalized lambda distribution cannot represent 
multimodal distributions, potential extensions to mixtures of generalized 
lambda distributions may provide a more flexible surrogate for simulators with 
multimodal response distribution \cite{Fadikar2018}.

\section*{Acknowledgments}
This paper is a part of the project ``Surrogate Modeling for Stochastic 
Simulators (SAMOS)'' funded by the Swiss National Science Foundation (Grant  
\#200021\_ 175524), whose support is gratefully acknowledged.

\bibliographystyle{siamplain}
\bibliography{references}

\appendix
\section{Appendix}
\subsection{Consistency of the maximum likelihood estimator}
\label{sec:MLEConsistency}
In this section, we prove the consistency of the maximum likelihood estimator, 
as described in \Cref{thm:consistMLE}. For the ease of derivation, we introduce 
the following notation:
\begin{equation*}
	q_{\ve{c}}(\ve{x},y) = 
	f_{Y\mid\ve{X}}\left(y\bigr\rvert\ve{\lambda}^{\PC}(\ve{x};\ve{c})\right), \quad
	p_{\ve{c}}(\ve{x},y) = f_{\ve{X},Y}(\ve{x},y)=f_{X}(\ve{x})q_{\ve{c}}(\ve{x},y),
\end{equation*}
where $q_{\ve{c}}$ denotes the conditional PDF with model parameters $\ve{c}$, 
and $p_{\ve{c}}$ corresponds to the associated joint PDF. Under this setting, 
we assume that the true distribution $q_{0}$ belongs to the family for a 
particular set of coefficients $\ve{c}_0$, \ie $q_{0}=q_{\ve{c}_0}$ and 
$p_{0}=p_{\ve{c}_0}$. We denote the probability measure of the probability 
space of $(\ve{X},Y)$ by $P_0$ and the Lebesgue measure by $\mu$. 
\par
The maximum likelihood estimation defined in \cref{eq:joint} belongs to the 
generalized method of moments (GMM) \cite{Hansen1982} for which we define the 
\emph{loss function} by
\begin{equation}
	\ell_{\ve{c}}(\ve{x},y) =-\log\left(q_{\ve{c}}(\ve{x},y)\right) 
	\mathbbm{1}_{q_0(\ve{x},y)>0}(\ve{x},y).
\end{equation}
It holds that
\begin{equation*}
	\ve{c}_0 = \arg\min_{\ve{c}}l(\ve{c}), \text{ where } 
	l(\ve{c})=\Esp{\ell_{\ve{c}}(\ve{X},Y)}.
\end{equation*}
The maximum likelihood estimator is then defined by
\begin{equation*}
	\hat{\ve{c}} = \arg\min_{\ve{c}} l_n(\ve{c}), \text{ where } l_n(\ve{c}) = \frac{1}{n} \sum_{i=1}^{n}\ell_{\ve{c}}\left(\ve{X}^{(i)},Y^{(i)}\right),
\end{equation*}
where $l_n$ is the empirical version of $l$.
\par
To prove the consistency of a GMM estimator, the \emph{uniform law of large numbers} is usually used. In the case of a maximum likelihood estimator for the generalized lambda 
model, classical methods \cite{Newey1994} to prove the uniform law of large 
numbers cannot be applied directly, due to the fact that the support of 
$q_{\ve{c}}$ can depend on the model parameters $\ve{c}$, as shown in 
\cref{eq:Bounds}. To circumvent this problem, we use the techniques suggested 
by \cite{vandegeer2000} for the proof. 
\par
\begin{lemma}\label{lem:BoundCon}
	Under the conditions described in \Cref{thm:consistMLE}, we have the following:
	\begin{enumerate}[label=(\roman*)]
		\item Boundedness: $\sup_{\ve{c} \in \cc} q_{\ve{c}}(\ve{x},y) < +\infty$.
		\item Continuity: $\forall \, \tilde{\ve{c}} \in \cc$, the map $\ve{c} \mapsto q_{\ve{c}}$ is continuous at $\tilde{\ve{c}}$ for $\mu$-almost all $\left(\ve{x},y\right) \in \cd_{\ve{x}}\times\Rr$.
	\end{enumerate}
\end{lemma}
\begin{proof}
	\emph{(i)} As the conditions of \Cref{thm:consistMLE} indicate that $\cd_{\ve{X}}$ and 	$\cc$ are compact, the two sets are bounded according to the 	\emph{Heine--Borel theorem}. Hence, the value of 	$\ve{\lambda}^{\PC}\left(\ve{x};\ve{c}\right)$ is also bounded. We denote 
	respectively $\acc{\overline{C}_i,i=1,\ldots,4}$ and 
	$\acc{\underline{C}_i,i=1,\ldots,4}$ as the upper and lower bounds for each 
	component of $\ve{\lambda}$:
	\begin{equation}\label{eq:lambound}
		\underline{C}_i\leq \lambda_i \leq \overline{C}_i \quad \forall i=1,\ldots,4.
	\end{equation}
	In addition, \cref{eq:lamPCE_lam2} guarantees that 
	$\lambda_2^{\PC}(\ve{x};\ve{c})$ is bounded away from 0, \ie 
	$\underline{C}_2 >0$. Consider now \cref{eq:FKMLpdf} to evaluate the PDF of 
	GLDs. If $u$ in \cref{eq:FKMLpdf} does not exist in $[0,1]$, $q_{\ve{c}} = 
	0$ and thus bounded. For $u \in [0,1]$, we have
	\begin{equation}\label{eq:ineq}
		\frac{\lambda_2}{u^{\lambda_3-1} + (1-u)^{\lambda_4-1}} \leq \frac{\overline{C}_2}{u^{\overline{k}}+(1-u)^{\overline{k}}} ,
	\end{equation}
	where
	\begin{equation*}\label{eq:min34}
		\overline{k}=\max\acc{\overline{C}_3-1,\overline{C}_4-1}.
	\end{equation*}
	\par
	Define the function $m(u) = u^{\overline{k}}+(1-u)^{\overline{k}}$, which corresponds to the denominator of \cref{eq:ineq}. For $\overline{k}=0$ and $1$, $m(u)$ is a constant function equal to $2$ and $1$, respectively. If $k\neq 0,1$, the derivative $m^\prime(u) = \overline{k}\left(u^{\overline{k}-1}-(1-u)^{\overline{k}-1}\right)$ is equal to $0$ only at $u=0.5$ in $[0,1]$. As a result, $\min m(u) = \min\acc{m(0),m(0.5),m(1)}$. For $\overline{k}<0$, $\min m(u) = m(0.5) = 2^{1-\overline{k}}$. While for $\overline{k}>0$, $\min m(u) = \min\acc{m(0),m(0.5),m(1)} = \min\acc{1,2^{1-\overline{k}}}$. 
	Hence, we have $\min m(u) \geq \min\acc{1,2^{1-\overline{k}}} = C_m$. 
	Taking this property into account, \cref{eq:ineq} becomes
	\begin{equation}\label{ineq:cq}
		\frac{\lambda_2}{u^{\lambda_3-1} + (1-u)^{\lambda_4-1}} \leq 
		\frac{\overline{C}_2}{C_m} = C_q.
	\end{equation}
	Therefore, $\sup_{\ve{c} \in \cc} q_{\ve{c}}(\ve{x},y) \leq C_q$.
	\par	
	\emph{(ii)} Next, we prove the continuity. For any $\tilde{\ve{c}} \in \cc$, we classify the points $(\ve{x},y) \in \cd_{\ve{x}}\times\Rr$ into three groups based on their corresponding latent variable $\tilde{u}$: (1) $\tilde{u} \in (0,1)$, (2) $\tilde{u}$ does not exist within $[0,1]$, and
	(3) $\tilde{u} = 0$ or $1$.
	\par
	For $(\ve{x},y)$ in the first class, $y$ is an interior point of the support 
	of the conditional distribution $q_{\tilde{\ve{c}}}(\ve{x},\cdot)$. 
	Thereby, the following equation holds:
	\begin{equation} \label{eq:nleq}
		y = Q(\tilde{u};\tilde{\ve{\lambda}}) = \tilde{\lambda}_1 + \frac{1}{\tilde{\lambda}_2}\left( \frac{\tilde{u}^{\tilde{\lambda}_3}-1}{\tilde{\lambda}_3} - \frac{(1-\tilde{u})^{\tilde{\lambda}_4}-1}{\tilde{\lambda}_4}\right),
	\end{equation}
	where the distribution parameters $\tilde{\ve{\lambda}}$ are obtained by evaluating $\ve{\lambda}^{\PC}\left(\ve{x};\tilde{\ve{c}}\right)$. The partial derivatives of $Q(u;\ve{\lambda})$ with respect to all the relevant parameters are
	\begin{align}
		\frac{\partial Q}{\partial u} &= \frac{1}{\lambda_2}\left(u^{\lambda_3-1} + (1-u)^{\lambda_4 - 1}\right), \label{eq:dqdu}\\
		\frac{\partial Q}{\partial \lambda_1} &= 1, \label{eq:dqdlam1}\\
		\frac{\partial Q}{\partial \lambda_2} &= - \frac{1}{\lambda^2_2}\left( \frac{u^{\lambda_3}-1}{\lambda_3} - \frac{(1-u)^{\lambda_4}-1}{\lambda_4}\right), \label{eq:dqdlam2}\\
		\frac{\partial Q}{\partial \lambda_3} & = \frac{1}{\lambda_2 \lambda^2_3}\left(u^{\lambda_3}\ln(u)\lambda_3 - \left(u^{\lambda_3}-1\right)\right), \label{eq:dqdlam3}\\
		\frac{\partial Q}{\partial \lambda_4} & = \frac{1}{\lambda_2 \lambda^2_4}\left(\left( (1-u)^{\lambda_4} - 1\right) - (1-u)^{\lambda_4}\ln(1-u)\lambda_4 \right).\label{eq:dqdlam4}
	\end{align}
	It can be easily observed that \cref{eq:dqdu} and \cref{eq:dqdlam1} are 
	continuous functions of $u \in (0,1)$ and $\ve{\lambda}$. Although 
	\cref{eq:dqdlam2} is undefined for $\lambda_3 = 0$ and $\lambda_4 = 0$, the limit exists 
	according to \emph{l'H\^opital's rule}. The same holds for \cref{eq:dqdlam3} and \cref{eq:dqdlam4}. As a result, we 
	can extend \crefrange{eq:dqdlam2}{eq:dqdlam4} by continuity, and thus they 
	become continuous functions of $u \in (0,1)$ and $\ve{\lambda}$. Therefore, 
	$Q(u,\ve{\lambda})$ is continuously differentiable. In addition, 
	\cref{eq:dqdu} is bounded away from 0. These two properties allow one to 
	apply the \emph{implicit function theorem}, and thus $u$ is a continuous 
	function of $\ve{\lambda}$ in a neighborhood of $\tilde{\ve{\lambda}}$, 
	which implies that $u$ is continuous at $\tilde{\ve{\lambda}}$. According 
	to \cref{eq:FKMLpdf}, the PDF is a continuous function of both $u$ and 
	$\ve{\lambda}$. Hence, using the continuity shown before, 
	$f_Y(y;\ve{\lambda})$ is continuous at $\tilde{\ve{\lambda}}$. Furthermore, 
	$\ve{\lambda}^{\PC}(\ve{x};\ve{c})$ are $C^{\infty}$ functions of $\ve{c}$, 
	and thus $\ve{\lambda}^{\PC}(\ve{x};\ve{c})$ is continuous at 
	$\tilde{\ve{c}}$. Combining both the continuity of $f_Y(y;\ve{\lambda})$ 
	and $\ve{\lambda}^{\PC}(\ve{x};\ve{c})$, we have that 
	$q_{\ve{c}}(\ve{x},y)$ is continuous at $\tilde{\ve{c}}$ for the point 
	$(\ve{x},y)$. 
	\par
	Now consider a point $(\ve{x},y)$ in the second class, which implies that $y$ is outside the support of $q_{\tilde{\ve{c}}}(\ve{x},\cdot)$, say, $y$ is smaller than the lower bound of the support of $q_{\tilde{\ve{c}}}(\ve{x},\cdot)$. In this case, $q_{\tilde{\ve{c}}}(\ve{x},y) = 0$. According to \cref{eq:Bounds}, if the lower bound is finite, it is a continuous function of $\ve{\lambda}$ and thus continuous at $\tilde{\ve{c}}$. As a result, for $\ve{c}$ within a certain neighborhood of $\tilde{\ve{c}}$, the lower bound is larger than $y$, which implies $q_{\ve{x}}(\ve{x},y)=0$ for $\ve{c}$ in this neighborhood. Thereby, $q_{\ve{c}}(\ve{x},y)$ is continuous at $\tilde{\ve{c}}$. Analogous reasoning holds for the case where $y$ is bigger than the upper bound of the support.
	\par
	The last class corresponds to the case where $y$ is located on the endpoint 
	of the support of $q_{\tilde{\ve{c}}}(\ve{x},\cdot)$. By taking 
	$\tilde{u}=0$ and $1$ in \cref{eq:nleq} or considering directly \cref{eq:Bounds}, we obtain two associated deterministic functions between $\ve{x}$ and $y$. As a result, points of the third class can be represented by two curves in $\cd_x \times \Rr$, whose Lebesgue measure is zero. This closes the proof of continuity.	
\end{proof}

\begin{lemma}\label{lem:ULLN}
	The class $\cg$ defined below satisfies the uniform strong law of large numbers:
	\begin{equation}
		\cg = \acc{g_{\ve{c}} = \log\left(\frac{q_{\ve{c}}+q_0}{2q_0}\right)\mathbbm{1}_{q_0>0}: \ve{c} \in \cc}.
	\end{equation}
\end{lemma}
\begin{proof}
	According to the continuity property in \Cref{lem:BoundCon}, it is obvious 
	that for all $\tilde{\ve{c}} \in \cc$, the map $\ve{c} \mapsto g_{\ve{c}}$ 
	is continuous at $\tilde{\ve{c}}$ for $\mu$-almost all $(\ve{x},y) \in 
	\cd\times\Rr$. By assumption, the probability measure $P_0$ is absolutely continuous with respect to $\mu$, and thus $g_{\ve{c}}$ is continuous for $P_0$-almost all $(\ve{x},y) \in \cd\times\Rr$. 
	\par
	Define $G$ as the envelope function of the class $\cg$, \ie $G(\ve{x},y) = 
	\sup_{\ve{c} \in \cc}\abs{g_{\ve{c}}(\ve{x},y)}$. Let us prove that $G \in 
	L_1(P_0)$, where $L_1(P_0)$ denotes the set of absolutely integrable 
	functions with respect to $P_0$.
	\par	
	Taking the boundedness property in \Cref{lem:BoundCon} into account, we obtain
	\begin{equation}
		g_{\ve{c}}(\ve{x},y) \leq\log\left(\frac{2C_q}{q_0(\ve{x},y)}\right) = 
		\log(2C_q) - \log(q_0(\ve{x},y)).
	\end{equation}
	Obviously, $g_{\ve{c}}(\ve{x},y)\geq -\log(2)$. Therefore,
	\begin{equation}
		\begin{split}
			\abs{g_{\ve{c}}(\ve{x},y)} &\leq \max\acc{\log(2),\abs{\log(2C_q)} + 
				\abs{\log(q_0(\ve{x},y))}} \\ 
			&\leq \log(2)+\abs{\log(C_q)} + \abs{\log(q_0(\ve{x},y))}.
		\end{split}
	\end{equation}
	Because the inequality is independent of $\ve{c}$, we have 
	\begin{equation}\label{eq:ineqgc}
		\begin{split}
			G(\ve{x},y)&\leq \log(2)+\abs{\log(C_q)} + \abs{\log(q_0(\ve{x},y))}, \\
			\Esp{G(\ve{X},Y)} & \leq \log(2)+\abs{\log(C_q)} + 
			\Esp{\abs{\log(q_0(\ve{X},Y)}}.
		\end{split}
	\end{equation}
	Now consider the last term in \cref{eq:ineqgc}:
	\begin{equation}\label{eq:envelop}
		\begin{split}
			\Esp{\abs{\log(q_0(\ve{X},Y)}}&=\int_{\cd_{\ve{x}}\times\Rr} 
			\abs{\log\left(q_0(\ve{x},y)\right)}p_0(\ve{x},y) \D\ve{x}\D y\\
			& = \int_{\cd_{\ve{x}}}\left(\int_{\Rr}\abs{\log\left(q_0(\ve{x},y)\right)}q_0(\ve{x},y)\D y\right)f_{\ve{X}}(\ve{x})\D\ve{x}. 
		\end{split}
	\end{equation}
	Through a change of variables, the integral within the parenthesis of \cref{eq:envelop} can be calculated as
	\begin{equation}
		B(\ve{x}) = \int_{\Rr}\abs{\log\left(q_0(\ve{x},y)\right)}q_0(\ve{x},y)\D y 
		= \int_{0}^{1} \abs{\log\left(\frac{\lambda_2}{u^{\lambda_3-1} + 
				(1-u)^{\lambda_4-1}}\right)}\D u,
	\end{equation}
	where $\ve{\lambda} = \ve{\lambda}^{\PC}(\ve{x};\ve{c}_0)$. 
	According to \cref{eq:lambound}, we have
	\begin{equation}
		\begin{split}
			B(\ve{x}) &\leq \int_{0}^{1} \abs{\log(\lambda_2)} + \abs{\log\left(u^{\lambda_3-1} + (1-u)^{\lambda_4-1}\right)} \D u \\
			&\leq k_2 + \int_{0}^{1}\max\acc{ \abs{\log\left(u^{\underline{k}} + 
					(1-u)^{\underline{k}}\right)}, \abs{\log\left(u^{\overline{k}} + 
					(1-u)^{\overline{k}}\right)}} \D u ,
		\end{split}
	\end{equation}
	where 
	\begin{equation*}	
		k_2 = \max\acc{\abs{\log\left(\overline{C}_2\right)} , 
			\abs{\log\left(\underline{C}_2\right)}}, \;
		\underline{k} = \min\acc{\underline{C}_3-1,\underline{C}_4-1},\;
		\overline{k} = \max\acc{\overline{C}_3-1,\overline{C}_4-1} .
	\end{equation*}
	Using the symmetry of the integrand, we get 
	\begin{equation}\label{eq:B(x)}
		\begin{split}
			B(\ve{x}) &\leq k_2 + 2\cdot\max\acc{\int_{0}^{\frac{1}{2}} \abs{\log\left(u^{\underline{k}} + (1-u)^{\underline{k}}\right)}\D u, \int_{0}^{\frac{1}{2}} \abs{\log\left(u^{\overline{k}} + (1-u)^{\overline{k}}\right)}\D u} \\
			&\leq k_2 + 2\cdot\left(\int_{0}^{\frac{1}{2}} \abs{\log\left(u^{\underline{k}} + (1-u)^{\underline{k}}\right)}\D u + \int_{0}^{\frac{1}{2}} \abs{\log\left(u^{\overline{k}} + (1-u)^{\overline{k}}\right)}\D u\right).
		\end{split}
	\end{equation}		
	Without loss of generality, we now study the property of the integral
	\begin{equation}\label{eq:abslh}
		\int_{0}^{\frac{1}{2}}\abs{\log\left(u^{k} + (1-u)^{k}\right)}\D u .
	\end{equation}
	For $k=0$, \cref{eq:abslh} is equal to $\frac{1}{2}\log(2)$. For $k>0$, we have $u^k\leq (1-u)^k$, and thus
	\begin{equation}
		\begin{split}
			\int_{0}^{\frac{1}{2}}\abs{\log\left(u^{k} + (1-u)^{k}\right)}\D u \leq \int_{0}^{\frac{1}{2}}\abs{\log\left(2 (1-u)^{k}\right)}\D u 
			& \leq \frac{1}{2}\log(2) - \int_{0}^{\frac{1}{2}} k\log(1-u) \D u\\
			& = \frac{1}{2}\log(2) + \frac{k}{2}\left( 1- \log(2)\right).
		\end{split}
	\end{equation}
	Through similar calculation, for $k<0$, we have
	\begin{equation}
		\begin{split}
			\int_{0}^{\frac{1}{2}}\abs{\log\left(u^{k} + (1-u)^{k}\right)}\D u \leq \int_{0}^{\frac{1}{2}}\abs{\log\left(2 u^{k}\right)}\D u 
			& \leq \frac{1}{2}\log(2) + \int_{0}^{\frac{1}{2}} k\log(u) \D u\\
			& = \frac{1}{2}\log(2) + \frac{-k}{2}\left( \log(2) + 1\right).
		\end{split}
	\end{equation}
	As a result, \cref{eq:abslh} is finite. More precisely, 	\begin{equation}\label{eq:ineqabslh}
		\int_{0}^{\frac{1}{2}}\abs{\log\left(u^{k} + (1-u)^{k}\right)}\D u \leq \frac{1}{2}\log(2) + \frac{\abs{k}}{2}\left( \log(2) + 1\right).
	\end{equation}	
	\Cref{eq:ineqabslh} implies
	\begin{equation}\label{eq:ineqBx}
		B(\ve{x})\leq k_2 + \log(2) + 
		\left(\abs{\underline{k}}+\abs{\overline{k}}\right)\left( \log(2) + 
		1\right) = C_B.
	\end{equation}
	By inserting \cref{eq:ineqBx} into \cref{eq:envelop}, we obtain 
	\begin{equation}
		\Esp{\abs{\log(q_0(\ve{X},Y)}} \leq C_B.
	\end{equation}
	Then, according to \cref{eq:ineqgc}, the envelope function $G$ fulfills
	\begin{equation}\label{eq:ineqgcpf}
		\begin{split}
			\Esp{G(\ve{X},Y)} &\leq \log(2)+\abs{\log(C_q)} + 
			\Esp{\abs{\log\left(q_0(\ve{X},Y)\right)}} \\
			& = \log(2) + \abs{\log(C_q)} + C_B < +\infty.
		\end{split}
	\end{equation}
	Since $G$ is always positive according to its definition, \cref{eq:ineqgcpf} means $G \in L_1(P_0)$. The continuity and the property of the envelope function $G$ shown above allow applying \cite[Lemma 3.10]{vandegeer2000}, which guarantees that $\cg$ satisfies the uniform weak law of large numbers:
	\begin{equation}
		\sup_{\ve{c} \in \cc} \left(\frac{1}{n}\sum_{i=1}^{n} g_{\ve{c}}\left(\ve{X}^{(i)},Y^{(i)}\right) - \Esp{g_{\ve{c}}\left(\ve{X},Y\right)}\right)\xrightarrow[n\rightarrow +\infty]{P} 0.
	\end{equation}
	Finally, \cite[Theorem 22]{Talagrand1987} extends the convergence to \emph{almost surely}, which is the uniform strong law of large numbers. 
\end{proof} 

Now, we have all the ingredients to prove \Cref{thm:consistMLE}.
\begin{proof}
	Following \cite[Lemma 4.1, 4.2]{vandegeer2000}, it can be easily shown that 
	\begin{equation}
		0 \leq  \int_{\cd_{\ve{x}}} h^2\left(q_{\hat{\ve{c}}},q_{0} \mid \ve{x}\right)f_{\ve{X}}(\ve{x})\D \ve{x} \leq 8\left(\sum_{i=1}^{N}g_{\hat{\ve{c}}}\left(\ve{X}^{(i)},Y^{(i)}\right) - \Esp{g_{\hat{\ve{c}}}\left(\ve{X},Y\right)}\right),\label{ineq:hellinger}
	\end{equation}
	where the Hellinger distance is given by
	\begin{equation*}
		h^2\left(q_{\hat{\ve{c}}},q_{0} \mid \ve{x}\right) = \frac{1}{2}\int_{\Rr}\left(\sqrt{q_{\hat{\ve{c}}}(\ve{x},y)} - \sqrt{q_{0}(\ve{x},y)}\right)^2\D y. \nonumber
	\end{equation*}
	According to \Cref{lem:ULLN}, \cref{ineq:hellinger} implies
	\begin{equation}\label{ineq:Hellingerpf}
		\int_{\cd_{\ve{x}}} h^2\left(q_{\hat{\ve{c}}},q_{0} \mid \ve{x}\right)f_{\ve{X}}(\ve{x})\D \ve{x} \xrightarrow{\text{a.s.}} 0,
	\end{equation}
	which is called the \emph{Hellinger consistency}.
	\par
	We define the function \begin{equation}
		R(\ve{c}) = \int_{\cd_{\ve{x}}} h^2\left(q_{\ve{c}},q_{0} \mid 
		\ve{x}\right)f_{\ve{X}}(\ve{x})\D \ve{x}.
	\end{equation}
	According to \Cref{lem:BoundCon}, $\forall \tilde{\ve{c}} \in \cc$, the map 
	$\ve{c} \mapsto \left(\sqrt{q_{\ve{c}}} - \sqrt{q_0}\right)^2$ is 
	continuous at $\tilde{\ve{c}}$ for all $\ve{x} \in \cd_x$ and almost 
	all $y \in \Rr$. Since $\left(\sqrt{q_{\ve{c}}} - \sqrt{q_0}\right)^2 \leq 
	q_{\ve{c}}+q_{0}$, and $\int_{\Rr}\left(q_{\ve{c}}+q_{0} \right)\D y = 2 < 
	+\infty$, the map $\ve{c} \mapsto h^2\left(q_{\ve{c}},q_{0} \mid 
	\ve{x}\right)$ is continuous for all $\ve{x} \in \cd_{\ve{x}}$, which is 
	guaranteed by the \emph{generalized Lebesgue dominated convergence 
		theorem}. Similarly, the map $\ve{c} \mapsto R(\ve{c})$ is also continuous. 
	\par
	Without going into lengthy discussions, it can be shown that the 
	GLD is \emph{not identifiable} only for 
	$\lambda_3 = \lambda_4 = 1$ and $\lambda_3 = \lambda_4 = 2$. In other 
	words, by excluding two points in the $\lambda_3-\lambda_4$ plane, different 
	values of $\ve{\lambda}$ lead to different distributions. Note that the two 
	exceptions are the only two cases where the corresponding distributions are 
	uniform distributions. As a result, the last condition in 
	\Cref{thm:consistMLE} excludes the nonidentifiable cases. Furthermore, 
	$\ve{\lambda}^{\PC}(\ve{x};\ve{c})$ are polynomials in $\ve{x}$ and linear 
	in $\ve{c}$. Therefore, for $\ve{c} \neq \tilde{\ve{c}}$, 
	$\ve{\lambda}^{\PC}(\ve{x};\ve{c})$ and $ 
	\ve{\lambda}^{\PC}\left(\ve{x};\tilde{\ve{c}}\right)$ are not identical for 
	$\mu$-almost all $\ve{x} \in \Rr^M$, and thus for $P_{\ve{X}}$-almost all 
	$\ve{x} \in \cd_{\ve{X}}$. Hence, there exists a set $\Omega_{\ve{x}}$ with 
	$P_{\ve{X}}(\Omega_{\ve{x}}) >0 $ such that as long as $\ve{c} \neq 
	\ve{c}_0$, $h\left(q_{\ve{c}},q_{0}\mid \ve{x}\right)>0$ $\forall \ve{x} 
	\in \Omega_{\ve{x}}$, which implies the uniqueness. Finally, combining 
	\cref{ineq:Hellingerpf} with the continuity and uniqueness of $R(\ve{c})$, 
	we have $\hat{\ve{c}}\xrightarrow{\text{a.s.}}\ve{c}_0$.
\end{proof} 


%
\end{document}